\documentclass[12pt]{article}
\usepackage{geometry, booktabs, subcaption}
\usepackage{titlesec}
\usepackage[T1]{fontenc}
\usepackage{authblk}
\usepackage{amssymb,amsmath, amsthm}
\usepackage{graphicx}
\usepackage{graphicx}
\usepackage{amsmath}
\usepackage{placeins}
\usepackage{amsfonts}
\usepackage{amssymb}
\usepackage{natbib}
\usepackage{cleveref}
\usepackage{mathrsfs}
\usepackage{bbm}
\usepackage{dsfont}
\usepackage{color}
\usepackage{graphicx}
\allowdisplaybreaks
\usepackage{bm}
\usepackage{algorithm}
\usepackage{algpseudocode}
\usepackage{upgreek}
\usepackage{booktabs}
\usepackage{hhline}
\usepackage{array}
\usepackage{url}

\def\E{\mathds{E}}

\def\R{\mathds{R}}

\def\indicator{I}

\def\simiid{\stackrel{\mbox{\scriptsize{iid}}}{\sim}}
\newcommand{\ddr}{\mathrm{d}}
\newcommand{\prob}{\mathrm{Pr}}
\allowdisplaybreaks

\newtheorem{theorem}{Theorem}

\newtheorem{proposition}[theorem]{Proposition}
\newtheorem{lemma}{Lemma}

\def\eppf{\mathrm{eppf}}

\newcommand{\blind}{1}

\usepackage{xcolor}
\colorlet{red}{black}

\allowdisplaybreaks[3]

\title{A smoothed-Bayesian approach to frequency recovery from sketched data}

\author[1]{Mario Beraha}
\affil[1]{Department of Economics, Management and Statistics, University of Milano-Bicocca}
\author[2]{Stefano Favaro}
\affil[2]{Department of Economics and Statistics, University of Torino and Collegio Carlo Alberto}
\author[3]{Matteo Sesia}
\affil[3]{Department of Data Sciences and Operations, University of Southern California}
\begin{document}

\maketitle

\begin{abstract}
We provide a novel statistical perspective on a classical problem at the intersection of computer science and information theory: recovering the empirical frequency of a symbol in a large discrete dataset using only a compressed representation, or sketch, obtained via random hashing. Departing from traditional algorithmic approaches, recent works have proposed Bayesian nonparametric (BNP) methods that can provide more informative frequency estimates by leveraging modeling assumptions about the distribution of the sketched data. In this paper, we propose an alternative {\em smoothed-Bayesian} approach, inspired by existing BNP methods but designed to overcome their computational limitations when dealing with large-scale data from realistic distributions, including those with power-law tail behaviors.  For sketches obtained with a single hash function, our approach is supported by
precise theoretical guarantees, including unbiasedness and optimality under a squared error loss function within an intuitive class of linear estimators. For sketches with multiple hash functions, we introduce an approach based on \emph{multi-view} learning to construct computationally efficient frequency estimators. We validate our method on synthetic and real data, comparing its performance to that of existing alternatives.
\end{abstract}

% \boxedtext{
% \begin{itemize}
% \item Key boxed text here.
% \item Key boxed text here.
% \item Key boxed text here.
% \end{itemize}}

%%%%%%%%%%%%%%%%%%%%%%%%%%%%%%%%
%%%%%%%%%%%%%%%%%%%%%%%%%%%%%%%%
%%%%%%%%%%%%%%%%%%%%%%%%%%%%%%%%
%%%%%%%%%%%%%%%%%%%%%%%%%%%%%%%%

% !TeX root = main_jasa.tex

\section{Introduction} \label{sec:intro}

\subsection{Background and motivation} \label{sec:background}

An interesting statistical problem born at the intersection of computer science and information theory is to recover the empirical frequency of an object in a large discrete dataset using a lossy compressed representation, or ``sketch". Sketches are central to many data science applications involving memory or privacy constraints, including real-time analysis, fast query processing, and scalable machine learning \citep{Cor(20)}. The count-min sketch (CMS) by \citet{Cor(05)}, reviewed in Appendix~\ref{app-cms}, is a popular algorithm that uses random hash functions to create a data sketch, providing a deterministic upper bound for any object's frequency. Additionally, it allows computing confidence intervals for an object's frequency by leveraging concentration inequalities that utilize the randomness in the hash functions while treating the data as fixed.

 A limitation of the CMS is that, as an algorithmic approach that treats the data as fixed, it may not lead to the most informative estimates when the data are random samples from a population \citep{ting2018count}. This issue has motivated the development of {\em learning-augmented} versions of the CMS, which apply optimization and machine learning techniques to  improve sketching algorithms, as well as frequentist and Bayesian statistical approaches to extract more informative estimates from sketches, leveraging modeling assumptions about the data distribution \citep{Cai(18), ting2018count,Hsu(19),Aam(19),bertsimas2021frequency, Ses(22), aamand2024improved, cao2024learning}.

Within the Bayesian framework, \citet{Cai(18)} pioneered a Bayesian nonparametric (BNP) approach to frequency recovery, assuming a Dirichlet process prior \citep{Fer(73)} for the data distribution and computing the posterior distribution of a new object's empirical frequency conditional on the sketch. Subsequently, \citet{Dol(21), Dol(23)} and \citet{Ber(23)}  extended this approach to more general prior distributions. While the BNP approach is effective under the Dirichlet process prior, it has two limitations. Firstly, it becomes very computationally expensive with more general prior distributions needed to describe important patterns often found in real data, such as power-law tails \citep{ferrer2001two,zipf2016human}. Secondly, it is challenging to apply to sketches obtained from multiple hash functions \citep{Cai(18), Dol(23)}.
See Appendix~\ref{relatedwork-bnp} for a more detailed review of the BNP approach and its computational challenges.

The current limitations of the BNP approach motivate us to develop a novel statistical method for frequency recovery from sketched data, by a single or multiple hash functions. Inspired by the BNP approach but integrating frequentist ideas, our method provides greater flexibility in modeling realistic data while enabling efficient large-scale applications.
Additionally, it enjoys desirable 
theoretical guarantees of unbiasedness and optimality.

\subsection{Problem statement}

To describe the frequency recovery problem, consider a data set $(x_{1},\ldots,x_{n})$, for $n \geq 1$, with the $x_{i}$'s taking values in a (possibly infinite) alphabet $\mathbb{S}$ of symbols.
Only a sketch of these data can be observed, obtained through {\em random hashing} \citep[Chapter 15]{Mit(17)}. For simplicity, we begin by focusing on a single hash function.

A hash function $h:\mathbb{S}\rightarrow [J] := \{1,\ldots,J\}$ with width $J \geq 1$ maps each symbol into one of $J$ {\em buckets}.
This function is assumed to be random and distributed as a  pairwise independent hash family $\mathscr{H}_{J}$. That is, $h: \mathbb{S} \to [J]$ such that $\text{Pr}[h(x_{1})=j_{1},\,h(x_{2})=j_{2}]=J^{-2}$ for any $j_1,j_2 \in [J]$ and fixed $x_1,x_2 \in \mathbb{S}$ such that $x_1 \neq x_2$. The pairwise independence of $\mathscr{H}_{J}$,  known as strong universality, implies uniformity, meaning that $\Pr[h(x)=j]=J^{-1}$ for all $j \in [J]$. Strong universality is a common and mathematically convenient assumption, although different settings may be also considered \citep{Chu(13)}. 

Hashing the data through $h$ produces a random vector $\mathbf{C}_{J}\in\mathbb{N}_{0}^{J}$, referred to as the sketch, whose $j$-th element $C_{j}$ is the number of $x_{i}$'s mapped into the $j$-th bucket:
\begin{displaymath}
C_{j} = \sum_{i=1}^{n} I[h(x_i) = j], \qquad \text{for all } j \in [J],
\end{displaymath}
where $I[\cdot]$ is the indicator function. Therefore, $\sum_{1\leq j\leq J} C_{j} = n$. 
By setting $J$ to be (much) smaller than the anticipated number of distinct symbols in $(x_{1},\ldots,x_{n})$, the sketch $\mathbf{C}_{J}$ is intended to have a (much) smaller memory footprint compared to the original sample.

The frequency recovery problem consists of estimating the number of occurrences of a symbol $x_{n+1}$ (the {\em query}) in the sample, i.e.,
\begin{displaymath}
f_{x_{n+1}} = \sum_{i=1}^{n}I[x_{i}=x_{n+1}],
\end{displaymath}
using only the information in $\mathbf{C}_{J}$, without looking at $(x_{1},\ldots,x_{n})$.
This task is challenging because distinct symbols may be mapped into the same bucket, an event known as a {\em hash collision}. 
When the sample size $n$ and the cardinality of $\mathbb{S}$ are both larger than $J$, as is typical in practice, hash collisions become numerous, making exact recovery of $f_{x_{n+1}}$ impossible. Fortunately, it is possible to obtain useful estimates of $f_{x_{n+1}}$ from $\mathbf{C}_{J}$, especially if we assume the data are a random sample from some discrete distribution $P$ on $\mathbb{S}$.

\subsection{Related works} \label{sec:related}

BNP approaches to the frequency recovery problem were pioneered by \citet{Cai(18)} and extended by \citet{Dol(21), Dol(23)} and \citet{Ber(23)}. These works noted that BNP approaches become computationally expensive beyond the Dirichlet process prior \citep{Dol(23), Ber(23)}, which is often too rigid to describe real data. Moreover, BNP approaches are challenging to extend to sketches obtained from multiple hash functions \citep{Cai(18), Dol(23)}, further limiting their applicability.

This paper draws inspiration from the aforementioned works while aiming to overcome their limitations by developing a practical and scalable {\em smoothed-Bayesian} method. Our approach addresses the computational challenges of existing BNP methods and provides greater flexibility to model realistic data. Additionally, this paper complements the recent works of \citet{Ses(22)} and \citet{Ses(23)}, who proposed a {\em distribution-free} frequentist approach based on conformal inference \citep{vovk2005algorithmic}. Their focus is on providing confidence intervals with finite-sample coverage guarantees for frequency estimates computed by any {\em black-box} model or algorithm. Combining their approach with ours can endow our smoothed-Bayesian method with frequentist uncertainty estimates.

\subsection{Main contributions}

Our first contribution is to introduce a class of smoothed-Bayesian estimators of the empirical frequency $f_{X_{n+1}}$. The terminology ``smoothed" stems from the seminal work of \citet{Goo(53)}, which studied the problem of estimating population frequencies by first producing an {\em oracle} estimator based on knowledge of the data-generating distribution, namely $P$, and then specifying a parametric family for that unknown (discrete) distribution.

In a similar spirit, we first consider the frequency recovery problem conditional on the distribution $P$, deriving an oracle estimator for $f_{X_{n+1}}$ denoted by $\varepsilon_f(P)$.
To relax the assumption that $P$ is known, we smooth this estimator by computing the expected value of $\varepsilon_f(P)$ with respect to a distribution $P \sim \mathscr P$, which can be interpreted as a prior distribution. We consider a broad class of normalized random measures \citep{Jam(02),Pru(02),Reg(03)}, for which the resulting smoothed estimator can be expressed as the expected value of a mixture of Binomial distributions. Concretely, we focus on $\mathscr{P}$ corresponding to the law of the Dirichlet process and the normalized generalized Gamma process  \citep{Jam(02),Pit(03),Lij(07)}---a flexible solution that allows the tail behaviour of $P$ to range from exponential tails to heavy power-law tails.

Our smoothed-Bayesian approach yields computationally efficient estimators that depend linearly on $C_{h(X_{n+1})}$. This is a key advantage over the BNP estimators reviewed in Section~\ref{sec:difficult_bayes}, which are typically intractable.
Our approach also enjoys desirable theoretical properties, including {\em unbiasedness} and {\em optimality} among linear estimators. Specifically, it minimizes mean squared error relative to the Bayes estimator, which, while optimal under a Bayesian framework, is generally intractable.
Furthermore, the smoothed-Bayesian and BNP estimators coincide when $\mathscr{P}$ is the law of the Dirichlet process, ensuring consistency with prior work in that computationally feasible special case. Crucially, however, our estimator does not require posterior evaluation, making it scalable beyond this setting.

Our second contribution extends the smoothed-Bayesian approach to deal with data sketched by multiple hash functions, in analogy with the CMS. This is a challenging problem, for which we propose a solution based on \emph{multi-view learning} \citep{Xu(13), Sha(18),Li(18b)}, leveraging by the fact that the smoothed-Bayesian framework leads not only to a principled {\em point estimate} but also to a {\em probability distribution} for $f_{X_{n+1}}$.
As we shall see, our multi-view learning method can recover the BNP approach of \citet{Cai(18)} in the special case of $P$ being a Dirichlet process, but it is more generally applicable.

\subsection{Organization of the paper}

Section~\ref{sec:setup-assumptions} introduces our framework and derives a closed-form expression for $\varepsilon_{f}(P)$ as a function of $P$. Section~\ref{sec:np-estim} presents a minimax analysis for the estimation of $f_{X_{n+1}}$, underscoring the need for prior (or smoothing) assumptions. 
Section~\ref{sec:difficult_bayes} reviews the difficulties associated with the BNP approach.
Section~\ref{sec:smoothed} presents our smoothed-Bayesian approach for a single hash function, and Section~\ref{sec:multi-view} extends it to multiple hash functions. Section~\ref{sec:numerics} validates our method empirically, comparing it to existing algorithmic and BNP solutions.
Section~\ref{sec:discussion} concludes with a discussion. The Appendices in the Supplementary Material contain proofs, additional methodological details, further comparisons, and numerical results.

%%%%%%%%%%%%%%%%%%%%%%%%%%%%%%%%
%%%%%%%%%%%%%%%%%%%%%%%%%%%%%%%%
%%%%%%%%%%%%%%%%%%%%%%%%%%%%%%%%
%%%%%%%%%%%%%%%%%%%%%%%%%%%%%%%%

\section{Preliminary results} \label{sec:setup}

\subsection{Statistical framework and problem setup} \label{sec:setup-assumptions}

We rely on the following assumptions. Firstly, $(x_{1},\ldots,x_{n})$ are modeled as a random sample $\mathbf{X}_{n}=(X_{1},\ldots,X_{n})$ from an unknown discrete distribution $P=\sum_{s\in\mathbb{S}}p_{s}\delta_{s}$ on $\mathbb{S}$, where $p_{s}$ is the (unknown) probability of $s\in\mathbb{S}$. Secondly, the strong universal hash family $\mathscr{H}_{J}$, from which $h$ is drawn, is independent of $\mathbf{X}_{n}$.
Formally, for any $n\geq1$, we can thus write:
\begin{align} \label{eq:exchangeable_model_hash}
\begin{split}
   X_1,\ldots,X_{n} &\,\simiid\, P, \\
  h&\,\overset{\text{ind}}{\sim}\,\mathscr{H}_{J}, \\
  C_{j} &\,=\, \sum_{i=1}^{n} I( h(X_i) = j), \qquad \forall j \in [J].
\end{split}
\end{align}
The next theorem provides the conditional distribution of $f_{X_{n+1}}$, given the sketch $\mathbf{C}_{J}$, the bucket $h(X_{n+1})$ in which $X_{n+1}$ is hashed, and $h$. This is a key component of our approach.

\begin{theorem}\label{teo_cond_prob}
For $n\geq1$, suppose $(x_{1},\ldots,x_{n})$ is a random sample $\mathbf{X}_{n}$ from \eqref{eq:exchangeable_model_hash} with corresponding sketch $\mathbf{C}_{J}=\mathbf{c}$, obtained using a fixed hash function $h$. If  $\mathbb{S}_{j}:=\{s\in\mathbb{S}\text{ : }h(s)=j\}$ and $q_{j}:=\prob[h(X_{i})=j \mid h]$ for any $j \in [J]$ and $i \in [n] := \{1,\ldots,n\}$, then the $X_{i}$'s hashed into the $j$-th bucket, i.e., $\{X_{i}\text{ : }h(X_{i})=j\}$, are i.i.d.~as $P_{j}=\sum_{s\in\mathbb{S}_{j}}\frac{p_{s}}{q_{j}}\delta_{s}$. Further, for $j \in [J]$ and $r \in \{0,1,\ldots,c_{j}\}$,
\begin{equation} \label{cond_prob}
\pi_{j}(r;P, h):= \prob[f_{X_{n+1}}=r  \mid \mathbf{C}_{J}=\mathbf{c},h(X_{n+1})=j, h]= {c_{j}\choose r}\sum_{s\in\mathbb{S}_{j}}\left(\frac{p_{s}}{q_{j}}\right)^{r+1}\left(1-\frac{p_{s}}{q_{j}}\right)^{c_{j}-r}.
\end{equation}
\end{theorem}

See Appendix~\ref{proof_teo_cond_prob} for the proof of Theorem~\ref{teo_cond_prob}. 
Note that $\mathbb{S}_{j}$, $q_j$, and $P_j$ all implicitly depend on the hash function $h$, which is treated as fixed here.
This result gives us the optimal ``oracle'' estimator $\varepsilon_{f}(P, h)$ of $f_{X_{n+1}}$, under squared error loss,  for given $P$ and $h$:
\begin{equation}\label{cond_prob_mean}
  \varepsilon_{f}(P, h)
  := \E\left[f_{X_{n+1}} \mid \mathbf{C}_{J}=\mathbf{c},h(X_{n+1})=j , h\right]
  =\sum_{r=0}^{c_{j}}r\pi_{j}(r;P,h)=c_{j}\sum_{s\in\mathbb{S}_{j}}\left(\frac{p_{s}}{q_{j}}\right)^{2}.
\end{equation}
This depends on the unknown distribution $P_{j}$, induced by $P$ through hashing  via $h$. 
From \eqref{cond_prob_mean}, it follows that the size $C_{j}$ of the bucket in which $X_{n+1}$ is hashed is a sufficient statistic to estimate $f_{X_{n+1}}$. 
Note that $\varepsilon_{f}(P, h)$ may be interpreted as measuring the diversity in the composition of the $j$-th bucket, because $\sum_{s\in\mathbb{S}_{j}}(p_{s}/q_{j})^{2}$ represents the probability that two randomly chosen elements from the $j$-th bucket correspond to the same symbol.

\subsection{The necessity of prior assumptions} \label{sec:np-estim}

To demonstrate the necessity of introducing modeling assumptions about the data distribution $P$ for obtaining an informative estimate of $f_{X_{n+1}}$, we begin by studying the frequency recovery problem from a frequentist minimax perspective.
This study is inspired by \cite{Pai(22c),Pai(22a),Pai(22b)}, which conducted similar worst-case minimax analyses in the different context of missing mass estimation.

Consider a class of {\em linear} estimators $\hat f_\beta = c_j \beta_{j}$, for $\beta_{j} \geq 0$. This includes the oracle estimator $\varepsilon_f(P, h)$ in \eqref{cond_prob_mean} for $\beta_{j} = \sum_{s \in \mathbb S_j}( p_s / q_j)^2$. The corresponding quadratic risk is:
\begin{equation}\label{eq:risk_freq}
    R(\hat f_{\beta}; P, h) = \E_{P} [(\beta_{h(X_{n+1})}C_{h(X_{n+1})} - f_{X_{n+1}} )^2 \mid h],
\end{equation}
which depends on the unknown data distribution $P$. 

In this analysis, we aim to minimize an upper bound $\tilde{R}(\hat f_{\beta}, h) \geq R(\hat f_{\beta}; P, h)$ for all $P$ within a suitable family $\mathcal{P}$. This leads to a \emph{worst-case optimal} estimator $\beta$ that solves $\min_\beta \tilde{R}(\hat f_{\beta}, h)$. 
We focus on a broad class of distributions with a bounded number of support points, defined as $\mathcal{P}_L := \{P : P \text{ has at most $L$ support points} \}$. While the technical details of this analysis are involved and deferred to Appendix~\ref{app:frequency-recovery-worstcase}, the result is easy to understand.

\begin{theorem}[Informal statement] \label{theorem:worst-case-informal}
For $n\geq1$, suppose $(x_{1},\ldots,x_{n})$ is a random sample $\mathbf{X}_{n}$ from \eqref{eq:exchangeable_model_hash}, with corresponding sketch $\mathbf{C}_{J}=\mathbf{c}$.
For a fixed hash-function $h$, if $L$ is large enough, the worst-case optimal estimator of $f_{X_{n+1}}$, over the class $\mathcal P_L$ for $\tilde{R}$,
% the risk defined in \eqref{eq:risk_freq},
is $ \hat f_\beta \equiv C_{h(X_{n+1})}$. Moreover, the upper bound $\tilde{R}(\hat f_{\beta}, h)$ is tight and it is achieved by $P \equiv P^{\ast}$, where $P^{\ast}$ is a degenerate distribution that places all probability mass on a single symbol $s^* \in \mathbb S$.
\end{theorem}

A rigorous statement of Theorem~\ref{theorem:worst-case-informal} is presented in Appendix~\ref{app:frequency-recovery-worstcase}, along with its proof.
At first sight, this result may seem disappointing because the worst-case optimal estimator is identical to the original CMS upper bound (see Appendix~\ref{app-cms}), and that is often inaccurate because it does not explicitly account for possible hash collisions. 
However, Theorem~\ref{theorem:worst-case-informal} is interesting because the tightness of the analysis highlights the inherent complexity of our frequency recovery problem. In fact, this result tells us that it is impossible to obtain a worst-case estimator that is more informative than the classical CMS upper bound. 

Similar conclusions can also be obtained from an alternative (more classical) minimax analysis that consists of solving $\inf_{\beta} \sup_{P \in \mathcal{P}} R(\hat f_{\beta}; P, h)$,  for a given $h$, where $\mathcal{P}$ is an appropriate family of discrete distributions that will be specified below.
In this case, analytical computations are not feasible but Appendix~\ref{app:minimax} presents a numerical investigation that leads to an equally unsatisfactory estimator.
This minimax estimator is generally uninformative because it always tends to 0 as the latent support size $L$ of the data distribution grows, irrespective of the information contained in the data sketch.

In conclusion, the minimax analyses demonstrate that to obtain informative estimates of $f_{X_{n+1}}$, some assumptions about the data distribution $P$ are needed.
This motivates the BNP approaches developed by prior works, which, however, have their own limitations.

\subsection{The limitations of BNP approaches}\label{sec:difficult_bayes}

In the BNP framework of \citet{Ber(23)}, the model in \eqref{eq:exchangeable_model_hash} is complemented with a prior for $P$, i.e., $P \sim \mathscr P$, an then inference is carried out through the {\em full-sketch} posterior distribution of $f_{X_{n+1}}$ given $\bm C_J$ and $h(X_{n+1})$, i.e.,
\begin{equation}\label{eq:full_post}
 \pi^{F}_j(r) = \frac{\E_{P\sim\mathscr{P}, h\sim\mathscr{H}_J} \left[ \prob[f_{X_{n+1}} = r, \mathbf C_j = \mathbf c, h(X_{n+1}) = j \mid P, h] \right] }{\E_{P\sim\mathscr{P}, h\sim\mathscr{H}_J} \left[ \prob[\mathbf C_j = \mathbf c, h(X_{n+1}) = j \mid P, h] \right]}.
\end{equation}
Theorem 2.2 in \citet{Ber(23)} provides an explicit expression for \eqref{eq:full_post} for a broad class of priors $\mathscr{P}$. However, beyond the Dirichlet process prior, the computational cost of evaluating $\pi^{F}_j(r)$ scales exponentially with $J$ and $n$, highlighting a fundamental challenge in BNP estimation. 
For further discussion on why these BNP estimators are unfeasible in their exact form and difficult to approximate using Monte Carlo methods, see Appendix~\ref{relatedwork-bnp}.

The earlier BNP approaches of \citet{Cai(18)} and \citet{Dol(21), Dol(23)} are somewhat simpler, as they rely on a {\em single-bucket} posterior distribution, $\pi_j^S(r)$, which conditions only on $C_{h(X_{n+1})}$ rather than the full sketch $\mathbf{C}_{J}$. This aligns more closely with the original CMS, which estimates $f_{X_{n+1}}$ using only the information in $C_{h(X_{n+1})}$. 
As shown by \citet{Ber(23)}, $\pi^F$ and $\pi^S$ coincide if and only if $\mathscr{P}$ is the Dirichlet process. While $\pi^S$ is computationally simpler than $\pi^F$, it remains impractical for large-scale applications. For instance, under a Pitman-Yor process prior, the cost of computing $\pi^S$ scales quadratically with $n$ \citep{Dol(23)}. 
Appendix~\ref{app:bnp_cost} provides a more detailed discussion of computational costs. Notably, computing $\pi^S$ is several orders of magnitude slower than our smoothed estimators, while computing $\pi^F$ is infeasible in any realistic setting. This fundamental lack of scalability is the primary limitation of BNP approaches.

The second limitation of BNP approaches is that they struggle to deal with multiple hash functions.
Consider that each datum is passed through $M$ independent hash functions $h_1, \ldots, h_M$, resulting in a sketch $\mathbf{C}_{M, J} \in \mathbb{N}_0^{M \times J}$, a matrix with entries $C_{m, j} = \sum_{i=1}^n \mathbf{1}[h_m(x_i) = j]$. In this case, \citet{Cai(18)} assume that the posterior distribution of $f_{X_{n+1}}$ given $(C_{1, h_1(X_{n+1})}, \ldots, C_{M, h_M(X_{n+1})})$ is proportional to the product of the single-bucket posteriors computed using each row of the sketch separately, motivated by the independence of the hash functions. However, this assumption relies not just on the independence of the hash functions but on the independence of the different rows of $\mathbf{C}_{M, J}$. 
Unfortunately, this independence property does not hold in general, as explained in Section~\ref{sec:ex_multihash}.

% exploit the independence of the hash functions to write
% \begin{equation}\label{eq:mult_hash_post}
%     \prob[f_{X_{n+1}} = r \mid \mathbf C_{\mathbf j} = \mathbf c_{\mathbf j}, \mathbf h(X_{n+1}) = \mathbf j] \propto \prod_{\ell = 1}^m \pi_{j_\ell}^{S}(r),
% \end{equation}
% where $\mathbf h(X_{n+1}) = (h_1(X_{n+1}), \ldots, h_M(X_{n+1}))$, $\mathbf j = (j_1, \ldots, j_M)$, $\mathbf C_{\mathbf j} = (C_{1, j_1}, \ldots,  C_{M, j_M})$, and $\pi_{j_\ell}^{S}(r)$ is \MBtext{the single bucket posterior} computed using the $\ell$-th row of the sketch.
% However, underpinning \eqref{eq:mult_hash_post} is not simply the independence of the hash functions, but rather the independence of the different rows of $\mathbf C_{M, J}$.
% Unfortunately, this independence property does not hold in general, as explained in Section~\ref{sec:ex_multihash}.

These limitations of BNP approaches, combined with the necessity of using prior assumptions on the data distribution $P$ highlighted in Section~\ref{sec:np-estim}, motivate our proposal of the novel {\em smoothed-Bayesian} approach presented in the next section.

%%%%%%%%%%%%%%%%%%%%%%%%%%%%%%%%
%%%%%%%%%%%%%%%%%%%%%%%%%%%%%%%%
%%%%%%%%%%%%%%%%%%%%%%%%%%%%%%%%
%%%%%%%%%%%%%%%%%%%%%%%%%%%%%%%%

\section{Smoothed frequency estimation} \label{sec:smoothed}

\subsection{Outline of our approach}

The smoothed-Bayesian estimators presented in this section are designed to incorporate prior information naturally, similar to Bayesian approaches, but without the same computational issues.
The ideas underlying our approach date back to \cite{Goo(53)}, which first introduced the idea of {\em smoothing an oracle estimator}, such as that in~\eqref{cond_prob_mean}. 
In the context of frequency recovery from sketched data, smoothed estimation can lead to informative approximations of $f_{X_{n+1}}$ by leveraging assumptions about the distribution $P$ on $\mathbb{S}$. 
The main difficulty, however, is that the oracle estimator $\varepsilon_{f}(P,  h)$ in \eqref{cond_prob_mean} depends on an intrinsically unobservable quantity, namely the distribution $P_{j}$ on $\mathbb{S}_{j}$ induced by the hashing procedure.

Therefore, our approach requires specifying $P$ carefully to ensure that the bucket-restricted distributions $P_{j}$ induced from $P$ are well-defined and mathematically tractable. This problem is challenging to address with parametric assumptions on $P$ directly, but it can be solved using suitable nonparametric assumptions, as discussed next.

Our solution, inspired by existing BNP approaches, is to treat $P$ as a {\em random} element from the space of discrete probability distributions on $\mathbb{S}$, modeling its distribution $\mathscr{P}$ through a low-dimensional smoothing parameter. Specifically, we consider the broad class of normalized random measures (NRMs) \citep{Jam(02),Pru(02),Pit(03),Reg(03)}. This choice is inspired by prior work on BNP estimators for $f_{X_{n+1}}$ \citep{Cai(18),Dol(21),Dol(23)}, and it leads to an explicit expression for the expected value of $\varepsilon_{f}(P)$ with respect to the smoothing distribution $\mathscr{P}$, which gives us a practical estimator of $f_{X_{n+1}}$. Further, NRMs allow for the empirical estimation of the smoothing parameters from the sketch $\mathbf{C}_{J}$, which, similar to the prior's parameters in BNP approaches, can significantly impact posterior inferences \citep{giordano2023evaluating}; see Appendix~\ref{sec:estimation} for details on the estimation of the smoothing parameters.
Until then, we will assume that the smoothing parameters are known.

\subsection{Background on normalized random measures}\label{sec:background_nrm}

Consider a purely atomic completely random measure (CRM) $\tilde{\mu}$ on $\mathbb{S}$, $\tilde{\mu}(\cdot) = \sum_{k \geq 1} J_k \delta_{S_k}(\cdot)$ with L\'evy intensity measure $\nu(\ddr x, \ddr s)$ on $\R_+ \times \mathbb S$ \citep{Kin(67),Kin(93)}.
We consider L\'evy intensities of the form $\nu(\ddr x,\,\ddr s) = \theta G_0(\ddr s) \rho(x) \ddr x$ where $\theta > 0$ is a parameter, $G_0$ is a non-atomic probability measure on $\mathbb{S}$, and $\rho(x) \ddr x$ is a measure on $\R_+$ such that $\int_{\R_+} \rho(x) \ddr x = +\infty$ and $\psi(u) = \int_{\R_+}(1 - e^{-ux}) \rho(x) \ddr x < +\infty$ for all $u>0$, ensuring the total mass $T=\tilde{\mu}(\mathbb{S})$ of $\tilde{\mu}$ is positive and almost-surely finite \citep{Pit(03),Reg(03)}.
Then, we say that a distribution $P$ on $\mathbb{S}$ is an NRM with parameter $(\theta, G_0,\rho)$, in short $P\sim\text{NRM}(\theta,G_{0},\rho)$, if 
$P(\cdot)= \tilde{\mu}(\cdot)/T=\sum_{k\geq1} (J_{k}/T) \delta_{S_{k}}(\cdot)$,
where the distribution of the (random) probabilities $(J_{k}/T)_{k\geq1}$ is directed by $\rho$, and the locations $(S_{k})_{k\geq1}$ are i.i.d.~from $G_{0}$, independent of $(J_{k}/T)_{k\geq1}$.

\subsection{Smoothed estimation with normalized random measures} \label{sec:smoothed-nrm}

We leverage smoothing assumptions for the distribution $P$ as follows. 
We consider $P\sim\text{NRM}(\theta,G_{0},\rho)$, and estimate $f_{X_{n+1}}$ by taking the expected value of $\varepsilon_{f}(P,  h)$ with respect to the law of $P$  and the distribution of $h$. It follows from \eqref{cond_prob_mean} that, to compute the expectation of $\varepsilon_{f}(P, h)$, it suffices to evaluate
\begin{equation}\label{eq:smoothed_proba}
    \pi_{j}(r) := \E_{P\sim\text{NRM}(\theta,G_{0},\rho), h \sim \mathscr H_J}[\pi_{j}(r;P, h)],
\end{equation}
for all $j \in [J]$. Each $\pi_{j}(r)$ can be computed by exploiting the Poisson process representation of NRMs and the Poisson coloring theorem \citep[Chapter 5]{Kin(93)}. This is achieved by the following {\em restriction property} of NRMs.

Suppose $P$ is a NRM and, for any Borel set $A\in\mathcal{S}$, let $P_{A}$ denote the random probability measure on $A$ induced by $P\sim\text{NRM}(\theta,G_{0},\rho)$; i.e., the renormalized restriction of $P$ to the set $A$.
Then, $P_{A}\sim\text{NRM}(\theta G_{0}(A),G_{0,A}/G_{0}(A),\rho)$, where $G_{0,A}$ is the restriction of the probability measure $G_{0}$ to $A$. 
This property of NRMs is critical to compute $\pi_{j}(r)$.

\begin{theorem} \label{teo_smooth}
For $n\geq1$, let $(x_{1},\ldots,x_{n})$ be a random sample $\mathbf{X}_{n}$ from \eqref{eq:exchangeable_model_hash}, and let $\mathbf{C}_{J}=\mathbf{c}$ be the corresponding sketch. If $P\sim\text{NRM}(\theta,G_{0},\rho)$, then for any $j \in [J]$ and $r \in \{0,1,\ldots,c_{j}\}$, 
\begin{equation}\label{mass_nrm}
\pi_{j}(r)=\int_{0}^{1}\text{Binomial}(r;c_{j},v)f_{V_{j}}(v)\ddr v,
\end{equation}
where
\begin{equation}\label{struct}
f_{V_{j}}(v)=\frac{\theta}{J}v\int_{0}^{+\infty}t\rho(tv)f_{T_{j}}(t(1-v))\ddr t.
\end{equation}
Above, $f_{T_{j}}$ denotes the density function of the total mass of a CRM with L\'evy intensity $\theta G_{0,\mathbb{S}_{j}}(\ddr s) \rho(x) \ddr x$.
\end{theorem}

See Appendix~\ref{proof_teo_smooth} for the proof of Theorem~\ref{teo_smooth}. By the tower property of the expectation, we get that, if $h(X_{n+1}) = j$, the smoothed version of \eqref{cond_prob_mean} is
\begin{equation}\label{eq:smooth_estim}
    \hat f^{\text{NRM}}_{X_{n+1}} = \E_{P \sim \text{NRM}, h \sim \mathscr H_J}[\varepsilon_f(P, h)] = c_j \E[V_j],
\end{equation}
where $V_j$ is as in \eqref{struct}. 
Intuitively, $\E[V_j]$ is equal to the probability that two symbols sampled independently at random from 
% $P_j$
$P_{\mathbb S_j}$
are equal \citep[see Equation (2.25) in][]{Pit(06)}. 
Therefore, the lower $\E[V_j]$, the higher the number of distinct symbols in the sample $\mathbf{X}_n$, leading to more hash collisions. This, in turn, inflates $c_j$ relative to $f_{X_{n+1}}$. Specific examples will be presented in Section~\ref{sec:smoothed-examples}, where it will also become clear that smoothed estimators are more practical and have much lower computational costs compared to their BNP counterparts.

The following result establishes an optimality property of $\hat{f}^{\text{NRM}}_{X_{n+1}}$.
\begin{theorem}\label{teo:optimality}
    Let $\hat f_{\beta} = c_j \beta$ for $\beta \geq 0$ and consider the quadratic risk associated with $\beta$ conditional on the bucket into which observation $n+1$ is hashed, namely $\mathrm{cMSE}(\beta) = \E\left[(\beta C_j - f_{X_{n+1}})^2 \mid h(X_{n+1}) = j \right]$.
    If $X_{1}, \ldots, X_{n+1} \mid P \simiid P$ and $P \sim \text{NRM}(\theta, G_0, \rho)$, then $\beta = \E[V_j]$ achieves the minimum risk. Moreover, this estimator is unbiased.
\end{theorem}
See Appendix~\ref{proof_teo_optimality} for the proof.
Theorem~\ref{teo:optimality} has a clear Bayesian interpretation. Under the hierarchical model $X_{1}, \ldots, X_{n+1} \mid P \simiid P$ and $P \sim \text{NRM}(\theta, G_0, \rho)$, we know that the optimal estimator is the Bayes estimator:
$\E[f_{X_{n+1}} \mid \mathbf{C}_j, h(X_{n+1}) = j]$,
which can be expressed as $\sum_{r=0}^{c_j} r \pi_j^F(r)$,
where $\pi_j^F$ denotes the full-sketch posterior distribution given in \eqref{eq:full_post}. However, as discussed in Section~\ref{sec:difficult_bayes}, this posterior is generally intractable. Theorem~\ref{teo:optimality} thus says that our smoothed estimator is the closest, among all linear estimators, to this optimal but generally intractable Bayes estimator.

Moreover, Theorem~\ref{teo:optimality} can guide us in estimating the smoothing parameter $\theta > 0$ and any hyper-parameters in the expression of the L\'evy intensity $\rho$. 
In fact, the main assumption of Theorem~\ref{teo:optimality} is that the data are randomly sampled from the random probability measure $P \sim \text{NRM}$. Therefore, a sensitive approach for choosing the smoothing parameters is to (approximately) maximize the marginal likelihood of the sketch. See Appendix~\ref{sec:estimation} for further details on how to estimate the smoothing parameters efficiently, by leveraging a subsampling shortcut that allows this component of our method to maintain a computational cost that is constant with respect to the size of the sketched data.

\subsection{Examples of smoothed estimators} \label{sec:smoothed-examples}

As interesting examples of NRM smoothing distributions, we consider the Dirichlet \citep[DP,][]{Fer(73)} and normalized generalized Gamma \citep[NGGP,][]{Jam(02),Pru(02),Pit(03),Lij(07)} processes. These are NRMs with L\'evy intensities $\rho(x) = \text{e}^{-x} x^{-1}$ and $\rho(x) = {\Gamma(1-\alpha)}^{-1} x^{-1-\alpha}\text{e}^{-\tau x}$, respectively. 
They lead to simple smoothed estimators that can be directly related to existing BNP approaches.
The NGGP, in particular, allows modeling $P$ with a flexible tail behaviour, ranging from geometric to heavy power-law tails.
 The details of the derivations summarized below are in Appendix~\ref{app:smoothing-examples}.

\paragraph{Smoothing with the Dirichlet process.}
Specializing \eqref{struct} to the L\'evy intensity of a DP with parameter $\theta > 0$, it is possible to see that $V_j \sim \mathrm{Beta}(1, \theta/J)$. Then, for any $j \in [J]$ and $r \in \{0,1,\ldots,c_{j}\}$, $\pi_{j}(r)$ in~\eqref{mass_nrm} becomes
\begin{equation}\label{mass_dp}
\pi_{j}(r) = {c_{j}\choose r}\frac{\theta}{J}\frac{\Gamma(r+1)\Gamma(\theta/J+c_{j}-r)}{\Gamma(\theta/J+c_{j}+1)},
\end{equation}
which is the Beta-Binomial distribution with parameter $(c_{j},1,\theta/J)$. Then, \eqref{eq:smooth_estim} reduces to:
\begin{align} \label{eq:prop-dp}
    \hat{f}_{X_{n+1}}^{\text{DP}}:=
    \E_{P \sim \text{DP}, h \sim \mathscr H_J}[\varepsilon_f(P, h)] =c_{j}\frac{J}{\theta+J},
\end{align}

Interestingly, this estimator coincides with the BNP estimator under the DP prior \citep{Cai(18)}. 
However, the DP is the sole NRM for which the smoothed-Bayesian and BNP approaches lead to the same estimator for $f_{X_{n+1}}$.
\begin{theorem}\label{dp_charact}
Let $\hat{f}_{X_{n+1}}^{\text{NRM}}$ be the smoothed estimator under the model \eqref{eq:exchangeable_model_hash} obtained with $P\sim\text{NRM}(\theta,G_{0},\rho)$. 
This estimator coincides with the BNP estimator obtained by considering the expectation of $\pi^F$ (cf.\ Section \ref{sec:difficult_bayes}) if and only if the NRM is the DP.
\end{theorem}
See Appendix~\ref{app_dp_charact} for a proof of Theorem~\ref{dp_charact}. We also refer to Appendix~\ref{relatedwork-bnp} for a more careful discussion of the relation between smoothed estimation and the BNP approach.

The estimator in~\eqref{eq:prop-dp} applies linear shrinkage to $c_j$ and is very fast to compute for any value of $\theta$. Furthermore, the smoothing parameter $\theta$ can be efficiently estimated by maximizing the marginal likelihood of the data sketch, as detailed in Appendix~\ref{sec:estimation}.

It is also interesting to note that the estimator in~\eqref{eq:prop-dp} converges to the CMS estimator as $\theta \rightarrow 0$, consistent with the worst-case analysis from Section~\ref{sec:np-estim}. 
Indeed, as $\theta \to 0$, $P \sim DP(\theta, G_0)$ approaches a degenerate distribution with only one support point; see Theorem~\ref{theorem:worst-case-informal}. In general, \eqref{eq:prop-dp} shrinks $c_j$ by a weight inversely proportional to the parameter $\theta$, which is intuitive since larger values of $\theta$ lead to a sample with more distinct symbols.

\paragraph{Smoothing with the normalized generalized Gamma process.}

We now consider smoothing with an NGGP with parameters $(\theta,\alpha,\tau)$, which does not have a practical counterpart in the BNP framework due to the intractability of the full posterior distribution for a general NGGP prior.
In this case,~\eqref{eq:smooth_estim} can be evaluated by noting that
\begin{equation}\label{struct_ngg}
V_{j}=V_{\theta,\alpha,\tau} := B_{1-\alpha,\alpha}\left(1-\left(\frac{\frac{\theta\tau^{\alpha}}{J\alpha}}{\frac{\theta\tau^{\alpha}}{J\alpha}+E}\right)^{1/\alpha}\right),
\end{equation}
where $B_{1-\alpha,\alpha}$ is a Beta random variable with parameter $(1-\alpha,\alpha)$, and $E$ is an independent negative Exponential random variable with parameter $1$.
Therefore, for any $j \in [J]$ and $r \in \{0,1,\ldots,c_{j}\}$, $\pi_{j}(r)$ in~\eqref{mass_nrm} becomes
\begin{equation}\label{mass_ngg}
\pi_{j}(r)=\int_{0}^{1}{c_{j}\choose r}v^{r}(1-v)^{c_{j}-r}f_{V_{\theta,\alpha,\tau}}(v)\ddr v,
\end{equation}
which is a generalization of the Beta-Binomial distribution introduced in \eqref{mass_dp}. Then, the estimator in \eqref{eq:smooth_estim} reduces to
\begin{align}  \label{eq:prop-nggp}
\hat{f}_{X_{n+1}}^{\text{\tiny{(NGGP)}}} = c_{j}(1-\alpha)\left(1-\frac{\theta\tau^{\alpha}}{J\alpha}\text{e}^{\frac{\theta\tau^{\alpha}}{J\alpha}}E_{1/\alpha}\left(\frac{\theta\tau^{\alpha}}{J\alpha}\right)\right),
\end{align}
where $E$ denotes the exponential integral function: $E_{a}(z) := \int_{1}^{+\infty}x^{-a}\exp\{-zx\}\ddr x$.
See Appendix~\ref{proof_prop_nggp} for the proof of \eqref{eq:prop-nggp}. 

Similar to the estimator under DP smoothing, \eqref{eq:prop-nggp} also applies linear shrinkage to $c_j$ and is straightforward to evaluate for any choice of $(\theta, \alpha, \tau)$, requiring only the computation of the exponential integral function, which is easily evaluated numerically. The smoothing parameters $(\theta, \alpha, \tau)$ can be estimated by (approximately) maximizing 
% the marginal likelihood of the sketch, as detailed in Appendix~\ref{sec:estimation}. 
either the marginal likelihood of the sketch, or the marginal likelihood of a (small) subset of the data.
Although this process is somewhat more involved than estimating the parameter $\theta$ in the DP case, due to the absence of a closed-form marginal likelihood expression for the general NGGP, it remains practical. 
% As explained in Appendix~\ref{sec:estimation}, one approach to estimate these parameters efficiently is to simulate via a Monte Carlo algorithm a much smaller subset of the sketched data.
See Appendix~\ref{sec:estimation} for further details.

Note that it is intuitive why the smoothed estimator in \eqref{eq:prop-nggp} is decreasing in all three parameters $\theta$, $\tau$, and $\alpha$: the expected number of distinct symbols in $\mathbf{X}_{n}$ is an increasing function of these parameters \citep{Lij(07)}.

In the special case of $\alpha=0$ and $\tau=1$,  \eqref{eq:prop-nggp} reduces to \eqref{eq:prop-dp}, as $\text{NGGP}(\theta,0,1)$ reduces to $\text{DP}(\theta)$.
This equivalence also can be seen by looking at the random variable $V_{\theta,\alpha,\tau}$ in \eqref{struct_ngg}, since
$
\lim_{\alpha\rightarrow0}V_{\theta,\alpha,1}=1-\text{e}^{-\frac{J}{\theta}E}\stackrel{\text{d}}{=}B_{1,\theta/J}.
$
A more interesting case is that of $\alpha=1/2$, which is special NGGP known as the normalized inverse Gaussian process (NIGP) \citep{Lij(05),Fav(12)}. The smoothing parameter $\alpha\in[0,1)$ controls the tail behaviour of the distribution $P$: the larger $\alpha$, the heavier the tail \citep{Lij(07)}. Therefore, the NGGP can capture flexible tail behaviours for the data distribution, ranging from the geometric tails of the DP to the heavy power-law tails, which are often observed in applications.

\paragraph{Beyond the normalized generalized Gamma process.} \label{sec:smoothing-beyond-nggp}

Given the generality of Theorem~\ref{teo_smooth}, it is natural to wonder whether it might be worth considering alternative smoothing distributions \citep{Lij(10)}. To address this question, we point out that the use of other NRMs may lead to less explicit distributions for $V_{j}$ in Theorem~\ref{teo_smooth}, which is likely to complicate the estimation of $f_{X_{n+1}}$. Fortunately, though, the NGGP already offers a reasonable trade-off between the flexibility of the smoothing assumption, in terms of enabling a flexible tail behaviour, and its tractability, in terms of leading to a distribution for $V_{j}$ that can be evaluated easily. To highlight the advantages of the NGGP, we plot in Figure~\ref{fig:smooth_comparison} the probability $\pi_j(r)$ computed under different choices of $\theta$ and $\alpha$, assuming $\tau=1$. The power of the NGGP can be evinced by comparing the relatively flexible behaviour of $\pi_j(r)$, for different choices of $\theta$ and $\alpha$, to the much more constrained behaviour corresponding to the DP, which limits $\pi_j(r)$ to be either monotonically increasing or decreasing.

\begin{figure}[t]
    \centering
    \includegraphics[width=\linewidth]{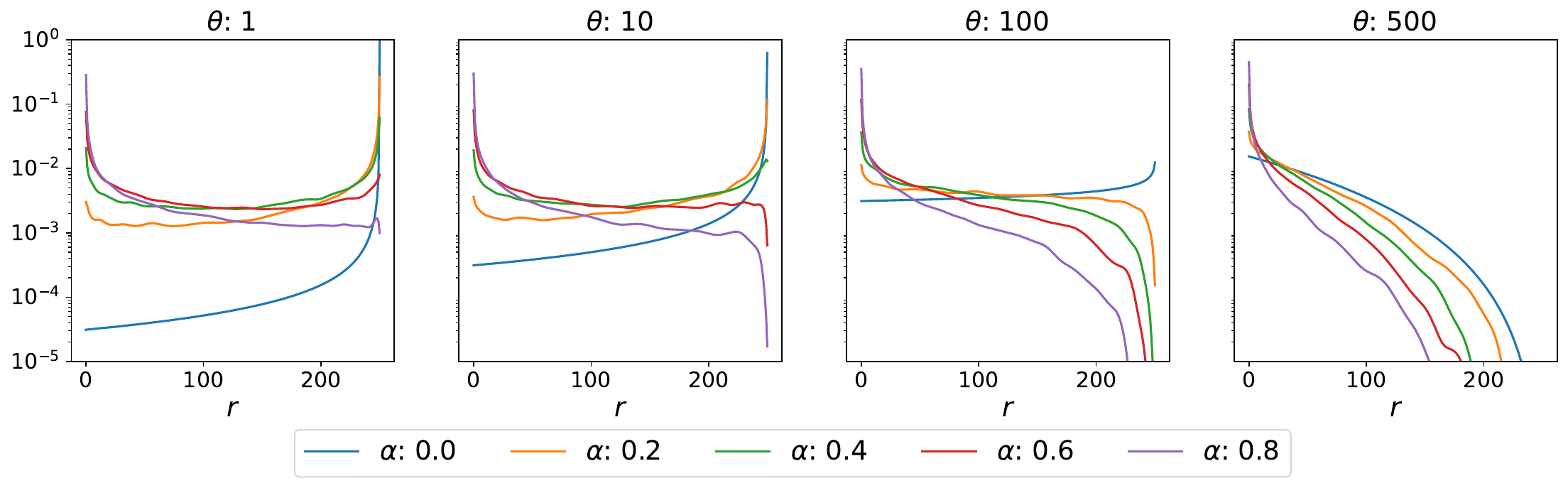}
    \caption{Visualization of the modelling flexibility of the NGGP. The probabilities $\pi_j(r)$ are plotted as a function of $r$ for $c_j = 250$, for different values of the NGGP smoothing parameters. Different panels focus on different values of $\theta$, while the curves drawn in different colors correspond to alternative values of $\alpha$. In all cases, $\tau = 1$.}
    \label{fig:smooth_comparison}
\end{figure}

We conclude with some remarks about the connection between the smoothed frequency estimators obtained in under the DP and NGGP smoothing distributions, and the estimators obtained with the BNP approach \citep{Ber(23)}. As discussed in Section~\ref{sec:difficult_bayes} the use of prior distributions other than the DP, including for example other NRMs, leads to computationally challenging BNP estimators for the empirical frequency $f_{X_{n+1}}$.
By contrast, the advantage of the approach of smoothed estimation is that it allows one to leverage much more flexible models from the broader class of NRMs and still obtain tractable estimators that can be written in terms of expectations of mixtures of binomial distributions.

%%%%%%%%%%%%%%%%%%%%%%%%%%%%%%%%
%%%%%%%%%%%%%%%%%%%%%%%%%%%%%%%%
%%%%%%%%%%%%%%%%%%%%%%%%%%%%%%%%
%%%%%%%%%%%%%%%%%%%%%%%%%%%%%%%%

\section{Sketches with multiple hash functions} \label{sec:multi-view}

\subsection{Problem statement and setup}

We discuss how to extend our smoothed-Bayesian approach to combine information from sketches produced by multiple independent hash functions, as in the general setting of the CMS of \citet{Cor(05)}. This strengthens the connection between our approach and the algorithmic approach of the CMS, reviewed in Appendix~\ref{app-cms}.
In fact, the latter is typically applied using multiple independent hash functions, to make the data compression more efficient \citep[Chapter 2]{Cor(20)}.

To introduce the more general sketch, for $M\geq1$, let $h_{1},\ldots,h_{M}$ be $J$-wide hash functions, with $h_{l}:\mathbb{S}\rightarrow [J]$ for $l \in [M]$, that are i.i.d.~from a pairwise independent hash family $\mathscr{H}_{J}$. Hashing $(X_{1},\ldots,X_{n})$ through $h_{1},\ldots,h_{M}$ produces a random matrix $\mathbf{C}_{M,J}\in\mathbb{N}_{0}^{M\times J}$, referred to as the sketch. Therefore, the natural extension of the setup of Section~\ref{sec:setup-assumptions} is:
\begin{align}\label{eq:exchangeable_model_multhash}
  \begin{split}
    X_1,\ldots,X_{n} &\,\simiid\, P,\\
    h_{1},\ldots,h_{M}&\,\simiid\,\mathscr{H}_{J},\\
    C_{l,j} &\,=\, \sum_{i=1}^{n} I( h_{l}(X_i) = j), \qquad \forall l \in [M], \; \forall j \in [J].
  \end{split}
\end{align}
Following the same strategy as in Section~\ref{sec:setup-assumptions}, one may think of investigating an oracle estimator for $f_{X_{n+1}}$, by assuming full knowledge of the distribution $P$. However, even this preliminary step becomes extremely challenging in the context of multiple hash functions. In particular, the resulting estimator involves intractable combinatorial coefficients that reduce to the multinomial coefficient only if $M=1$. See Appendix~\ref{app:multihash} for further details.

To circumvent the computational challenges associated with direct estimation in the case of multiple hash functions, we propose an approach inspired by {\em multi-view} (or multimodal) learning \citep{Xu(13), Sha(18),Li(18b)}. In general, multi-view learning is concerned with aggregating distinct inferences obtained from different views or representations of the same data set. Here, we focus on aggregating the separate inferences corresponding to the distinct sketches produced by each hash function.

\subsection{Multi-view learning}

We consider the sketch $\mathbf{C}_{M,J}$ as a collection of $M$ distinct sketches, say $\mathbf{C}_{1,J},\ldots,\mathbf{C}_{M,J}$, where each $\mathbf{C}_{\ell,J}$ for $\ell \in [M]$ is obtained by applying the corresponding hash function $h_{\ell}$ to the sample $(x_{1},\ldots,x_{n})$. Following the general multi-view learning approach, we can interpret the sketches $\mathbf{C}_{1,J},\ldots,\mathbf{C}_{M,J}$ as $M$ different views of the data set, and then apply Theorem~\ref{teo_cond_prob} to each of them separately.
This leads to a distinct conditional distribution for $f_{X_{n+1}}$, given the sketch $\mathbf{C}_{\ell,J}$ and the bucket $h_{\ell}(X_{n+1})$ in which $X_{n+1}$ is hashed, for each $\ell \in [M]$. That is, for each $r \in \{0,1,\ldots, c_{\ell,j}\}$, we compute the probability
\begin{align}\label{eq:single_view_post}
\begin{split}
  \pi_{j_{\ell}}(r;P)
   : & = \prob[f_{X_{n+1}}=r \mid \mathbf{C}_{l,J}=\mathbf{c}_{\ell},h_{\ell}(X_{n+1})=j_{l}] \\
  & = {c_{\ell,j}\choose r}\sum_{s\in\mathbb{S}_{j_{\ell}}}\left(\frac{p_{s}}{q_{j_{\ell}}}\right)^{r+1}\left(1-\frac{p_{s}}{q_{j_{\ell}}}\right)^{c_{\ell,j}-r}.
\end{split}
\end{align}
Then, we consider two practical and intuitive multi-view rules to aggregate the probabilities $\pi_{j_{\ell}}(r;P)$ in \eqref{eq:single_view_post}, as explained below, although other options are also possible.

The first multi-view rule that we consider is known as the {\em product of experts} \citep{Hin(02)}. 
For each $r \in \{0,1,\ldots,n\}$, it defines the probability mass function as follows:
\begin{equation}\label{eq:poe_post}
    \tilde{\pi}_{\mathbf j}(r;P):= \frac{1}{Z} \prod_{l=1}^M \pi_{j_\ell}(r, P),
\end{equation}
where $Z$ is the normalizing constant. 
 Intuitively, the distribution function \eqref{eq:poe_post} assigns high mass to values of $r$ for which none of the individual ``single-view" distributions $\pi_{j_\ell}(r)$ is too small and are well supported by most of these distributions. 
This can be interpreted as seeking a consensus among ``experts" (i.e., distributions) that agree on a region of values. 
For further intuition on the product of experts, we refer to \cite{Hin(02)}.

The second multi-view rule is called {\em minimum of experts}, and it is inspired by the CMS. 
It considers the distribution of the minimum of $M$ independent random variables distributed according to \eqref{eq:single_view_post}. That is, if $\Pi_{j_{\ell}}(\cdot;P)$ denotes the cumulative distribution function corresponding to \eqref{eq:single_view_post}, for each $r \in \{0,1,\ldots,n\}$ we define a probability mass function $\mathring \pi_{\mathbf j}(r; P)$ associated with a
cumulative distribution function given by
 \begin{equation}\label{eq:min_post}
    \mathring \Pi_{\mathbf j}(r;P) := 1 - \prod_{\ell=1}^M (1 - \Pi_{j_{\ell}}(r, P)).
\end{equation}

Based on the approximations in \eqref{eq:poe_post} and \eqref{eq:min_post} of the ``true'' (intractable) conditional distribution of $f_{X_{n+1}}$, we propose to utilize the corresponding expected values as practical estimators of the empirical frequency $f_{X_{n+1}}$ from the sketch $\mathbf{C}_{M,J}$, following the same arguments developed in Sections~\ref{sec:setup} and~\ref{sec:smoothed}. 

Under the (fully) nonparametric approach of Section~\ref{sec:np-estim}, it can be seen that the application of the worst-case analysis to each single view leads to recovering the classical CMS algorithm from \eqref{eq:min_post}.
Indeed, from Theorem~\ref{theorem:worst-case-informal} it is straightforward to see that $\pi_{j_\ell}(r) = \delta_{c_{\ell, j_\ell}}(r)$ and \eqref{eq:min_post} reduces to $\min \{c_{1, j_1}, \ldots, c_{M, j_M}\}$.
Under the approach of smoothed estimation of Section~\ref{sec:smoothed}, the DP smoothing assumption on \eqref{eq:poe_post} leads to an expression for $\tilde{\pi}_{\mathbf j}$ which coincides  with the posterior distribution of  $f_{X_{n+1}}$, given $\mathbf{C}_{M,J}$ and $\mathbf{h}(X_{n+1})$, under a DP prior, obtained in \citet[Theorem 2]{Cai(18)}. See also Equation (8) in \cite{Dol(23)}. Therefore, the smoothed estimator $\tilde{\hat{f}}_{X_{n+1}}^{\text{DP}}$ coincides with the BNP estimator proposed by \citet{Cai(18)} as the posterior mean. See also \citet{Dol(23)} for further details.  

Along the same lines, one may also consider further generalizations by applying the NGGP smoothing assumption to \eqref{eq:poe_post}. These smoothing assumptions could also be applied to \eqref{eq:min_post}, leading to estimators that have not been previously considered in the literature.

%%%%%%%%%%%%%%%%%%%%%%%%%%%%%%%%
%%%%%%%%%%%%%%%%%%%%%%%%%%%%%%%%
%%%%%%%%%%%%%%%%%%%%%%%%%%%%%%%%
%%%%%%%%%%%%%%%%%%%%%%%%%%%%%%%%

\section{Numerical illustrations} \label{sec:numerics}

\subsection{Setup and performance metrics} \label{sec:numerics-setup}

We present several simulation studies to assess the performance of our methods. We consider two generative models for $X_1, \ldots, X_n$: i) the process (generalized P\'olya urn) induced by the Pitman-Yor process, with parameters $\gamma$ (strength) and $\sigma$ (discount) to be specified case-by-case, in short $\text{PYP}(\gamma,\sigma)$; ii) a Zipfian distribution with tail parameter $c$. Both display a power law behaviour. 
We evaluate the results based on mean absolute error between true and estimated frequencies, stratified by true frequency  \citep{Dol(23)}. That is,
\begin{displaymath}
    \text{MAE}_{j} = \frac{1}{\sum_{s \in \mathbb S} I(f_s \in (l_j, u_j])} \sum_{s \in \mathbb S} |f_s - \hat f_s | \, I\left(f_s \in (l_j, u_j]\right),
\end{displaymath}
where $f_s = \sum_{i=1}^n I(X_i = s)$ is empirical frequency, $\hat f_s$ its  estimate, and $(l_j, u_j]_{j \geq 1}$ are non-overlapping frequency bins.

\subsection{Synthetic data: single hash function}\label{sec:ex_singlehash}

We consider the two smoothing distributions discussed in Section~\ref{sec:smoothed}: the Dirichlet process and the NGGP.
To estimate the parameters in the DP, we maximize the integrated likelihood of the observed sketch as discussed in Appendix~\ref{sec:estimation}.
While considering the NGGP, we store in memory the first $m$ observations $X_1, \ldots, X_m$, with $m \approx n / 20$, and choose the parameters that maximize the integrated likelihood of $X_1, \ldots, X_m$.
See Appendix~\ref{sec:estimation} for further details on the associated algorithms.
As far as the computational cost is concerned, optimizing under the DP smoothing is almost instantaneous ($<0.1$ seconds) while for the  NGGP the optimization takes 1--5 seconds in all experiments.
% is between 30 seconds and 5 minutes. 
Note that the optimization is done only once per experiment and does not need to be repeated for each query.

We generate $n=$ 100,000 observations from $\mathrm{PYP}(\gamma, \sigma)$ for $\gamma \in \{1, 10, 100, 1000\}$ and $\sigma \in \{0, 0.25, 0.5, 0.75\}$ and from a Zipfian distribution with parameter $c = 1.3, 1.6, 1.9, 2.2$.
The data are hashed by a random hash function of with $J=128$.
Figure~\ref{fig:freq_singlehash_py} reports the MAEs of the frequency estimators stratified across 3 frequency bins when data are generated from $\text{PYP}(\gamma, 0.75)$ and from the Zipf distribution, while Table~\ref{tab:freq_singlehash_py} in the Appendix reports all the remaining numerical values, while in Appendix~\ref{sec:app_role_J}, we discuss the impact of the parameter $J$.
Results are averaged across 50 independent repetitions.
It is clear that the NGGP smoothing outperforms the DP smoothing for both data generating processes, and across all values of their parameters.
Henceforth, we will present results only for the NGGP smoothing. The results for the DP smoothing are available in the appendix, where it is clear that the NGGP outperforms the DP also in the other experiments.

\begin{figure}[htb]
    \centering
    \includegraphics[width=\linewidth]{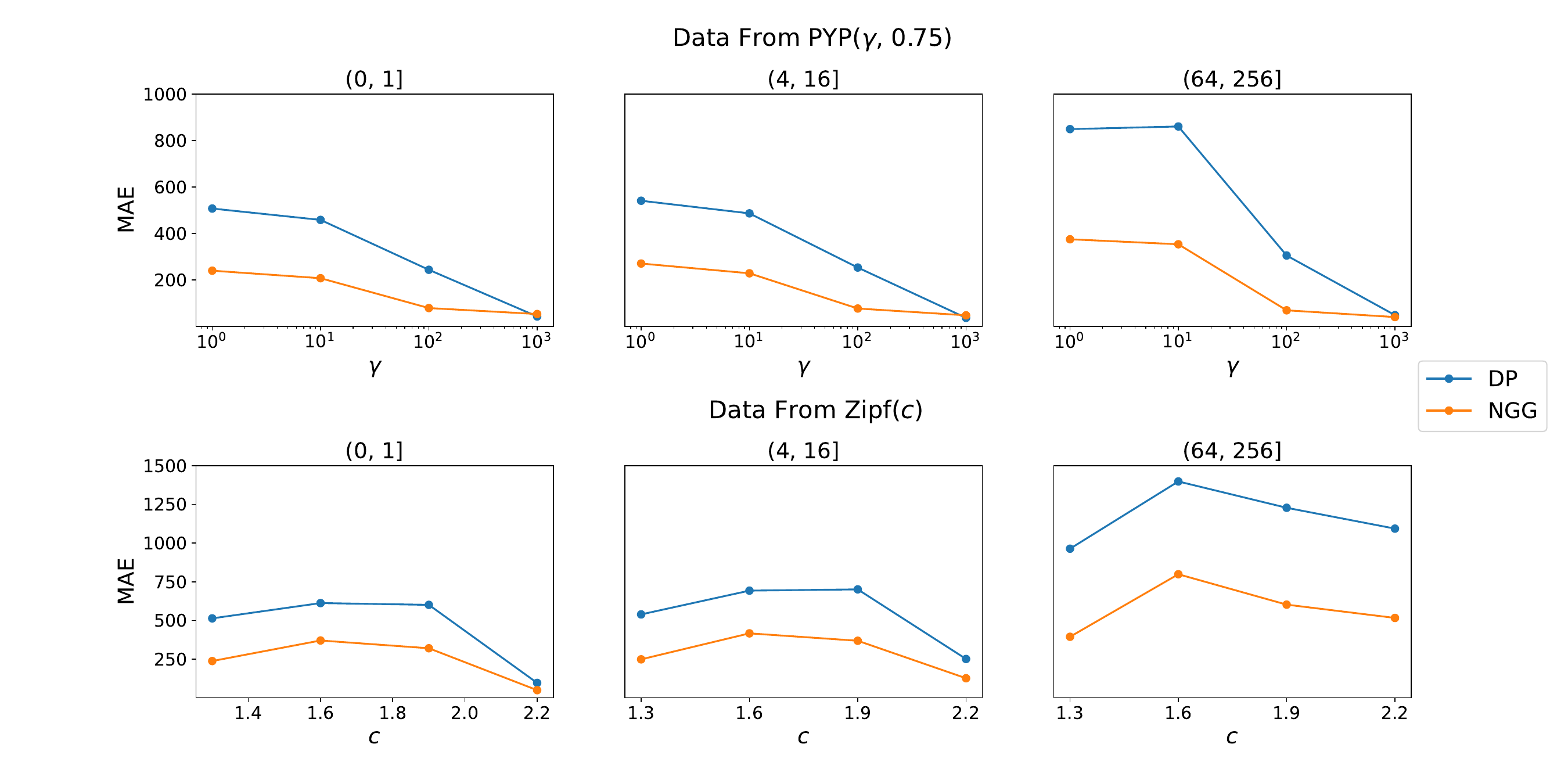}
    \caption{MAEs for the frequency estimators with DP and NGG smoothing, in experiments on synthetic data from a Pitman-Yor process with parameters $\gamma$ (varies across the $x$-axis) and $\sigma=0.75$ (see Section~\ref{sec:ex_singlehash}).
    Different plots correspond to different frequency bins.}
    \label{fig:freq_singlehash_py}
\end{figure}

% \begin{figure}[htb]
%     \centering
%     \includegraphics[width=\linewidth]{images/freq_zipf.pdf}
%     \caption{MAEs for the frequency estimators with DP and NGG smoothing in experiments on synthetic data from a Zipf distribution with parameter $c$ (varies across the $x$-axis), see Section~\ref{sec:ex_singlehash}. Different plots correspond to different frequency bins.}
%     \label{fig:freq_singlehash_zipf}
% \end{figure}

\subsection{Synthetic data: multiple hash functions}\label{sec:ex_multihash}

We now move to the case of multiple hash functions. In the following, we fix a total memory budget for the sketch $\mathbf C_{M,J}$ so that $M \times J =$ 1,000.
We generate data from a $\mathrm{PYP}(\gamma, \sigma)$, letting $\gamma \in \{10, 100, 1000\}$ and $\sigma \in \{0.25, 0.75\}$, and simulate $n=$ 500,000 observations for each possible combination of the parameters.
In analogy with the multiview literature, we analyze each view separately (in parallel) and estimate one set of parameters for each view.

We consider frequency estimators based on the product of expert (PoE) aggregation \eqref{eq:poe_post} and the minimum of expert (min) aggregation \eqref{eq:min_post} rules, under NGG smoothing. 
We also compare their performance with the original count-min sketch algorithm and the ``debiased'' count-min from \cite{ting2018count}, which we call D-CMS.
Figure~\ref{fig:multiview} reports the MAEs for two choices of the parameters of the data generating process, and Tables~\ref{tab:freq_multihash_J50} -~\ref{tab:freq_multihash_J100} in the Appendix report all the remaining values including the case of the DP smoothing.

Several insights emerge from these experiments. First, although it may not be immediately evident in Figure~\ref{fig:multiview}, the min aggregation rule generally outperforms the PoE rule, as shown by Tables~\ref{tab:freq_multihash_J50}–\ref{tab:freq_multihash_J100} in the Appendix.
The CMS and D-CMS demonstrate competitive performance with the smoothed estimator when the data exhibit a moderate power-law behavior (i.e., when the discount parameter of the PYP generating the data is $\sigma = 0.25$). However, with heavy power-law tails, the CMS and D-CMS incur larger errors, particularly at low frequencies.
Across all datasets, the smoothed estimators perform well with a moderate (e.g., 10) number of hash functions. 
The CMS also benefits from using a moderate number of hash functions with moderate power laws, but with heavier power laws, the CMS achieves better performance with only one or two hash functions. 
In contrast, the D-CMS performs poorly when using only a single hash function.

\begin{figure}[htb]
    \centering
    \includegraphics[width=\linewidth]{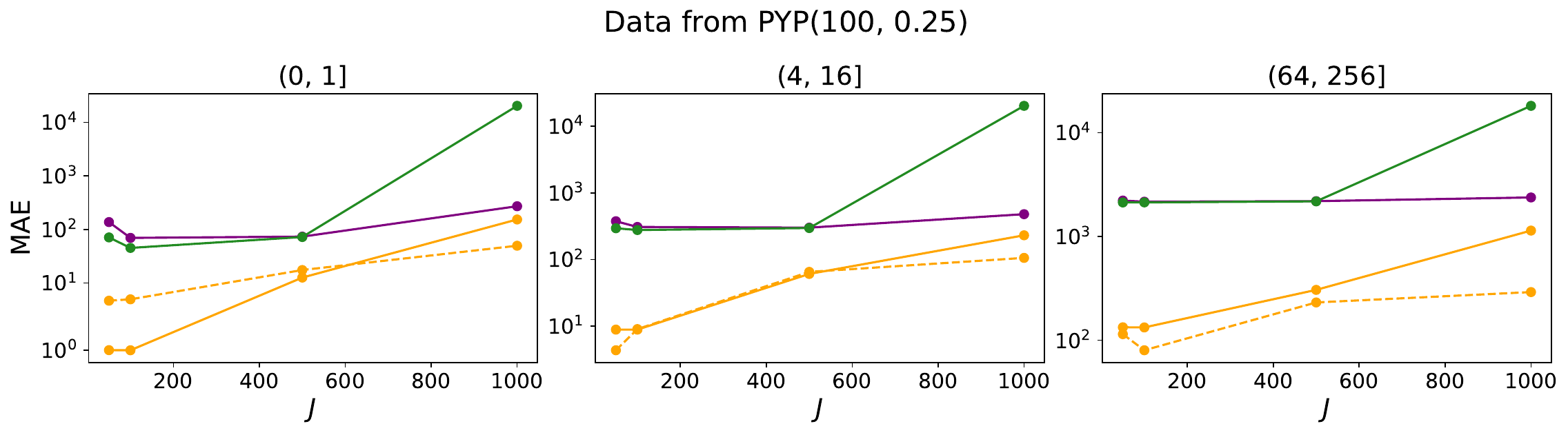}
     \includegraphics[width=\linewidth]{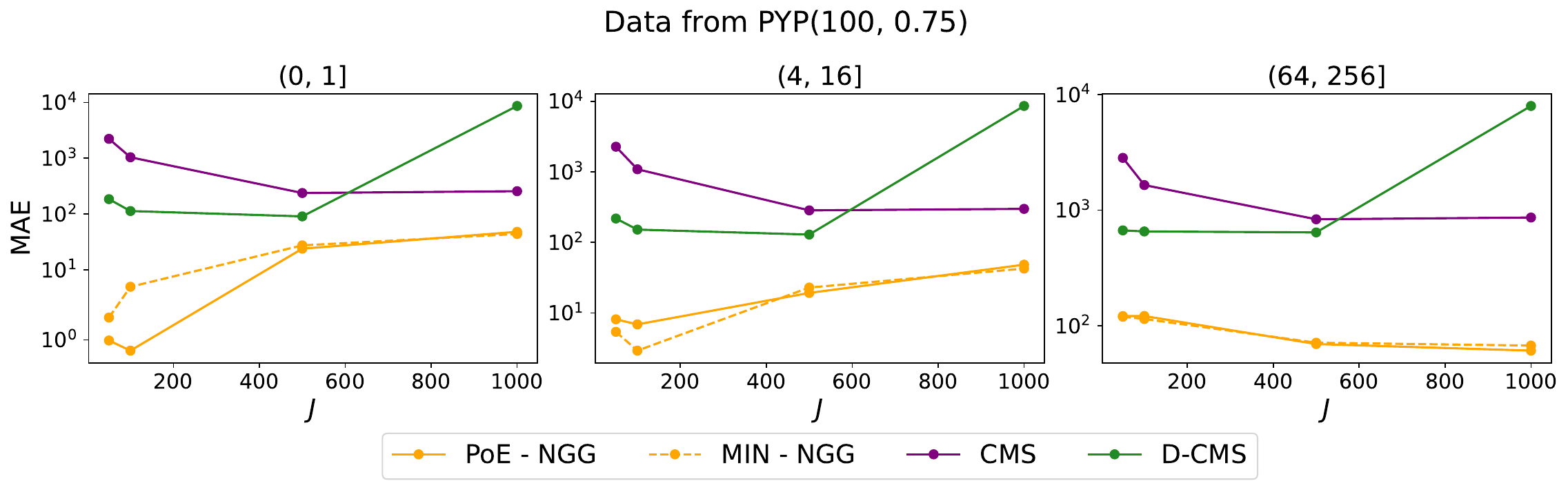}
    \caption{MAEs for the frequency estimators in Section~\ref{sec:ex_multihash}, stratified by true frequency bins. Top row, data generated from a Pitman-Yor process with parameters $(100, 0.25)$. Bottom row, data from a Pitman-Yor process with parameters $(100, 0.75)$.
    In the bottom row, the PoE-NGG and MIN-NGG lines are overlapping.}
    \label{fig:multiview}
\end{figure}

\subsection{Calibration of multi-hash aggregation rules}

Our methodology can produce prediction {\em intervals} for $f_{X_{n+1}}$. For instance, a 95\% prediction interval can be constructed using the 2.5\% and 97.5\% quantile of \eqref{eq:poe_post} or \eqref{eq:min_post}.
It is then crucial to evaluate both the calibration of these intervals and their width. 
\cite{Ses(22)} empirically demonstrated that the BNP posterior distribution for $f_{X_{n+1}}$ under a DP prior—originally obtained in \cite{Cai(18)} and coinciding with our PoE rule with DP smoothing—is often miscalibrated and results in excessively wide intervals.

The following empirical results will show that miscalibration is not unique to the Bayesian setting, but it is rather a common problem, one that affects  both of the aggregation rules that we proposed. 
We will call the intervals obtained using \eqref{eq:poe_post} and \eqref{eq:min_post} the ``smoothed-PoE'' and ``smoothed-min'' intervals, respectively.
We propose to overcome this pitfall using the conformal inference approach proposed by \cite{Ses(22)}.
In particular, we show that using the point estimators considered previously (i.e., obtained as the expected value of \eqref{eq:poe_post} or \eqref{eq:min_post}) to produce intervals as in \cite{Ses(22)} leads to shorter prediction intervals with valid coverage.
We will call such intervals the ``conformal-PoE'' and ``conformal-min'' intervals, respectively.
Observe that both the smoothed and conformal intervals depend on the choice of smoothing distribution.

We generate data from a PYP with parameters $(10, 0.25)$ and $(100, 0.25)$, simulating $n = 250,000$ observations. For the ``conformalized'' approach, we follow \citet{Ses(22)}, using their default tuning values. Specifically, we reserve the first $m = 25,000$ observations for calibration and assess coverage and interval length on an additional 2,500 data points. Figure~\ref{fig:multihash-calibration} presents results averaged over 50 independent replications.
All intervals exceed the nominal 0.9 coverage level due to the discreteness of the data, which includes many repeated observations of the same symbol. Achieving exact 90\% marginal coverage would require randomly perturbing the intervals corresponding to repeated queries, but we do not introduce such randomness here for simplicity. 

The average length of the intervals varies significantly across settings. The top plots in Figure~\ref{fig:multihash-calibration}, for $\mathrm{PYP}(10, 0.25)$, show that smoothed-PoE and smoothed-min intervals are relatively large, while conformal-PoE and conformal-min intervals are much shorter and of similar length. In contrast, the bottom plots, for $\mathrm{PYP}(100, 0.25)$, indicate that smoothed and conformal intervals have comparable average length, though the conformal calibration still provides a notable improvement.

\begin{figure}[htb]
    \centering
    \includegraphics[width=\linewidth]{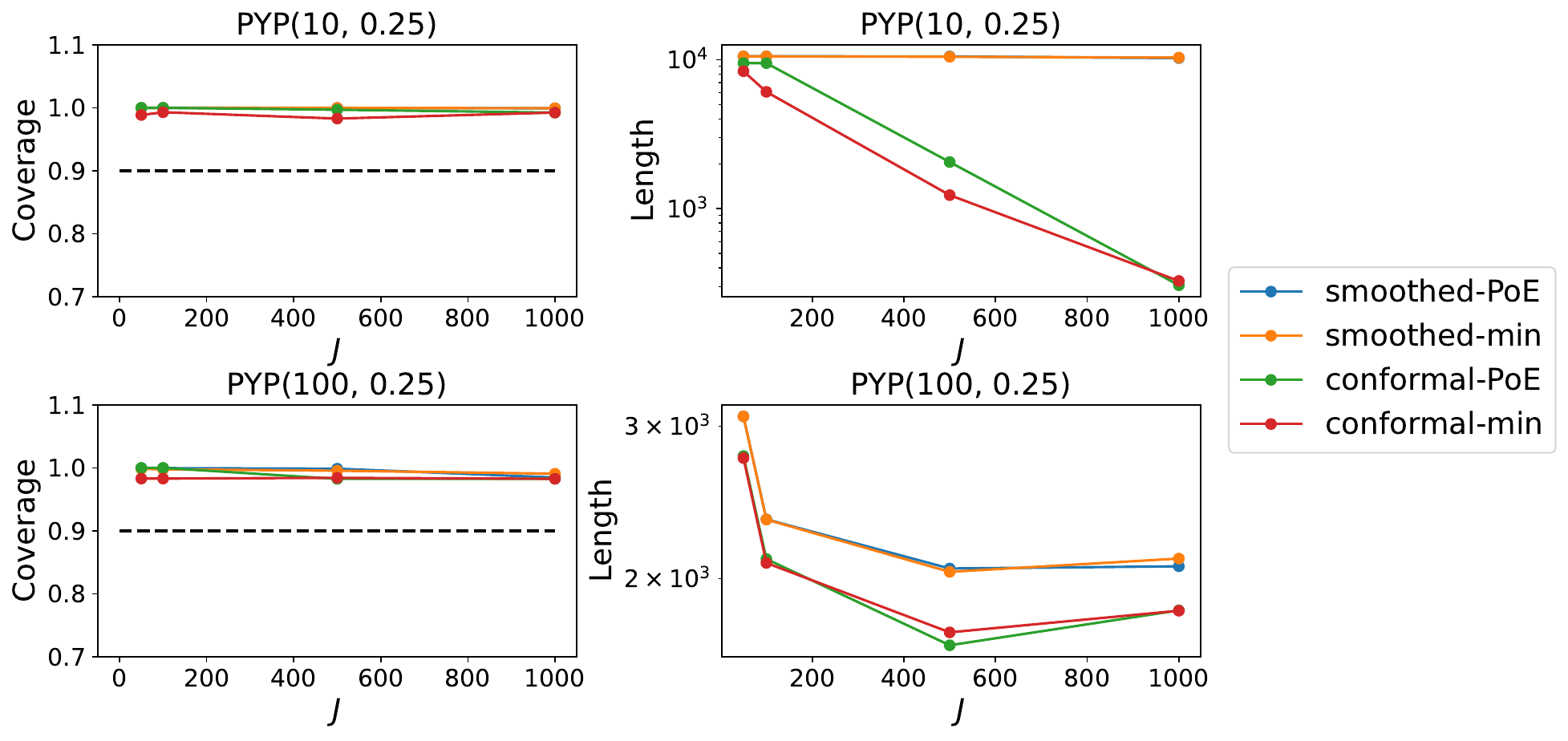}
    \caption{Calibration and average lengths of the intervals derived from the product of expert and the min aggregation rule using an NGGP smoothing for two different data generating processes, in experiments involving multiple independent hash functions.The results are shown as a function of the hash width, while the total memory budget is fixed.}
    \label{fig:multihash-calibration}
\end{figure}

\subsection{Real data sets}

We consider two real data sets, displaying remarkably different behaviours in terms of their frequency distribution, that were also analyzed in \cite{Ses(22)}.
The first data set consists of the texts of 18 open-domain classic pieces of English literature from the Gutenberg Corpus.
The frequency distribution here has a clear power-law behaviour. We subsample 600,000 bigrams from the corpus, displaying approximately 120,000 distinct combinations.
The second data set contains nucleotide sequences from SARS-CoV-2 viruses, made publicly available by the National Center for Biotechnology Information \citep{Hat(17)}.
These data include 43,196 sequences, each consisting of approximately 30,000 nucleotides. The goal is to estimate the empirical frequency of each possible 16-mer, a distinct sequence of 16 DNA bases found in contiguous nucleotides. Since each nucleotide has one of 4 bases, there are  $4^{16} \approx$ 4.3 billion possible 16-mers.
The frequency distribution of the DNA sequences is rather unusual: there are no ``common'' sequences, and most sequences appear approximately 1,000 times in the data set. 
There are also a few common sequences that appear 2000 times and some rare sequences that appear only a few times.
We subsample 2,000,000 data points, displaying approximately 20,000 unique sequences.

% First, we focus on the recovery of the number of distinct symbols in the data. To this end, we sketch the data with a random hash function of size 100,000.
% We consider the estimators based on the DP and NGGP smoothing and estimate the parameters as discussed above.
% Figure~\ref{fig:real_kn} shows the true end estimated $K_n$ for different values of $n$. For the bigrams data set, the NGGP smoothed estimator clearly outperforms the DP alternative and neatly resembles the true $K_n$. Instead, for the Covid-DNA data, both estimators overestimate $K_n$ and it is the DP that performs slightly better.

\begin{figure}[htb]
    \centering
    \includegraphics[width=\linewidth]{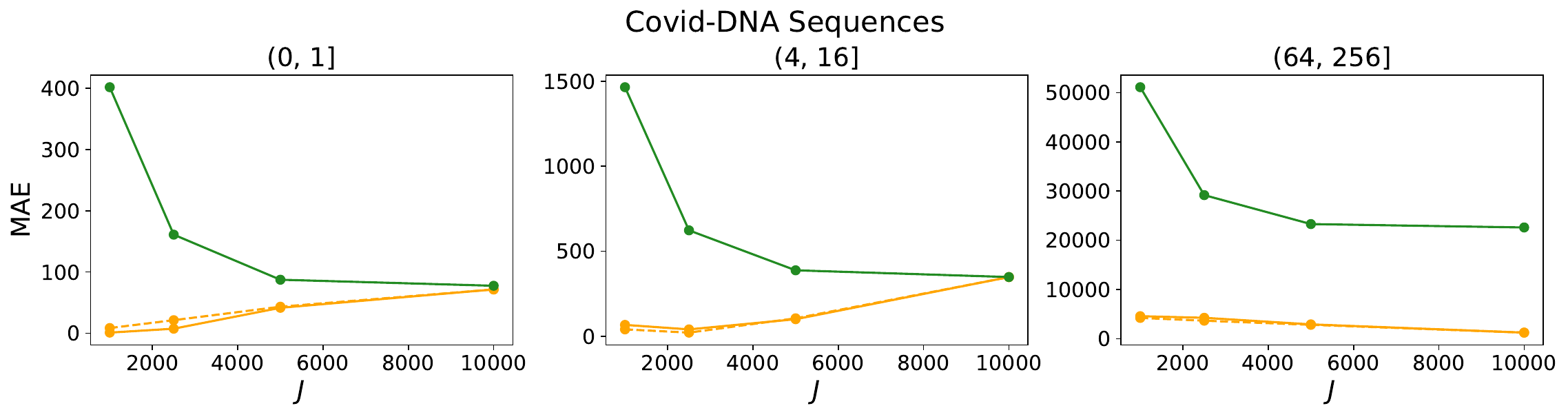}
    \includegraphics[width=\linewidth]{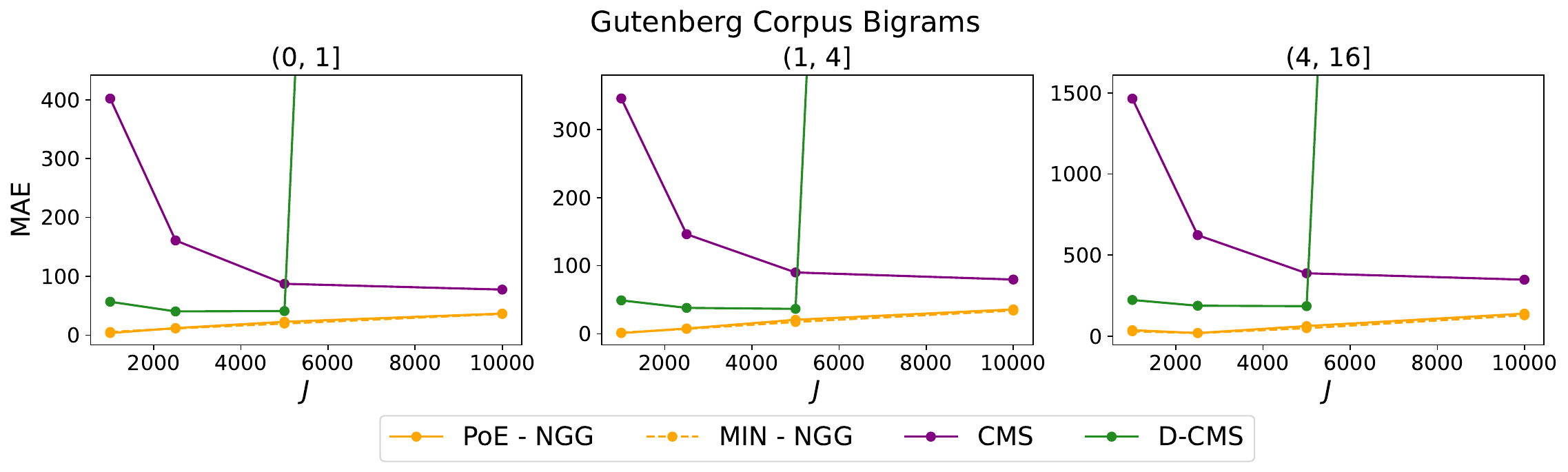}
    \caption{MAEs for the frequency estimators on the Covid-DNA sequences (top row) and Gutenberg corpus bigrams (bottom row), stratified by true frequency bins. MAEs of the CMS estimator are not reported for the Covid-DNA sequences as they are much higher.}
    \label{fig:real_maes}
\end{figure}

We follow the same setup of Section~\ref{sec:ex_multihash} and fix a memory budget of $M \times J =$ 10,000, letting $M = 10, 4, 2, 1$. We sketch the data and evaluate the  estimators based on the NGGP smoothing combined via the product of experts and the ``min'' rule. Figure~\ref{fig:real_maes} shows the MAEs stratified by frequency.
For the both the English bigrams and the Covid-DNA data, the behaviour is as expected: the NGGP achieves the best performance.
For comparison, in both datasets, the errors of the CMS algorithm are at least an order of magnitude larger than those of our smoothed estimators, which also outperform the D-CMS.

% \begin{figure}[!htb]
%     \centering
%     \includegraphics[width=\linewidth]{images/fig7.pdf}
%     \caption{From left to right: frequency distribution for the Gutenberg corpus bigrams, true and estimated $K_n$ for the bigrams dataset, frequency distribution for the Covid-DNA  dataset, true and estimated $K_n$ for the Covid-DNA dataset.}
%     \label{fig:real_kn}
% \end{figure}

%%%%%%%%%%%%%%%%%%%%%%%%%%%%%%%%
%%%%%%%%%%%%%%%%%%%%%%%%%%%%%%%%
%%%%%%%%%%%%%%%%%%%%%%%%%%%%%%%%
%%%%%%%%%%%%%%%%%%%%%%%%%%%%%%%%

\section{Discussion}\label{sec:discussion}

%We proposed a smoothed-Bayesian approach to frequency recovery from sketched data, inspired by existing BNP approaches but overcoming their practical limitations.

Although it was not discussed in the paper, our smoothed-Bayesian approach also lends itself to address different recovery problems, such as the estimation of the total number of distinct symbols in the sketched data. This is the \emph{cardinality recovery} problem, for which there exist algorithmic solutions \citep{Fla(85),Fla(07)} relying on different sketching algorithms, as well as solutions that rely on modeling assumptions for the data \citep{Cha(06),Che(11),ting2019approximate,Pet(21)}. Appendix~\ref{sec:cardinality} extends our results to address cardinality recovery using the same sketch obtained via random hashing. These additional results include both a worst-case theoretical analysis and a novel class of smoothed-Bayesian estimators.
We believe these extensions may be of interest especially in the context of privacy-preserving analyses.

This work opens several opportunities for future research. For example, one may conduct an in-depth analysis for the problem of endowing our method with uncertainty estimates, going beyond current conformal inference approaches \citep{Ses(23)}.
Further, one may consider the more general setting in which data belong to more than one symbol, referred to as traits, and exhibit nonnegative integer levels of associations with each trait; e.g., single cell data containing multiple genes with their expression levels, or documents containing different topics with their words.
In the trait allocation setting, a BNP approach to frequency recovery has been developed in \citet{Cai(18)} and \citet{Ber(23)}, showing some computational issues in the evaluation of posterior distributions. Extending smoothed estimation to the trait allocation setting may lead to more computationally efficient estimators compared to the BNP approach, under flexible modeling assumptions.

Another research direction could involve combining smoothed estimation with ``learning-based'' hashing algorithms \citep{aamand2024improved}, which leverage additional data features (not considered in this paper) to identify the most common symbols and hash them separately from the rest of the dataset.
Further, one could explore smoothed estimation in large-scale streaming and distributed settings.
This may require smoothing distributions that can describe non-exchangeable data \citep{airoldi2014generalized}, adapting our method to allow for a sequential estimation of smoothing parameters, and extending our multi-view formulation to account for the cost of aggregating inference across different servers under communication constraints. Moreover, our framework may be extended to problems involving more complex data such as graphs \citep[][Chapter 7]{Cor(20)}, for which several BNP models have been proposed \citep{Car(17), Ric(22)}.

% as well as for more general forms of randomized sketching used for other numerical, statistical, and
% learning problems \citep{Cor(20)}.

\if1\blind
\section*{Software Availability}
The code implementing the methods described in this paper, as well as scripts to reproduce the numerical experiments and data analysis, is publicly available at: \url{https://github.com/mberaha/SmoothedSketching}.

\section*{Acknowledgements}
M.B.~and S.F.~were supported by the European Research Council (Horizon 2020, grant 817257).
M.B.~was also supported by MUR, grant Dipartimento di Eccellenza 2023-2027, and S.F.~by the Italian Ministry of Education, University and Research (MIUR), ``Dipartimenti di Eccellenza" grant 2018-2022. S.F. is also affiliated with IMATI-CNR ``Enrico  Magenes" (Milan, Italy).
M.S.~was supported by NSF grant DMS 2210637 and by an Amazon Research Award.
The authors are grateful to the editors and anonymous referees for their valuable suggestions, which helped improve an earlier version of this manuscript.
\else
\section*{Acknowledgements}
The authors are grateful to the editors and anonymous referees for their valuable suggestions, which helped improve an earlier version of this manuscript.
\fi
\textbf{Generative AI Use:}
 GPT-4o was used for language improvement.

\section*{Disclosure Statement}
The authors report there are no competing interests to declare.

%%% Local Variables:
%%% mode: latex
%%% TeX-master: "main_jasa"
%%% End:

%%%%%%%%%%%%%%%%%%%%%%%%%%%%%%%%
%%%%%%%%%%%%%%%%%%%%%%%%%%%%%%%%
%%%%%%%%%%%%%%%%%%%%%%%%%%%%%%%%
%%%%%%%%%%%%%%%%%%%%%%%%%%%%%%%%
%%%%%%%%%%%%%%%%%%%%%%%%%%%%%%%%

\bibliographystyle{chicago}
\bibliography{references.bib}

%%%%%%%%%%%%%%%%%%%%%%%%%%%%%%%%
%%%%%%%%%%%%%%%%%%%%%%%%%%%%%%%%
%%%%%%%%%%%%%%%%%%%%%%%%%%%%%%%%
%%%%%%%%%%%%%%%%%%%%%%%%%%%%%%%%
%%%%%%%%%%%%%%%%%%%%%%%%%%%%%%%%
\appendix
\section{Further details about related prior work} \label{relatedwork}

\subsection{Review of the CMS algorithm} \label{app-cms}

Unlike the smoothed estimation framework discussed in this paper or the BNP framework reviewed in Appendix~\ref{relatedwork-bnp}, the classical CMS algorithm \citep{Cor(05)} treats the data $(x_{1},\ldots,x_{n})$ as arbitrarily fixed. Therefore, one may say that the CMS is not, strictly speaking, a statistical approach.
Instead, the CMS computes frequency estimates by leveraging only the randomness in the hash functions, as reviewed below.

For $M\geq1$ let $h_{1},\ldots,h_{M}$ be $J$-wide random hash functions, with $h_{l}:\mathbb{S}\rightarrow\{1,\ldots,J\}$ for $l \in [M]$, that are i.i.d. from the pairwise independent hash family $\mathscr{H}_{J}$ \citep[Chapter 5 and Chapter 15]{Mit(17)}. 
Sketching $(x_{1},\ldots,x_{n})$ through $h_{1},\ldots,h_{M}$ produces a random matrix $\mathbf{C}_{M,J}\in\mathbb{N}_{0}^{M\times J}$ whose $(l,j)$-th bucket is
\begin{displaymath}
C_{l,j}=\sum_{i=1}^{n}I(h_{l}(x_{i})=j),
\end{displaymath}
for each $l \in [M]$ and $j \in [J]$, such that $\sum_{1\leq j\leq J}C_{l,j}=n$. More precisely, each $h_{l}$ maps each $x_{i}$ into one of the $J$ buckets, defining the sketch $\mathbf{C}_{M,J}$ whose $(l,j)$-th element $C_{l,j}$  is the number of $x_{i}$'s hashed by $h_{l}$ in the j-th bucket. 

Based on $\mathbf{C}_{M,J}$, the CMS estimates $f_{x_{n+1}}$ by taking the smallest count among the $M$ buckets into which $x_{n+1}$ is hashed, 
\begin{equation}\label{eq_cms}
\hat{f}^{\text{\tiny{(CMS)}}}_{x_{n+1}}=\min\{C_{1,h_{1}(x_{n+1})},\ldots,C_{M,h_{M}(x_{n+1})}\}.
\end{equation}
As $\hat{f}_{v}^{\text{\tiny{(CMS)}}}$ is the count $C_{l,h_{k}(x_{n+1})}$ associated with the hash function with the fewest collisions, it provides a deterministic upper bound on the true $f_{x_{n+1}}$: $$ \hat{f}_{x_{n+1}}^{\text{\tiny{(CMS)}}}\geq f_{x_{n+1}}.$$
Further, if $J=\left \lceil{e/\varepsilon}\right \rceil $ and $M=\left \lceil{\log 1/\delta}\right \rceil$, for any $\varepsilon>0$ and $\delta>0$, then $\hat{f}_{x_{n+1}}^{\text{\tiny{(CMS)}}}$ also satisfies
$$\hat{f}_{x_{n+1}}^{\text{\tiny{(CMS)}}}\leq f_{x_{n+1}}+\varepsilon m,$$
 with probability at least $1-\delta$ over the randomness in the hash functions.
These results provide a theoretical guarantee for $\hat{f}^{\text{\tiny{(CMS)}}}_{x_{n+1}}$ in terms of a confidence interval, but they often tend to be too conservative to be useful in practice if the data are randomly sampled from some distribution instead of fixed arbitrarily \citep{ting2018count,Ses(23)}.
We refer to \citet[Chapter 3]{Cor(20)} for a more comprehensive account on the CMS, and generalizations thereof.

The interplay between our approach and the CMS algorithm emerges through our nonparametric analysis.
In particular, we find that the worst-case estimator of $f_{X_{n+1}}$ derived in Section~\ref{sec:np-estim} coincides exactly with the original CMS upper bound.

\subsection{Review of the BNP approach} \label{relatedwork-bnp}

This section reviews the BNP approach of \citet{Cai(18)}, which proposed first to address our frequency estimation problem by assuming a Dirichlet process prior for the data distribution.
This BNP framework was later extended by \citet{Dol(21)}, \citet{Dol(23)}, and \citet{Ber(23)} to accommodate a broader class of prior distributions belonging to the NRM class.
Here, we present a unified summary of the BNP approach by following the notation of \citet{Ber(23)}, and then we discuss the key computational challenges arising when the BNP approach is applied with prior distributions other than the Dirichlet process prior.

\subsubsection{Main assumptions and approach}

The BNP framework rests on the two key assumptions.
Firstly, the data $(x_{1},\ldots,x_{n})$ are modeled as a random sample $\mathbf{X}_{n}=(X_{1},\ldots,X_{n})$ from an unknown distribution $P=\sum_{s\in\mathbb{S}}p_{s}\delta_{s}$ on $\mathbb{S}$, endowed by a nonparametric prior $\mathscr{P}$.
Secondly, the hash family $\mathscr{H}_{J}$ is independent of $\mathbf{X}_{n}$. 
That is, for any $n\geq1$,
\begin{align}\label{eq:exchangeable_model_hash2} 
  \begin{split}
    P&\,\sim\,\mathscr{P}, \\
    X_1,\ldots,X_{n} \mid P &\,\simiid\, P, \\
    h&\,\sim\,\mathscr{H}_{J},\\
    C_{j} &\,=\, \sum_{i=1}^{n} I( h(X_i) = j), \quad \forall j \in \{1,\ldots,J\}.
  \end{split}
\end{align}
  
\citet{Ber(23)} studied the estimation of the empirical frequency $f_{X_{n+1}}$ under the model in \eqref{eq:exchangeable_model_hash2} with $P\sim\text{NRM}(\theta,G_{0},\rho)$, extending the prior work of \citet{Cai(18)} which focused on the sub-family of DP priors.
In particular, Theorem 2.2 in \citet{Ber(23)} provides the posterior distribution of $f_{X_{n+1}}$, given the sketch $\mathbf{C}_{J}$ and the bucket $h(X_{n+1})$ in which $X_{n+1}$ is hashed, namely
\[
 \pi_j(r) = \frac{\E_{P \sim \mathscr P, h \sim \mathscr H} \prob[f_{X_{n+1}} = r, \mathbf C_j = \mathbf c, h(X_{n+1}) = j \mid P, h] }{\E_{P \sim \mathscr P, h \sim \mathscr H} \prob[\mathbf C_j = \mathbf c, h(X_{n+1}) = j \mid P, h] },
\]
for $r=0, \ldots, c_j$, whose expected value,
\begin{equation}\label{bnp_est}
\varepsilon_{f}(\mathscr P)= \sum_{r=0}^{c_j} r \pi_j(r)
\end{equation}
is a BNP estimator $f_{X_{n+1}}$ under a squared loss function. 

In particular, if $\phi^{(n)}(u) = (-1)^{n} \frac{\mathrm{d}^n}{\mathrm{d} u^n} e^{-\theta/J \psi(u)}$, where $\psi(u) = \int_{\R_+}(1 - e^{-us}) \rho(s) \mathrm{d} s$, and  $\kappa(u, n) = \int_{\R_+} e^{-us} s^n \rho(s) \mathrm{d} s$, it holds that
\[
 \pi_j(r) = \frac{\theta}{J} \binom{c_j}{r} \frac{\int_{\R_+} u^{n}  \phi^{(c_j - l)}(u)
 \prod_{k\neq j} \phi^{(c_k)}(u)  \, \kappa(u, r+1) \mathrm{d} u
 }{
  \int_{\R_+ } u^{n}   \phi^{(c_j + 1)}(u)
  \prod_{k\neq j} \phi^{(c_k)}(u)  \mathrm{d} u}.
\]
An application of the Faa di Bruno formula shows that, in general, the terms $\phi^{(c_k)}(u)$ involve complex combinatorial objects, and, more importantly, computing products of the kind $\prod_{k} \phi^{(c_k)}(u)$ require a summation over a discrete index set whose size grows exponentially in $J$ and $c_j$, hence, also in the sample size. 
An exception is the case of $P \sim \text{DP}(\theta, G_0)$, in which case the BNP estimator in \eqref{bnp_est} coincides with the learning-augmented CMS of \citet{Cai(18)}, namely $\varepsilon_{f}(\text{DP}(\theta, G_0)) = c_j J / (\theta + J)$, which depends on the observed sketch $\mathbf{C}_{J}$ only through $C_{h(X_{n+1})}$, namely the size of the bucket in which $X_{n+1}$ is hashed.
As shown in \citet{Ber(23)} this property uniquely characterizes the Dirichlet process.

% Moreover, it is shown in \citet{Ber(23)} that: i) the BNP estimator $\varepsilon_{f}(\text{DP}(\theta,G_{0}))$ coincides with the learning-augmented CMS of \citet{Cai(18)}; ii) the DP prior is the sole NRMs for which the BNP estimator \eqref{bnp_est} depends on the sketch $\mathbf{C}_{J}$ only through $C_{h(X_{n+1})}$, namely the size of the bucket in which $X_{n+1}$ is hashed. 

% The interplay between our approach and the BNP approach emerges through the framework of smoothed estimation in Section~\ref{sec:smoothed}, which also relies on distributional assumptions on $P$. 
% In the BNP framework, the (prior) assumption on $P$ is applied to obtain estimators as expected values under the posterior of $P$ given $\mathbf{X}_{n}$.
% Instead, in the framework of smoothed estimation, the  (smoothing) assumption on $P$ is applied to obtain estimators as expected values under the distribution of $P$. From \citet{Ber(23)} it emerges how priors beyond the DP, such as NRM priors, lead to BNP estimators that involve non-trivial computational challenges. 
% By contrast, smoothed estimation leads to simple estimators for any smoothing assumption in the class of NRMs. 

\subsubsection{The computational complexity of the BNP approach}\label{app:bnp_cost}

Consider the Bayesian model in \eqref{eq:exchangeable_model_hash2} where $P \sim \text{PYP}(\theta, \alpha)$ is distributed as a Pitman-Yor process with parameters $\theta$ and $\alpha$. 

The \emph{full-sketch} posterior found in \cite{Ber(23)} is as follows
\begin{align}\label{post_pyp}
\pi_j^F(r) = \frac{\theta}{J} \binom{c_j}{r} (1 - \alpha)_{(r)} \frac{\sum_{\bm{i} \in \mathcal S(\bm{c}, j, -r)} \frac{\Gamma\left(\frac{\theta + \alpha}{ \alpha} + |\bm{i}|\right)}{J^{|\bm{i}|}} \prod_{k=1}^J \mathscr{C}(c_k - r \delta_{k,j}, i_k; \alpha)}{\sum_{\bm{i} \in \mathcal S(\bm{c}, j, 1)}\frac{ \Gamma\left(\frac{\theta}{\alpha} + |\bm{i}|\right)}{J^{|\bm{i}|}}  \prod_{k=1}^J \mathscr{C}(c_k + \delta_{k,j}, i_k; \alpha)},
\end{align}
where ($i$) $\mathcal S(\bm{c}, j, q)$ is the Cartesian product $\times_{1\leq k\leq J}\{0, \ldots, c_k + \delta_{k, j} q\}$, ($ii$) $\delta_{k,j}$ is the Kronecker delta, $(iii)$ $|\bm{i}| = \sum_{1\leq k\leq J} i_k$, and $(iv)$ $\mathscr{C}(n,k;\alpha)$ denotes the generalized factorial of $n$ of order $k$. That is, for $n\geq0$ and $0\leq k\leq n$, 
\begin{equation}\label{eq:gen_fac}
	\mathscr{C}(n,k;\alpha)=\frac{1}{k!}\sum_{i=0}^{k}(-1)^{i}{n\choose i}(-i\alpha)_{(n)}
\end{equation}
with $\mathscr{C}(0,0;\alpha):=0$, and  $\mathscr{C}(n,0;\alpha):=1$ \citep[Chapter 2]{Cha(05)}.

Observe that the index set $\mathcal S(\bm{c}, j, q)$ has a size that grows as $\mathcal O(n^J)$, making the evaluation of \eqref{post_pyp} unfeasible even for small datasets. As an example, computing the expectation of \eqref{post_pyp} when $n=50$ and $J=5$ takes around 20 minutes on a standard laptop. 

Consider now the \emph{single-bucket} posterior given in \cite{Dol(23)}, namely
\begin{equation}\label{eq:py_single}
    \pi^S_j(r) = \frac{\theta}{J}  \binom{c_j}{r} (1 - \alpha)_{(r)} 
\frac{\sum_{i=0}^{c_j - r} \sum_{j=0}^{n - c_j} \left( \frac{\theta + \alpha}{\alpha} \right)_{(i + j)} \left( \frac{1}{J} \right)^i \left(1 - \frac{1}{J} \right)^j \mathscr{C}(c_j - r, i; \alpha) \mathscr{C}(n - c_j, j; \alpha)}
{\sum_{i=0}^{c_j + 1} \sum_{j=0}^{n - c_j} \left( \frac{\theta}{\alpha} \right)_{(i + j)} \left( \frac{1}{J} \right)^i \left( 1 - \frac{1}{J} \right)^j \mathscr{C}(c_n + 1, i; \alpha) \mathscr{C}(n - c_j, j; \alpha)}.
\end{equation}
Evaluating \eqref{eq:py_single} for a single $r$ requires $\mathcal O\left( c_j n\right) = \mathcal O\left(n^2\right)$ operations, which is considerably faster than $\mathcal O(n^J)$.
However, the double summations in \eqref{eq:py_single} involve the generalized factorial coefficients. Computing them directly using \eqref{eq:gen_fac} is unfeasible in all practical setting. 
Hence, the strategy is to pre-compute them using the recursive relation $\mathscr{C}(n+1, k+1) = \alpha \mathscr{C}(n, k) + (n - k\alpha) \mathscr{C}(n+1, k)$. 
Storing all the generalized factorial coefficients requires $\mathcal O(n^2)$ memory. For instance, if $n=100000$, this takes around 80GBs of memory, much more than storing all the dataset.
Even assuming that storing the generalized factorial coefficients is not an issue, evaluating the expectation of \eqref{eq:py_single} is extremely slow. 
On a standard laptop, one single evaluation of this expectation takes several minutes.
By contrast, in the experiments in Section~\ref{sec:ex_singlehash}, evaluating the $\hat f_{X_{n+1}}$ under the NGGP smoothing via \eqref{eq:prop-nggp} takes 4 milliseconds, while computing the full distribution $(\pi_j(r))_{r \geq 0}$ in \eqref{mass_ngg} takes 18 milliseconds. 

\subsubsection{Approximate inference for BNP frequency recovery} \label{app:approximate-BNP}

The discussion in Appendix~\ref{app:bnp_cost} rules out the use of the exact Bayesian posterior for frequency recovery due to its impractically high computational cost.
 However, one may ask whether approximate inference strategies based on Markov chain Monte Carlo (MCMC) could still be a viable option. Here, we discuss the theoretical challenges associated with MCMC algorithms in the context of frequency recovery from sketched data.
For the purposes of this discussion, it is sufficient to consider the hash function fixed and given.

A first first approach could be to consider the model 
\begin{equation}\label{eq:model1}
    \begin{aligned}
        C_1, \ldots, C_J \mid P(\mathbb S_1), \ldots, P(\mathbb S_J) & \sim \text{Multinomial}(n; P(\mathbb S_1), \ldots, P(\mathbb S_J)), \\
        P(\mathbb S_1), \ldots, P(\mathbb S_J) & \sim \Pi(\mathbb S_1, \ldots, \mathbb S_J),
    \end{aligned}
\end{equation}
where $\Pi(\mathbb S_1, \ldots, \mathbb S_J)$ denotes the finite-dimensional probability distribution of the random probability measure $P \sim \mathscr P$.
It would seem that posterior inference can be carried easily for $(P(\mathbb S_1), \ldots, P(\mathbb S_J))$. 
However, the probability distribution $\Pi$ is generally intractable. 
For instance, the expression of $\Pi$ for a Pitman--Yor process found in \cite{sangalli2006} involve summations over all possible partitions of $\{1, \ldots, n\}$.
The only notable exceptions are the case of the Dirichlet process \citep{Fer(73)} and the normalized inverse-Gaussian process \citep[NIGP,][]{Lij(05)}. Of course, in the case of the Dirichlet process prior, MCMC would be superfluous as posterior inference is available in closed form.
\cite{Dol(21)} employed the NIGP for the frequency recovery problem specifically for the availability of the finite-dimensional distributions and computed a closed-form expression for the single-bucket posterior. However, as noted in \cite{Dol(23)}, the NIGP is not flexible enough to model different power-law behaviors.
Furthermore, having posterior samples from $(P(\mathbb S_1), \ldots, P(\mathbb S_J))$ is not sufficient to answer the frequency recovery problem, as $f_{X_{n+1}}$ depends also on $P(\{X_{n+1}\})$.

A second approach is to introduce $X_1, \ldots, X_n$ as latent variables and consider the model
\begin{equation}\label{eq:model2}
    \begin{aligned}
        C_j &= \sum_{n = 1}^J I[h(X_n) = j], \quad j=1, \ldots, J  \\
        X_i \mid P & \simiid P, \quad i=1, \ldots, n \\
        P & \sim \mathscr P
    \end{aligned}
\end{equation}
While introducing the $X_i$'s in the MCMC state allows to compute $f_{X_{n+1}}$, inference under model \eqref{eq:model2} does not appear to be cheaper than under \eqref{eq:model1}. Indeed, marginalizing out $P$ from \eqref{eq:model2}, we can sample the $X_i$'s using a generalized Chinese restaurant process. However, then each iteration of the MCMC scales as $\mathcal{O}(n k_n)$ where $k_n$ is the number of distinct symbols in the $X_i$'s, which could be of the order of $n^\alpha$ for some $\alpha \in (0, 1)$. Hence, each MCMC iteration would scale super-linearly with $n$, making this approach unfeasible.

%%%%%%%%%%%%%%%%%%%%%%%%%%%%%%%%
%%%%%%%%%%%%%%%%%%%%%%%%%%%%%%%%
%%%%%%%%%%%%%%%%%%%%%%%%%%%%%%%%
%%%%%%%%%%%%%%%%%%%%%%%%%%%%%%%%

\section{Further details on frequency recovery}\label{app:frequency-recovery}

\subsection{Proofs for Section~\ref{sec:setup}}

\subsubsection{Proof of Theorem~\ref{teo_cond_prob}}\label{proof_teo_cond_prob}
Recall that under the statistical model \eqref{eq:exchangeable_model_hash}, $\mathbf{X}_{n}$ is a random sample from the distribution $P=\sum_{s\in\mathbb{S}}p_{s}\delta_{s}$, with $p_{s}\in(0,1)$ being the probability of $s\in\mathbb{S}$, 
% and the random hash function $h:\mathbb{S}\rightarrow\{1,\ldots,J\}$ is independent of $\mathbf{X}_{n}$.
and the hash function $h:\mathbb{S}\rightarrow\{1,\ldots,J\}$ is fixed.
Under these assumptions, for any $s\in\mathbb{S}$
\begin{displaymath}
\text{Pr}[X_{i}=s \mid h(X_{i})=j]=\frac{p_{s}}{q_{j}}.
\end{displaymath}
That is, if $\mathbb{S}_{j}=\{s\in\mathbb{S}\text{ : }h(s)=j\}$, with $j=1,\ldots,J$, then the random variables $X_{i}$'s that are hashed into the $j$-the bucket are independent and identically distributed according to
\begin{displaymath}
P_{j}=\sum_{s\in\mathbb{S}_{j}}\frac{p_{s}}{q_{j}}\delta_{s}.
\end{displaymath}
Now, we evaluate the conditional distribution of $f_{X_{n+1}}$, given $\mathbf{C}_{J}$ 
% and $h(X_{n+1})$, 
 $h$
i.e, for $r=0,1,\ldots,c_{j}$
\begin{equation}\label{eq_cond_probab}
\text{Pr}[f_{X_{n+1}}=r \mid \mathbf{C}_{J}=\mathbf{c},h(X_{n+1})=j, h]=\frac{\text{Pr}[f_{X_{n+1}}=r,\mathbf{C}_{J}=\mathbf{c},h(X_{n+1})=j   \mid h]}{\text{Pr}[\mathbf{C}_{J}=\mathbf{c},h(X_{n+1})=j \mid h]}.
\end{equation}
We first evaluate the denominator of \eqref{eq_cond_probab}, which is trivial, and then consider the numerator. Since $\mathbf{X}_{n}$ is a random sample from $P$, the sketch $\mathbf{C}_{J}$ is distributed as a Multinomial distribution with parameter $(n,q_{1},\ldots,q_{J})$, where $q_{j}:=\text{Pr}[h(X_{i})=j]$ for $j=1,\ldots, J$. Then,
\begin{align}\label{eq_cond_probab_num}
&\text{Pr}[\mathbf{C}_{J}=\mathbf{c},h(X_{n+1})=j  \mid h]={n\choose c_{1},\ldots,c_{J}}q_{1}\cdots q_{j-1}^{c_{j-1}}q_{j}^{c_{j}+1}q_{j+1}^{c_{j}}\cdots q_{J}^{c_{J}}.
\end{align}
Now, we evaluate the numerator of \eqref{eq_cond_probab}. To such a purpose, it is useful to define the following event: $B(n,r)=\{X_{1}=\cdots=X_{r}=X_{n+1},\{X_{r+1},\ldots,X_{n}\}\cap\{X_{n+1}\}=\emptyset\}$. Then,
\begin{align*}
&\text{Pr}\left[f_{X_{n+1}}=r,\mathbf{C}_{J}=\mathbf{c},h(X_{n+1})=j  \mid h\right]\\
&\quad=\text{Pr}\left[f_{X_{n+1}}=r,\sum_{i=1}^{n}I(h(X_{i})=1)=c_{1},\ldots,\sum_{i=1}^{n}I(h(X_{i})=J)=c_{J},h(X_{n+1})=j  \mid h\right]\\
&\quad={n\choose r}\sum_{s\in\mathbb{S}}\text{Pr}\left[B(n,r),X_{n+1}=s,\sum_{i=1}^{n}I(h(X_{i})=1)=c_{1},\ldots\right.\\
&\quad\quad\left.\ldots,\sum_{i=1}^{n}I(h(X_{i})=J)=c_{J},h(X_{n+1})=j  \mid h\right]\\
&\quad={n\choose r}\sum_{s\in\mathbb{S}}\text{Pr}\left[B(n,r),X_{n+1}=s,\sum_{i=r+1}^{n}I(h(X_{i})=1)=c_{1},\ldots\right.\\
&\quad\quad\ldots\left.,\sum_{i=r+1}^{n}I(h(X_{i})=j)=c_{j}-r,\ldots\right.\\
&\quad\quad\ldots\left.,\sum_{i=r+1}^{n}I(h(X_{i})=J)=c_{J},h(X_{n+1})=j  \mid h\right].
\end{align*}
Accordingly, the distribution of $(f_{X_{n+1}},\mathbf{C}_{J},h(X_{n+1}))$ is completely determined by the distribution of $(X_{1},\ldots,X_{n},X_{n+1})$. Therefore, from the last equation, we can write the following
\begin{align*}
&\text{Pr}[f_{X_{n+1}}=r,\mathbf{C}_{J}=\mathbf{c},h(X_{n+1})=j   \mid h]\\
&\quad={n\choose r}\sum_{s\in\mathbb{S}}\text{Pr}\left[X_{1}=\cdots=X_{r}=X_{n+1}=s,X_{r+1}\neq s,\ldots,X_{n}\neq s\right]\\
&\quad\quad\times\text{Pr}\left[\sum_{i=r+1}^{n}I(h(X_{i})=1)=c_{1},\ldots,\sum_{i=r+1}^{n}I(h(X_{i})=j)=c_{j}-r,\ldots\right.\\
&\quad\quad\quad\left.\ldots,\sum_{i=r+1}^{n}I(h(X_{i})=J)=c_{J},h(X_{n+1})=j \mid  h, X_{r+1}\neq s,\ldots,X_{n}\neq s,X_{n+1}=s\right]\\
&\quad={n\choose r}\sum_{s\in\mathbb{S}}p_{s}^{r+1}(1-p_{s})^{n-r}I(s\in\mathbb{S}\text{ : }h(s)=j)\\
&\quad\quad\times\text{Pr}\left[\sum_{i=r+1}^{n}I(h(X_{i})=1)=c_{1},\ldots,\sum_{i=r+1}^{n}I(h(X_{i})=j)=c_{j}-r,\ldots\right.\\
&\quad\quad\quad\left.\ldots,\sum_{i=r+1}^{n}I(h(X_{i})=J)=c_{J}\mid  h, X_{r+1}\neq s,\ldots,X_{n}\neq s,X_{n+1}=s\right].
\end{align*}
Based on the last identity, we can write the conditional probability of $f_{X_{n+1}}$, given $\mathbf{C}_{J}$ and $h(X_{n+1})$, with respect to the distribution $P_{j}$ on $\mathbb{S}_{j}$. In particular, we can write the following:
\begin{align}\label{eq_cond_probab_den}
&\notag\text{Pr}[f_{X_{n+1}}=r,\mathbf{C}_{J}=\mathbf{c},h(X_{n+1})=j  \mid h]\\
&\notag \quad = {n\choose r}\sum_{s\in\mathbb{S}_{j}}p_{s}^{r+1}(1-p_{s})^{n-r}\\
&\notag\quad\quad\times{n-r\choose c_{1},\ldots,c_{j}-r,\ldots,c_{J}}\left(\frac{q_{1}}{1-p_{s}}\right)^{c_{1}}\cdots\left(\frac{q_{j}-p_{s}}{1-p_{s}}\right)^{c_{j}-r}\cdots\left(\frac{q_{J}}{1-p_{s}}\right)^{c_{J}}\\
&\notag\quad={n\choose r}\sum_{s\in\mathbb{S}_{j}}p_{s}^{r+1}\\
&\notag\quad\quad\times{n-r\choose c_{1},\ldots,c_{j}-r,\ldots,c_{J}}q_{1}^{c_{1}}\cdots(q_{j}-p_{s})^{c_{j}-r}\cdots q_{J}^{c_{J}}\\
&\notag\quad={n\choose r}\sum_{s\in\mathbb{S}_{j}}q_{j}^{c_{j}-r}p_{s}^{r+1}\\
&\notag\quad\quad\times{n-r\choose c_{1},\ldots,c_{j}-r,\ldots,c_{J}}q_{1}^{c_{1}}\cdots\left(1-\frac{p_{s}}{q_{j}}\right)^{c_{j}-r}\cdots q_{J}^{c_{J}}\\
&\notag\quad={n\choose r}\sum_{s\in\mathbb{S}_{j}}q_{j}^{c_{j}+1}\left(\frac{p_{s}}{q_{j}}\right)^{r+1}\\
&\notag\quad\quad\times{n-r\choose c_{1},\ldots,c_{j}-r,\ldots,c_{J}} q_{1}^{c_{1}}\cdots\left(1-\frac{p_{s}}{q_{j}}\right)^{c_{j}-r}\cdots q_{J}^{c_{J}}\\
&\notag\quad={n\choose r}\frac{{n-r\choose c_{1},\ldots,c_{j}-r,\ldots,c_{J}}}{{c_{j}\choose r}}q_{j}^{c_{j}+1}\prod_{j \in \{1, \ldots, J\} \setminus \{i\}}q_{i}^{c_{i}}\sum_{s\in\mathbb{S}_{j}}{c_{j}\choose r}\left(\frac{p_{s}}{q_{j}}\right)^{r+1}\left(1-\frac{p_{s}}{q_{j}}\right)^{c_{j}-r}\\
&\quad={n\choose c_{1},\ldots,c_{J}} q_{j}^{c_{j}+1} \prod_{j \in \{1, \ldots, J\} \setminus \{i\}} q_{i}^{c_{i}}\sum_{s\in\mathbb{S}_{j}}{c_{j}\choose r}\left(\frac{p_{s}}{q_{j}}\right)^{r+1}\left(1-\frac{p_{s}}{q_{j}}\right)^{c_{j}-r}.
\end{align}
Therefore, according to \eqref{eq_cond_probab}, which is combined with \eqref{eq_cond_probab_num} and \eqref{eq_cond_probab_den}, we obtain the following:\begin{align*}
&\text{Pr}[f_{X_{n+1}}=r \mid \mathbf{C}_{J}=\mathbf{c},h(X_{n+1})=j  \mid h]\\
&\quad=\frac{{n\choose c_{1},\ldots,c_{J}}q_{j}^{c_{j}+1}\prod_{j \in \{1, \ldots, J\} \setminus \{i\}}q_{i}^{c_{i}}\sum_{s\in\mathbb{S}_{j}}{c_{j}\choose r}\left(\frac{p_{s}}{q_{j}}\right)^{r+1}\left(1-\frac{p_{s}}{q_{j}}\right)^{c_{j}-r}}{{n\choose c_{1},\ldots,c_{J}}q_{1}\cdots q_{j-1}^{c_{j-1}}q_{j}^{c_{j}+1}q_{j+1}^{c_{j}}\cdots q_{J}^{c_{J}}}\\
&\quad=\sum_{s\in\mathbb{S}_{j}}{c_{j}\choose r}\left(\frac{p_{s}}{q_{j}}\right)^{r+1}\left(1-\frac{p_{s}}{q_{j}}\right)^{c_{j}-r}\\
&\quad={c_{j}\choose r}\sum_{s\in\mathbb{S}_{j}}\left(\frac{p_{s}}{q_{j}}\right)^{r+1}\left(1-\frac{p_{s}}{q_{j}}\right)^{c_{j}-r},
\end{align*}
which is a functional of $P_{j}=\sum_{s\in\mathbb{S}_{j}}(p_{s}/q_{j})\delta_{s}$ on $\mathbb{S}_{j}$, for $j=1,\ldots,J$. This completes the proof.

\subsection{Worst-case analysis}\label{app:frequency-recovery-worstcase}

For notation's sake, we define a slightly different class of distributions for our worst-case analysis, that allows us to take into account the number of support points of each $P_j$.
Let
\[
    \mathcal P_{J, K} : = \{P \text{ such that } P_j \text{ has at most $K$ support points for each } j=1, \ldots, J\}.
\]
Of course, $\mathcal P_L \subseteq  \mathcal P_{J, K}$.
The main technical result to prove Theorem~\ref{theorem:worst-case-informal} is the following lemma.
\begin{lemma}\label{teo:worst_case_f}
Let $K > 0$. Set $R_j(\hat f_\beta; P, h) = \E_P[(\beta_{j}C_{j} - f_{X_{n+1}} )^2 \mid h, h(X_{n+1}) = j]$ as the quadratic risk conditional to $h(X_{n+1}) = j$.
For any $P \in  \mathcal P_{J, K}$ we have
    \begin{equation}\label{eq:risk_freq_ub}
        R_j(\hat f_\beta; P, h) \leq \max_{A_{j}, B_{J}, C, D \in \Omega} \beta^2 n A_{j} + n(n-1) \left(\beta^2 - \frac{2\beta}{K}\right) B_{j} + n C + n(n-1) D,
    \end{equation}
  where the maximum is taken over the set  $\Omega: = \{(a, b, c, d) \in [0, 1]^4: \, a \geq b, a \geq c, b \geq d, c \geq d \}$, such that $A_{j}:= q_j$, $B_{j}:= q_j^2$, $C: = \sum_{s\in\mathbb{S_j}} p_s^2 / q_j$ and $D:=\sum_{s\in\mathbb{S}_j} p_s^3 / q_j$.
\end{lemma}
\begin{proof}
    As a first step, we write the conditional quadratic risk $R(f_{\beta}, P)= \E_{P}[(\beta C_{h(X_{n+1})} - f_{X_{n+1}})^2]$ as follows:
\begin{align*}
   R_j(\hat f_\beta; P, h)= \E_{P} \left[\left(\sum_{i = 1}^n (\beta_j \indicator(h(X_i) = j) - \indicator(X_i=X_{n+1})) \right)^2  \mid h, h(X_{n+1}) = j\right],
\end{align*}
where
\begin{align*}
    &\left(\sum_{i = 1}^n \beta \indicator(h(X_i)=j) - \indicator(X_i=X_{n+1}) \right)^2 \\
    & \quad = \beta^2 \sum_{i, l = 1}^n \indicator(h(X_i)=j) \indicator(h(X_l)=j)\\
    &\quad\quad - 2 \beta \sum_{i, l=1}^n \indicator(h(X_i)=j) \indicator(X_{n+1}=X_l) + \sum_{i, l = 1}^n \indicator(X_{n+1}=X_i) \indicator(X_{n+1}=X_l).
\end{align*}
Hence,
\begin{align*}
     R_j(\hat f_\beta; P, h) =& \beta^2 \sum_{i, l = 1}^n \prob[h(X_i) = h(X_l) = j\mid h] \\
        & -2\beta \sum_{i, l = 1}^n \prob[h(X_i)=j, X_{n+1} = X_l  \mid h, h(X_{n+1}) = j] \\
        &+ \sum_{i, l = 1}^n \prob[X_i = X_l = X_{n+1} \mid h(X_{n+1}) = j]
\end{align*}
Now, we consider each term separately. In particular, by separating the cases $i=l$ and $i \neq l$ we get
\begin{displaymath}
 \sum_{i, l = 1}^n \prob[h(X_i) = h(X_l) = j  \mid h ]  = n q_j + n (n-1) q_j^2,
\end{displaymath}
\begin{displaymath}
 \sum_{i, l = 1}^n \prob[h(X_i) = j, X_{n+1} = X_l  \mid h, h(X_{n+1}) = j] = n q_j\sum_{s \in \mathbb S_j} \left(\frac{p_s}{q_j}\right)^2 + n(n-1)  q_j^2 \sum_{s \in \mathbb S_j}\left(\frac{p_s}{q_j}\right)^2
\end{displaymath}
and
\begin{displaymath}
 \sum_{i, l = 1}^n \prob[X_i = X_l = X_{n+1} \mid h(X_{n+1}) = j] = n q_j \sum_{s \in \mathbb S_j}  \left(\frac{p_s}{q_j}\right)^2 + n (n-1) q_j^2 \sum_{s \in \mathbb S_j}  \left(\frac{p_s}{q_j}\right)^3,
\end{displaymath}
i.e,
\begin{align*}
    R_j(\hat f_\beta; P, h) =& \beta_j^2 \left( n q_j + n (n-1) q_j^2 \right) \\
                    &\quad  - 2 \beta \left(n q_j\sum_{s \in \mathbb S_j} \left(\frac{p_s}{q_j}\right)^2 + n(n-1)  q_j^2 \sum_{s \in \mathbb S_j}\left(\frac{p_s}{q_j}\right)^2 \right) + n (n-1) q_j^2 \sum_{s \in \mathbb S_j}  \left(\frac{p_s}{q_j}\right)^3.
\end{align*}
Now, let us focus on the set of distributions with at most $K \times J$ support points, such that each distribution $P_j$ has at most $K$ support points. Under this assumption it holds that $\sum_{s \in \mathbb S_j} (p_s/q_j)^2 \geq \frac{1}{K}$, where equality is achieved for $p_s = (K)^{-1}$.
Simple algebra completes the proof.
\end{proof}

The proof of Theorem~\ref{theorem:worst-case-informal} follows by noticing that for $K$  large enough, in such a way that $\beta^2 - \frac{2\beta}{K} > 0$, then the maximum on the right-hand side of \eqref{eq:risk_freq_ub} is obtained by setting $A_{j} = B_{j} = C = D = 1$, which entails $P \equiv P^{\ast}_j$, with $P^{\ast}_j$ being a degenerate distribution on $s^*_j \in \mathbb S_j$. 
Then
\begin{equation}\label{eq:ineq}
    R(\hat f_{\beta}; P, h) = \sum_{j=1}^J q_j  R_j(\hat f_\beta; P, h) \leq  \sum_{j=1}^J q_j  R_j(\hat f_\beta;P^{\ast}_j, h) 
\end{equation}
whose minimum is attained when $\beta_1 = \beta_2 = \ldots \beta_J = 1$.That is, the resulting estimator for $f_{X_{n+1}}$ is
\begin{displaymath}
\hat{f}_{1}=C_{h(X_{n+1})}.
\end{displaymath}

Finally, when $P = \delta_{s^*}$, the inequality in \eqref{eq:ineq} can be replaced with an equality showing the tightness of the upper bound.

\subsection{Algorithmic minimax analysis for frequency estimation}\label{app:minimax}

In the minimax framework, we look for $\beta$ and $P$ that solve the following optimization problem
\begin{equation}\label{eq:minmax}
    \inf_{\beta \in \mathbb R} \sup_{P \in \Delta_{K, J}} R(f_{\beta};P,  h),
\end{equation}
where $\Delta_{K, J}$ is the set of distributions on $\mathbb S$ such that, when restricted on each $\mathbb S_j$, $P$ gives positive mass to at most $K$ elements. The problem \eqref{eq:minmax} is convex in $\beta$ but is not concave in $P$. Then, looking for an analytic solution, or checking that a particular $(\beta, P)$ is a solution, is a non-trivial task (this might be impossible or, at least, NP-Hard). Here, we rely on numerical software in order to approximate \eqref{eq:minmax}. By using the epigraph representation of the minimax problem, we consider the constrained (equivalent) optimization
\begin{align*}
     \inf_{\beta \in \mathbb R} \ \  &Z \\
      \text{sub.} \ &Z \geq R(f_{\beta}; P, h) \\
      & P \in \Delta_{K, J},
\end{align*}
which is amenable to numerical optimization. We make use the Python package Pyomo as an interface to the IpoptSolver, which is the state of the art solver for constrained optimization. Since the problem is not concave in $P$, there will be locally-optimal solutions to \eqref{eq:minmax}. We run the algorithm several times with different starting points for $P$ and $\beta$, using the high-confidence criterion for the number of restarts \citep{Dic(14)}. We set $n=10,000$ and let $J = 10, 25, 50$, $K=10, 25, 50, 100$. For all such values, we find that the numerical solution to \eqref{eq:minmax} identifies as optimal parameters $\beta^* = 1/K$ and $P^*$ being the uniform distribution over the $K \times J$ support points. Figure~\ref{fig:minmax} summarizes these findings. 

\begin{figure}[!htb]
\centering
\includegraphics[width=\linewidth]{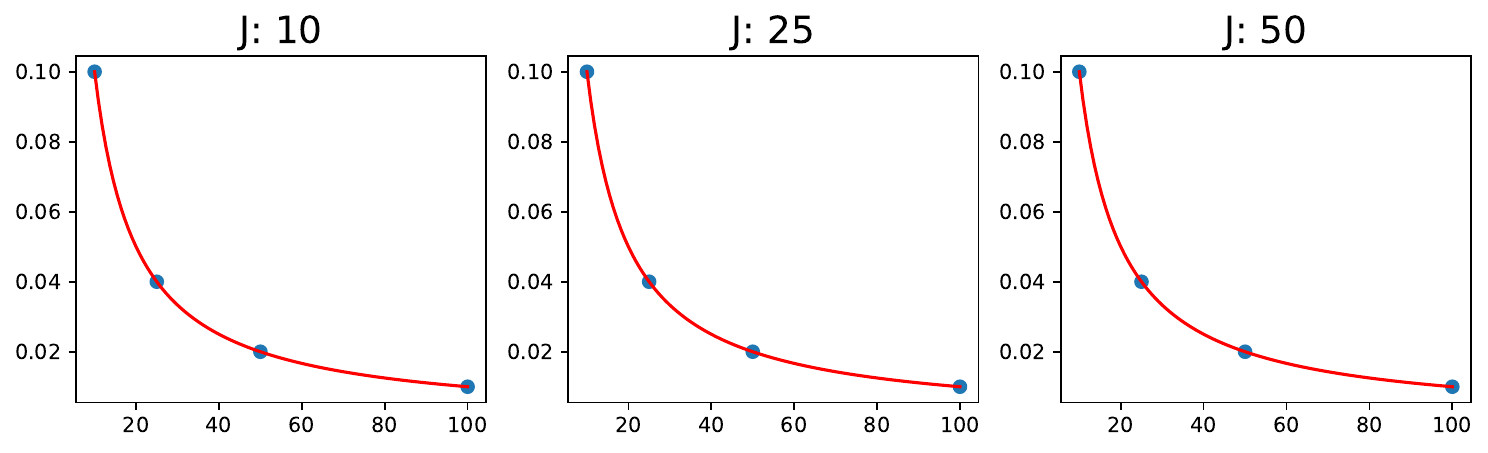}
\includegraphics[width=\linewidth]{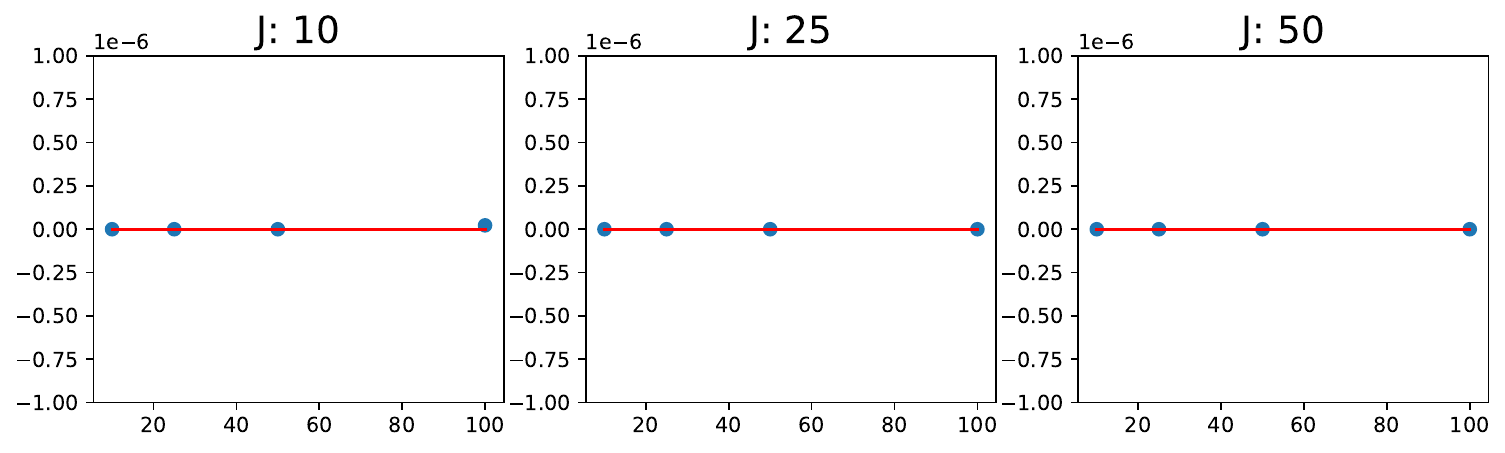}
\caption{Top row: the coefficients $\beta^*$ found numerically (blue dots) as $K$ varies on the x-axis. The red line is $1/K$. 
Bottom row: total variation distance between $P^*$ found numerically and the uniform distribution.}
\label{fig:minmax}
\end{figure} 

\FloatBarrier

\subsection{Proofs for Section~\ref{sec:smoothed}}

\subsubsection{Proof of Theorem~\ref{teo_smooth}}\label{proof_teo_smooth}

We start by recalling the definition of $\pi_j(r; P, h)$, namely
\[
    \pi_{j}(r;P, h):= {c_{j}\choose r}\sum_{s\in\mathbb{S}_{j}}\left(\frac{p_{s}}{q_{j}}\right)^{r+1}\left(1-\frac{p_{s}}{q_{j}}\right)^{c_{j}-r}.
\]
% We observe the following facts. 
Let $P \sim \mathrm{NRM}(\theta, G_0, \rho)$ such that $P = \sum_{k \geq 1} p_{k} \delta_{s_k}$. Fix $h$ and let $P_j$ denote the renormalized restriction of $P$ to $\mathbb S_j$. It follows from the restriction property of NRMs that $P_j \sim \mathrm{NRM}(\theta G_0(\mathbb S_j), G_{0, \mathbb S_j} / G_0(\mathbb S_j), \rho)$. 
Hence, $P_j = \sum_{S \geq 1} p_{j, k} \delta_{\tilde{s}_k}$ almost surely, where $\tilde s_k \simiid  G_{0, \mathbb S_j} / G_0(\mathbb S_j)$ and $(p_{j, k})_{k \geq 1}$ are probabilities obtained by normalizing the jumps of a Poisson process on $\mathbb R_+$ as explained in Section \ref{sec:background_nrm}. Then, we have
% First, fix $h$ and let $P_j \sum_{S \geq 1} p_{j, k} \delta_{\tilde{s}_k}$. {\color{red}[TODO: Please define the notation explicitly. Make sure it does not conflict with the definition of $P_j$ used in the main text.]}
% Then the restriction property of NRMs entails
\[
    \E_{P \sim \text{NRM}(\theta, G_0, \rho)}(\pi_{j}(r;P, h)) = \binom{c_j}{r} \E_{P_j \sim \text{NRM}(\theta G_0(\mathbb S_j),G_0 / G_0(\mathbb S_j),\rho)} \left[ \sum_{k \geq 1} p_{j, k}^{r+1} (1 - p_{j, k})^{c_j-r}  \right].
\]
In particular, the latter expectation does not depend on the atoms $\tilde{s}_k$ nor on the set $\mathbb S_j$, but only on the law of $( p_{j, k})_{k \geq 1}$.
Therefore, letting $\tilde{P}_j \sim \text{NRM}(\theta G_0(\mathbb S_j),\mathcal{U}([0, 1]),\rho)$ we further have
\[
    \E_{P \sim \text{NRM}(\theta, G_0, \rho)}(\pi_{j}(r;P, h)) = \binom{c_j}{r} \E_{\tilde{P}_j \sim \text{NRM}(\theta G_0(\mathbb S_j),\mathcal{U}([0, 1]),\rho)} \left[ \sum_{k \geq 1} p_{j, k}^{r+1} (1 - p_{j, k})^{c_j-r}  \right].
\]
The strong universality of of the hash family $\mathscr{H}_{J}$ and the independence of $\mathscr{H}_{J}$ from $P$ imply $\E_{h \sim \mathscr{H}}[G_0(\mathbb S_j)] = 1/J$ for any $j=1, \ldots, J$, so that 
$\E_{h \sim \mathscr{H}}[\tilde{P}_{j}]\sim\text{NRM} \left(\theta/J,\mathcal{U}([0, 1]),\rho\right)$, leading to
\begin{equation}\label{eq:aaaa}
     \E_{h \sim \mathscr H}\left[\E_{P \sim \text{NRM}(\theta, G_0, \rho)}(\pi_{j}(r;P, h))\right] = \binom{c_j}{r} \E_{\bar P_j \sim \text{NRM}(\theta / J,\mathcal{U}([0, 1]),\rho)} \left[ \sum_{k \geq 1} p_{j, k}^{r+1} (1 - p_{j, k})^{c_j-r}  \right].
\end{equation}
Now recall that  $\bar P_j \sim \text{NRM}(\theta/ J,\mathcal{U}([0, 1]),\rho)$ entails $\bar P_j = \mu_j / T_j$ where $\mu_j$ is a CRM with L\'evy intensity $\theta / J \ddr s \rho(\ddr x)$ with $T_{j}$ being the total mass of $\mu_{j}$, i.e. $T_{j}={\mu}_{j}([0, 1])$. 
Then, from \citet[Equation 11]{Pit(03)} the right-hand side of \eqref{eq:aaaa} is
\begin{displaymath}
{c_{j}\choose r}\int_{0}^{1}v^{r}(1-v)^{c_{j}-r}f_{V_{j}}(v)\ddr v
\end{displaymath}
where
\begin{displaymath}
f_{V_{j}}(v)=\frac{\theta}{J}v\int_{0}^{+\infty}t\rho(tv)f_{T_{j}}(t(1-v))\ddr t,
\end{displaymath}
follows from \citet[Lemma 1]{Pit(03)}, i.e. \citet[Equation 25]{Pit(03)}. This completes the proof.

% The proof combines the restriction property of NRMs with some identities for functionals of NRMs. In particular, from the restriction property, $P_{j}\sim\text{NRM}(\theta/J,JG_{0,\mathbb{S}_{j}},\rho)$ such that
% \begin{equation}\label{pk1}
% \pi_{j}(r)=\E_{P_{j}\sim\text{NRM}(\theta/J,JG_{0,\mathbb{S}_{j}},\rho)}\left[{c_{j}\choose r}\sum_{s\in\mathbb{S}_{j}}\left(\frac{p_{s}}{q_{j}}\right)^{r+1}\left(1-\frac{p_{s}}{q_{j}}\right)^{c_{j}-r}\right],
% \end{equation}
% for $j=1,\ldots,J$. Let $\tilde{\mu}_{j}$ be a CRM on $\mathbb{S}_j$ with L\'evy intensity $\nu(\ddr x,\,\ddr s)=\theta G_{0,\mathbb{S}_{j}}(\ddr s)\rho(\ddr x)$, such that
% \begin{displaymath}
% P_{j}(\cdot)=\frac{\tilde{\mu}_{j}(\cdot)}{T_{j}}
% \end{displaymath}
% with $T_{j}$ being the total mass of $\tilde{\mu}_{j}$, i.e. $T_{j}=\tilde{\mu}_{j}(\mathbb{S}_{j})$, whose distribution has density function denoted by $f_{T_{j}}$. Then, from \citet[Equation 11]{Pit(03)} the right-hand side of \eqref{pk1} is
% \begin{displaymath}
% {c_{j}\choose r}\int_{0}^{1}v^{r}(1-v)^{c_{j}-r}f_{V_{j}}(v)\ddr v
% \end{displaymath}
% where
% \begin{displaymath}
% f_{V_{j}}(v)=\frac{\theta}{J}v\int_{0}^{+\infty}t\rho(tv)f_{T_{j}}(t(1-v))\ddr t,
% \end{displaymath}
% follows from \citet[Lemma 1]{Pit(03)}, i.e. \citet[Equation 25]{Pit(03)}. This completes the proof.

\subsubsection{Proof of Theorem~\ref{teo:optimality}}\label{proof_teo_optimality}

Let $(\tilde{X}_1, \ldots, \tilde{X}_{C_j}) = \{X_i \colon h(X_i) = j\}$. We have
\begin{align*}
    \text{cMSE}(\beta) &= \E\left[ \left(\beta C_j - \sum_{i=1}^{C_j} \indicator[\tilde{X}_i = X_{n+1}] \right)^2\right].
\end{align*}
Recall that, conditional to $h(X_{n+1}) = j$, $(\tilde{X}_1, \ldots, \tilde{X}_{C_j}, X_{n+1})$ is a sample from $P_j \sim \text{NRM}(\theta/J, JG_{0, \mathbb S_j}, \rho)$. Then,
\begin{align*}
    \text{cMSE}(\beta) &= \beta^2 \E\left[C_j^2\right] - 2 \beta \E\left[C_j \sum_{i=1}^{C_j} \indicator[\tilde{X}_i = X_{n+1}] \right] \\
    & \qquad + \E\left[\sum_{i, \ell=1}^{C_j} \indicator[\tilde{X}_i = \tilde{X}_\ell = X_{n+1}]\right]
\end{align*}
focusing on the second term on the r.h.s., by the tower property of the expectation, we have
\begin{align*}
    & \E\left[C_j \sum_{i=1}^{C_j} \indicator[\tilde{X}_i = X_{n+1}] \right] \label{eq:aa}\\
    & \qquad = \E\left[ \E\left[C_j \sum_{i=1}^{C_j} \indicator[\tilde{X}_i = X_{n+1}] \mid C_j \right]\right] \\
    & \qquad = \E\left[ C_j^2 \E\left[\indicator[\tilde{X}_i = X_{n+1}]  \right]\right] \\
    & \qquad = \E[C_j^2] \eppf_j(2),
\end{align*}
where $\eppf_j$ is the exchangeable partition probability function \citep{Pit(06)} associated to $P_j$.
Therefore, the optimal $\beta$ is $\hat \beta_c = \eppf_j(2)$.

The equality between $\hat \beta_c$ and the expression implied by the smoothed estimator follows from noticing that $\eppf_j(2)$ can be written (see equation 2.25 in \cite{Pit(06)}) as
\[
    \eppf_j(2) = \int_{0}^{\infty} v f_{V_j}(v)
\]
where $f_{V_j}(v)$ is as \eqref{struct}.
Now, consider the expression of $\pi_j(r)$ and let the smoothed estimator be $\hat f^S = \sum_{r \geq 0} r \pi_j(r)$. By linearity and the tower property and the expectation yields
\begin{align*}
    \hat f^S &= \E\left[\sum_{r \geq 1} r \text{Binomial}(r; c_j, V_j) \right] \\
            &= c_j \E[V_j] = c_j \eppf(2)
\end{align*}

It is easy to see that $\hat \beta_c C_j$ is an unbiased estimator:
\[
    \E[f_{X_{n+1}} \mid h(X_{n+1}) = j] = \E\left[\sum_{i=1}^{C_j} \indicator[\tilde{X}_i = X_{n+1}]\right] = \E[C_j] \eppf(2).
\]

\subsection{Analytical details and proofs for the smoothing examples}\label{app:smoothing-examples}

\subsubsection{The distribution of $V_j$ in the Dirichlet process case}\label{proof_struct_dp}

We show here that, under the DP smoothing assumption, $V_j \sim \mathrm{Beta}(1, \theta/J)$.
The proof follows by combining the general form \eqref{struct} with the L\'evy intensity of the DP. 
In particular, if $P\sim\text{DP}(\theta,G_{0})$ then by the restriction property $P_{j}\sim\text{DP}(\theta/J,JG_{0,\mathbb{S}_{j}})$, for $j=1,\ldots,J$. That is, $P_{j}$ is a NRM obtained by normalizing a CRM $\tilde{\mu}_{j}$ on $\mathbb{S}_j$ with L\'evy intensity 
\begin{equation}\label{pk2}
\nu(\ddr x,\ddr s)=\theta G_{0,\mathbb{S}_{j}}(\ddr s)\frac{\text{e}^{-x}}{x}\ddr x.
\end{equation}
From \eqref{pk2}, by means of the L\'evy-Khintchine formula for the Laplace functional of $\mu_{j}$, one has
\begin{equation}\label{pk3}
f_{T_{j}}(t)=\frac{1}{\Gamma(\theta/J)}t^{\theta/J-1}\text{e}^{-t},
\end{equation}
\citep[Section 5.1]{Pit(03)}. Then, by combining \eqref{struct} with \eqref{pk2} and \eqref{pk3} we can write that
\begin{align*}
f_{V_{j}}(v)&=\frac{\theta}{J}v\int_{0}^{+\infty}t\frac{\text{e}^{-tv}}{tv}\frac{1}{\Gamma(\theta/J)}(t(1-v))^{\theta/J-1}\text{e}^{-t(1-v)}\ddr t\\
&=\frac{\theta/J}{\Gamma(\theta/J)}\int_{0}^{+\infty}(t(1-v))^{\theta/J-1}\text{e}^{-t}\ddr t\\
&=\frac{\theta/J}{\Gamma(\theta/J)}(1-v)^{\theta/J-1}\int_{0}^{+\infty}t^{\theta/J-1}\text{e}^{-t}\ddr t\\
&=\frac{\theta}{J}(1-v)^{\theta/J-1},
\end{align*}
which is the density function of a Beta distribution with parameter $(1,\theta/J)$. The proof is completed.

\subsubsection{Proof of Equation \eqref{eq:prop-dp}}\label{proof_prop_dp}

By Fubini's theorem we have
\begin{align*}
    \hat{f}_{X_{n+1}}^{\text{DP}} &= \E_{P \sim \text{DP}, h \sim \mathscr{H}_J}[\varepsilon_f(P,h)] \\
    &= \sum_{r=0}^{c_{j}} r \E_{P \sim \text{DP}, h \sim \mathscr{H}_J}[\pi_j(r; P,h)] \\
    &= \sum_{r=0}^{c_{j}}r\int_{0}^{1}{c_{j}\choose r}v^{r}(1-v)^{c_{j}-r}f_{V_{j}}(v)\ddr v\\
&=c_{j}\int_{0}^{1}vf_{V_{j}}(v)\ddr v\\
&=c_{j}\frac{J}{\theta+J}.
\end{align*}

% It is useful to recall the expression of $\varepsilon_{f}(P; h)$ in terms of the $\pi_{j}(r,P)$'s. From \eqref{cond_prob_mean}:
% \begin{equation}\label{pk5}
% \varepsilon_{f}(P;h)=\sum_{r=0}^{c_{j}}r\pi_{j}(r,P; h)
% \end{equation}
% and
% % \begin{equation}\label{pk6}
% % \varepsilon_{k}(P)=\sum_{j=1}^{J}q_{j}\sum_{r=0}^{c_{j}}\frac{n}{r+1}\pi_{j}(r,P)
% % \end{equation}
% The estimator of $f_{X_{n+1}}$ then follows from \eqref{pk5} by replacing the probabilities $\pi_{j}(r,P;h)$'s with the estimated probabilities $\pi_{j}(r)$ in \eqref{mass_dp}:
% \begin{align*}
% \hat{f}_{X_{n+1}}^{\text{DP}}&=\sum_{r=0}^{c_{j}}r\int_{0}^{1}{c_{j}\choose r}v^{r}(1-v)^{c_{j}-r}f_{V_{\theta}}(v)\ddr v\\
% &=c_{j}\int_{0}^{1}vf_{V_{\theta}}(v)\ddr v\\
% &=c_{j}\frac{J}{\theta+J}.
% \end{align*}

\subsubsection{Proof of Theorem~\ref{dp_charact}}\label{app_dp_charact}

We start by observing that $\hat f_{X_{n+1}}^{\text{NRM}} = \E_{P \sim \text{NRM}} \left[\varepsilon_f(P)\right] = \sum_{r=0}^{c_j} r \pi_j(r)$, where $\pi_j$ is defined in Theorem~\ref{teo_smooth}, depends on the sketch $\mathbf C_J$ only through $C_{h(X_{n+1})}$ by definition.
On the other hand, the Bayesian posterior in Theorem 2.2 in \cite{Ber(23)} depends, in general, on the whole sketch. Furthermore, Theorem 2.3 in \cite{Ber(23)} characterizes the DP as the sole Poisson-Kingman (a fortiori, the only NRM) for which the posterior of $f_{X_{n+1}}$ depends on $\mathbf C$ only through $C_{h(X_{n+1})}$.
The proof is concluded by noting the equality between the posterior under the DP prior \citep{Cai(18), Dol(23)} and $\hat f^{DP}_{X_{n+1}}$.

\subsubsection{Proof of Equation \eqref{struct_ngg}}\label{proof_struct_ngg}

The proof follows by combining the general form \eqref{struct} with the L\'evy intensity of the NGGP. 
If $P\sim\text{NGGP}(\theta,G_{0},\alpha,\tau)$ then by the restriction property $P_{j}\sim\text{NGGP}(\theta/J,JG_{0,\mathbb{S}_{j}},\alpha,\tau)$, for $j=1,\ldots,J$. That is, $P_{j}$ is a NRM obtained by normalizing a CRM $\tilde{\mu}_{j}$ on $\mathbb{S}_j$ with L\'evy intensity 
\begin{equation}\label{pk21}
\nu(\ddr x,\,\ddr s)=\theta G_{0,\mathbb{S}_{j}}(\ddr s)\frac{1}{\Gamma(1-\alpha)}x^{-1-\alpha}\text{e}^{-\tau x}.
\end{equation}
From \eqref{pk21}, by means of the L\'evy-Khintchine formula for the Laplace functional of $\mu_{j}$, one has
\begin{equation}\label{pk31}
f_{T_{j}}(t)=\text{e}^{\frac{(\theta/J)\tau^{\alpha}}{\alpha}}\text{e}^{-\tau t}f_{\alpha}(t),
\end{equation}
where $f_{\alpha}$ is a density function such that $\int_{0}^{+\infty}\exp\{-\lambda t\}f_{\alpha}(t)\ddr t=\exp\{-(\theta/\alpha J)\lambda^{\alpha}\}$ \citep[Section 5.4]{Pit(03)}. Then, by combining \eqref{struct} with \eqref{pk21} and \eqref{pk31} we can write that
\begin{align*}
f_{V_{j}}(v)&=\frac{\theta}{J}v\int_{0}^{+\infty}t\frac{1}{\Gamma(1-\alpha)}(tv)^{-1-\alpha}\text{e}^{-\tau tv}\text{e}^{\frac{(\theta/J)\tau^{\alpha}}{\alpha}}\text{e}^{-\tau t(1-v)}f_{\alpha}(t(1-v))\ddr t\\
&=\frac{\theta/J}{\Gamma(1-\alpha)}\text{e}^{\frac{(\theta/J)\tau^{\alpha}}{\alpha}}v^{-\alpha}\int_{0}^{+\infty}t^{-\alpha}\text{e}^{-\tau t}f_{\alpha}(t(1-v))\ddr t\\
&=\frac{\theta/J}{\Gamma(1-\alpha)}\text{e}^{\frac{(\theta/J)\tau^{\alpha}}{\alpha}}v^{-\alpha}(1-v)^{\alpha-1}\int_{0}^{+\infty}z^{-\alpha}\text{e}^{-\tau\frac{z}{1-v}}f_{\alpha}(z)\ddr z\\
&[\text{by using }z^{-\alpha}=\frac{1}{\Gamma(\alpha)}\int_{0}^{+\infty}y^{\alpha-1}\text{e}^{-yz}\ddr y]\\
&=\frac{\theta/J}{\Gamma(1-\alpha)}\text{e}^{\frac{(\theta/J)\tau^{\alpha}}{\alpha}}v^{-\alpha}(1-v)^{\alpha-1}\int_{0}^{+\infty}\left(\frac{1}{\Gamma(\alpha)}\int_{0}^{+\infty}y^{\alpha-1}\text{e}^{-yz}\ddr y\right)\text{e}^{-\tau\frac{z}{1-v}}f_{\alpha}(z)\ddr z\\
&=\frac{\theta/J}{\Gamma(\alpha)\Gamma(1-\alpha)}\text{e}^{\frac{(\theta/J)\tau^{\alpha}}{\alpha}}v^{-\alpha}(1-v)^{\alpha-1}\int_{0}^{+\infty}\int_{0}^{+\infty}y^{\alpha-1}\text{e}^{-z\left(y+\frac{\tau}{1-v}\right)}f_{\alpha}(z)\ddr z\ddr y\\
&=\frac{\theta/J}{\Gamma(\alpha)\Gamma(1-\alpha)}\text{e}^{\frac{(\theta/J)\tau^{\alpha}}{\alpha}}v^{-\alpha}(1-v)^{\alpha-1}\int_{0}^{+\infty}y^{\alpha-1}\text{e}^{-\frac{\theta/J}{\alpha}\left(y+\frac{\tau}{1-v}\right)^{\alpha}}\ddr y\\
&=\frac{\theta/J}{\Gamma(\alpha)\Gamma(1-\alpha)}\text{e}^{\frac{(\theta/J)\tau^{\alpha}}{\alpha}}v^{-\alpha}(1-v)^{\alpha-1}\int_{0}^{+\infty}y^{\alpha-1}\text{e}^{-\frac{(\theta/J)\tau^{\alpha}}{\alpha}\left(\frac{y}{\tau}+\frac{1}{1-v}\right)^{\alpha}}\ddr y\\
&[\text{by the change of variable }h=y/\tau]\\
&=\tau^{\alpha}\frac{\theta/J}{\Gamma(\alpha)\Gamma(1-\alpha)}\text{e}^{\frac{(\theta/J)\tau^{\alpha}}{\alpha}}v^{-\alpha}(1-v)^{\alpha-1}\int_{0}^{+\infty}h^{\alpha-1}\text{e}^{-\frac{(\theta/J)\tau^{\alpha}}{\alpha}\left(h+\frac{1}{1-v}\right)^{\alpha}}\ddr h\\
&[\text{by the change of variable }y=(h-hv+v)/(h-hv+1)]\\
&=\tau^{\alpha}\frac{\theta/J}{\Gamma(\alpha)\Gamma(1-\alpha)}\text{e}^{\frac{(\theta/J)\tau^{\alpha}}{\alpha}}v^{-\alpha}(1-v)^{\alpha-1}\\
&\quad\times\int_{v}^{1}\frac{1}{(1-y)^{2}}\left(\frac{y-v}{(1-v)(1-y)}\right)^{\alpha-1}\text{e}^{-\frac{(\theta/J)\tau^{\alpha}}{\alpha}\left(\frac{y-v}{(1-v)(1-y)}+\frac{1}{1-v}\right)^{\alpha}}\ddr y\\
&=\tau^{\alpha}\frac{\theta/J}{\Gamma(\alpha)\Gamma(1-\alpha)}\text{e}^{\frac{(\theta/J)\tau^{\alpha}}{\alpha}}v^{-\alpha}\\
&\quad\times\int_{v}^{1}(y-v)^{\alpha-1}(1-y)^{-\alpha-1}\text{e}^{-\frac{(\theta/J)\tau^{\alpha}}{\alpha}(1-y)^{-\alpha}}\ddr y\\
&=\tau^{\alpha}\frac{\theta/J}{\Gamma(\alpha)\Gamma(1-\alpha)}\text{e}^{\frac{(\theta/J)\tau^{\alpha}}{\alpha}}\\
&\quad\times\int_{v}^{1}\frac{1}{y}\left(\frac{v}{y}\right)^{1-\alpha-1}\left(1-\frac{v}{y}\right)^{\alpha-1}(1-y)^{-\alpha-1}\text{e}^{-\frac{(\theta/J)\tau^{\alpha}}{\alpha}(1-y)^{-\alpha}}\ddr y,
\end{align*}
which is the density function of the distribution of the product of two independent random variables: i) a Beta random variable with parameter $(1-\alpha,\alpha)$; ii) the random variable 
\begin{displaymath}
1-\left(\frac{\frac{\theta\tau^{\alpha}}{J\alpha}}{\frac{\theta\tau^{\alpha}}{J\alpha}+E}\right)^{1/\alpha},
\end{displaymath}
with $E$ being a negative Exponential random variable with parameter $1$. The proof is completed.

\subsubsection{Proof of Equation \eqref{eq:prop-nggp}}\label{proof_prop_nggp}

Applying the same reasoning of Appendix~\ref{proof_prop_dp}, we get
% It is useful to recall the expression of $\varepsilon_{f}(P)$ and $\varepsilon_{k}(P)$ in terms of the $\pi_{j}(r,P)$'s, i.e.,
% \begin{equation}\label{pk55}
% \varepsilon_{f}(P)=\sum_{r=0}^{c_{j}}r\pi_{j}(r,P)
% \end{equation}
% and
% \begin{equation}\label{pk66}
% \varepsilon_{k}(P)=\sum_{j=1}^{J}q_{j}\sum_{r=0}^{c_{j}}\frac{n}{r+1}\pi_{j}(r,P)
% \end{equation}
% from \eqref{cond_prob_mean} and \eqref{estim_k}, respectively. 
% Then, the estimator of $f_{X_{n+1}}$ follows from \eqref{pk55} by replacing the probabilities $\pi_{j}(r,P)$'s with the estimated probabilities $\pi_{j}(r)$ in \eqref{mass_ngg}:
\begin{align*}
\hat{f}_{X_{n+1}}^{\text{\tiny{(NGGP)}}}&=\sum_{r=0}^{c_{j}}r\int_{0}^{1}{c_{j}\choose r}v^{r}(1-v)^{c_{j}-r}f_{V_{\theta,\alpha,\tau}}(v)\ddr v\\
&=c_{j}\int_{0}^{1}vf_{V_{\theta,\alpha,\tau}}(v)\ddr v\\
&=c_{j}\E\left[B_{1-\alpha,\alpha}\left(1-\left(\frac{\frac{\theta\tau^{\alpha}}{J\alpha}}{\frac{\theta\tau^{\alpha}}{J\alpha}+E}\right)^{1/\alpha}\right)\right]\\
&=c_{j}(1-\alpha)\left(1-\E\left[\left(\frac{\frac{\theta\tau^{\alpha}}{J\alpha}}{\frac{\theta\tau^{\alpha}}{J\alpha}+E}\right)^{1/\alpha}\right]\right)\\
&=c_{j}(1-\alpha)\left(1-\frac{\theta\tau^{\alpha}}{J\alpha}\text{e}^{\frac{\theta\tau^{\alpha}}{J\alpha}}E_{1/\alpha}\left(\frac{\theta\tau^{\alpha}}{J\alpha}\right)\right),
\end{align*}
where $E$ denotes the exponential integral function, defined as $E_{a}(z)=\int_{1}^{+\infty}x^{-a}\exp\{-zx\}\ddr x$. 

\subsection{Empirical estimation of smoothing parameters} \label{sec:estimation}

We turn now to the problem of estimating the smoothing parameters from the sketch $\mathbf{C}_{J}$. 
The NRM smoothing assumption requires specifying the parameter $\theta>0$ and a collection of parameters, here denoted by $\xi$, introduced by the specification of $\rho$ in the L\'evy intensity of the NRM. For instance, the specification of the L\'evy intensity of the NGGP implies that $\xi=(\alpha,\tau)$. 
Motivated by Theorem~\ref{teo:optimality}, which ensures the optimality of our smoothed estimator if $X_1, \ldots, X_n$ is a sample from an NRM, we propose to estimate such hyperparameters by maximizing the marginal likelihood of the sketch, i.e.,
\begin{align} \label{eq:marginal-likelihood}
    L(\theta, \xi; \mathbf c) = \E_{P \sim \text{NRM}(\theta, G_0, \rho)}[  \prob[\mathbf C_J = \mathbf c; P]],
\end{align}
where $\prob[\mathbf C_J = \mathbf c; P]$ is defined as
\[
\text{Pr}[\mathbf{C}_{J}=\mathbf{c}; P]={n\choose c_{1},\ldots,c_{J}}\prod_{j=1}^{J}q_{j}^{c_{j}}.
\]
An equivalent strategy, motivated from the Bayesian standpoint, was adopted in \cite{Cai(18)} and \cite{Dol(23)}.

A tractable closed-form expression of \eqref{eq:marginal-likelihood} is available in the special case of the DP smoothing assumption \citep{Cai(18)}, as detailed in Appendix~\ref{app:estimation-dp} below. 
To the best of our knowledge, it is not possible to compute a closed-form expression of \eqref{eq:marginal-likelihood} for more general NGGP smoothing assumptions  \citep{Pit(03),Lij(10)}. The only exception may be when $\alpha=1/2$, which corresponds to the NIGP. In that case, \eqref{eq:marginal-likelihood} can be computed by relying on the definition of the NIGP presented in \citet[Section 3.1]{Lij(05)}, although such an approach results in a complicated expression involving Bessel functions \citep{Dol(21)}. For these reasons, we propose two alternative strategies to estimate smoothing parameters.

The first strategy is relatively efficient from a computational perspective, but it can only be applied if one has direct access to a subset of the original (un-sketched) data. This approach is inspired by the work of \cite{Ses(22)}, which utilized a similar solution for a somewhat different purpose. The key idea is that, if one can store the first $m<n$ samples in memory before sketching the rest of the data, which is an assumption that is  sometimes realistic, then the smoothing parameter $(\theta,\xi)$ can be estimated by maximizing the marginal likelihood of the observed data, which tends to be a more tractable quantity compared to the marginal likelihood of the sketch. We refer to Appendix~\ref{sec:estimation-nggp} for further details on this approach. 

The second strategy to estimate $(\theta,\xi)$ was originally proposed by \cite{Dol(23)} and it is more widely applicable, because it only requires access to the sketched data, but it is also more computationally expensive. The key idea is to simulate synthetic data based on a smoothing distribution given the smoothing parameter $(\theta,\xi)$, apply the hash function to produce a sketch of them, and then optimize the choice of $(\theta,\xi)$ in such a way as to minimize the $1$-Wasserstein distance between the empirical distribution of the synthetic sketch and the observed sketch. This minimum-distance approach is explained more carefully in Appendix~\ref{sec:estimation-nggp}, where we provide some implementation details for the special case of NGGP smoothing assumption, and also outline some convenient computational shortcuts.

\subsubsection{The Dirichlet Process case}\label{app:estimation-dp}

From the finite-dimensional laws of the Dirichlet process, it is easy to show that the sketch $\mathbf C_J$ follows a Dirichlet-multinomial distribution with parameter $(\theta/J, \ldots, \theta/J)$.
The log-likelihood is then
\[
    \ell(\theta; \mathbf C_J) \propto \log \Gamma(\theta) - \log \Gamma(\theta + n ) - J \log \Gamma(\theta/J) + \sum_{j=1}^J \log \Gamma (\theta/J + c_j)
\]
We choose $\theta$ by maximizing $\ell(\theta; \mathbf C_J)$ via the BFGS algorithm.

\subsubsection{The Normalized Generalized Gamma Process case} \label{sec:estimation-nggp}

The finite-dimensional laws under the NGGP do not have closed-form expressions so that the likelihood of $\mathbf C_J$ is not available. We consider then two strategies. The first one is more computationally efficient but applies only when we have access to a subsample of the $X_i$'s. 
This is often possible when handling streaming data, for which we can simply decide to store the first $m$ samples in memory before hashing them. 
Such a setting was recently considered in \cite{Ses(22)}.
The second one is more general but more demanding from the computational side. 

In the first case, we estimate the parameters by maximizing the marginal likelihood of $X_1, \ldots, X_m$.  Its expression can be found in \cite{Lij(07)} but it involves a sum of $m$ terms (with alternating signs) involving incomplete Gamma functions. We found it essentially impossible to obtain a numerically stable procedure to evaluate such an expression for $m > 50$.
Instead, we focus here on the alternative expression found in Proposition 3 of \cite{Jam(09)}. Denoting by $X^*_1, \ldots, X^*_{k_m}$ the $k_m$ the unique values in $X_1, \ldots, X_m$ and by $n_{1}, \ldots, n_{k_m}$ their cardinalities, we have
\[
    \prob[X_1 \in \mathrm d x_1, \ldots, X_n \in \mathrm d x_n] = \frac{\theta^{k_m} e^\beta}{\Gamma(m)} \prod_{j=1}^{k_m} (1 - \alpha)_{(n_j - 1)} \int_{\mathbb R_+} u^{m-1}  e^{-\frac{\theta}{\alpha}(\tau + u)^\alpha} (\tau + u)^{-m + k_m \alpha}\mathrm d u
\]
where $\beta = \theta/\alpha \tau^\alpha$ and $(a)_{(b)}$ is the Pocchammer symbol.

For numerical stability reasons we work in the logarithmic scale so that the log-likelihood is
\[
     \ell(\theta; \mathbf C_J) \propto k_m \log \theta + \beta - k_m \log(1 - \alpha) \sum_{j=1}^{k_m} \log \Gamma(n_j - \alpha) + \log \left( F_{m, k_m, \theta, \alpha, \tau} \right)
\]
where $F_{m, k_m, \theta, \alpha, \tau} = \int_{\mathbb R_+} u^{m-1}  e^{-\frac{\theta}{\alpha}(\tau + u)^\alpha} (\tau + u)^{-m + k_m \alpha}\mathrm d u$. To evaluate the logarithm of the last integral we employ once again the log-sum-exp trick.

Given that the parameters $\theta$ and $\tau$ are redundant (i.e., if $\mu \sim NGGP(\theta, \alpha, \tau, G_0)$ then for any $q > 0$ $q \mu \sim NGGP(\theta q^\alpha, \alpha, \tau / q, G_0)$) we fix $\tau = 1/2$ and optimize over $\theta \in \mathbb{R_+}$ and $\alpha \in (0, 1)$ using the BFGS algorithm.

If instead, we have access only to the sketch $\mathbf C_J$, we follow the approach proposed in \cite{Dol(23)} based on a minimum-distance estimator.
Specifically, this consists in finding $(\hat \theta, \hat \alpha)$ that minimize 
\begin{equation}\label{eq:exp_wass}
    \E\left[W_1(\mathbf C_J, \tilde{\mathbf C}_J(\theta, \alpha, m) \frac{n}{m} )\right]
\end{equation}
where $W_1$ is the 1-Wasserstein distance and the expectation is with respect $\tilde{\mathbf C}_J(\theta, \alpha, m)$, which is a synthetic sketch obtained by sampling $\tilde{X}_1, \ldots, \tilde{X}_m$ from the marginal distribution of an NGGP with parameter $(\theta, \alpha, 1/2)$ and hashing them with another hash function. 
Note that the use of the Wasserstein distance is fundamental due to the arbitrariness of the labels attached to the $X_i$'s.
In practice, \eqref{eq:exp_wass} is approximated via Monte Carlo.
To sample $\tilde{X}_1, \ldots, \tilde{X}_m$, we use a generalized P\'olya Urn scheme and sample $\tilde{X}_1 \sim G_0$, $\tilde{X}_2 \mid \tilde{X}_1$, and so on.
Explicit formulas for these conditional distributions are given in Eqs. (9)-(11) of \cite{Lij(07)}. However, as discussed above in the case of the marginal distribution, these involve sums of incomplete Gamma functions that are extremely prone to numerical instability.  Instead, we use the generative scheme presented in Section 3.3 of \cite{Jam(09)}. We report it below.
\begin{itemize}
    \item Sample $\tilde{X}_1 \sim G_0$, set $X^*_1 = \tilde{X}_1$, $n_1 = 1$, $k=1$.

    \item For $i = 1, \ldots, m-1$

    \begin{enumerate}
        \item Sample a latent variable $U$ with density
        \[
            f_U(u) \propto \frac{u^{i-1}}{(u + \tau)^{i - \alpha k}} e^{-\frac{\theta}{\alpha}((u+\tau)^\alpha - \tau^\alpha)}
        \]
        \item Sample $c$ from a categorical distribution over $1, \ldots, k + 1$ such that
        \[
            \prob[c = h \mid U = u] \propto \begin{cases}
                n_h - \alpha & \text{ for } h = 1, \ldots, k \\
                \theta(u + \tau)^\alpha & \text{ for } h = k+1
            \end{cases}
        \]
        If $c = k+1$, sample $X^*_{k+1} \sim G_0$ and update $k \mapsto k+1$. Set $X_{i+1} = X^*_{c}$ and $n_k \mapsto n_k + 1$.
    \end{enumerate}
\end{itemize}
When sampling the latent variable $U$, it is more convenient to sample $V = \log U$ whose density is
\begin{equation*}
        f_V(v) \propto \frac{e^{v i}}{(e^v + \tau)^{i - \alpha k_i}} e^{-\frac{\theta}{\alpha}((e^v+\tau)^\alpha - \tau^\alpha)}
\end{equation*}
which is log-concave. We employ the adaptive rejection sampling algorithm \citep{Gil(92)} implemented in the \texttt{Julia} package \texttt{AdaptiveRejectionSampling.jl} \citep{Tec(18)}.

\section{Cardinality estimation} \label{sec:cardinality}

The results presented in Sections~\ref{sec:setup} and~\ref{sec:smoothed} can be extended to also address the distinct problem of cardinality recovery. 
To describe this problem, consider the usual sample of $\mathbb{S}$-valued data points $(x_{1},\ldots,x_{n})$, for $n\geq1$, and let $\mathbf{C}_{j}$ be a sketch of the $x_{i}$'s obtained through a random hash function $h$ from the strong universal family $\mathscr{H}_{J}$. The goal is then to estimate the number of distinct symbols in the sample, i.e., 
\begin{displaymath}
k_{n}:=|\{x_{1},\ldots,x_{n}\}|,
\end{displaymath}
where $|\cdot|$ denotes the cardinality of a set. This problem may be considered to be complementary to that of frequency estimation but, as we shall see, it can be addressed using the same data sketch. 

Cardinality recovery has already received considerable attention in diverse applications, including spam detection \citep{Bec(06)}, text mining \citep{Bro(00)}, genomics \citep{Bre(18)}, and database management \citep{Hal(16)}. A classical approach is based on the hyperloglog algorithm \citep{Fla(83),Fla(85),Fla(07)}, which may be described as an analogous of the CMS for frequency recovery. While the hyperloglog uses random hashing to sketch the data differently compared to the CMS \citep{Cor(20)}, we will show that many of the concepts and techniques developed in Sections~\ref{sec:setup} and~\ref{sec:smoothed} can be adapted quite naturally to address the cardinality recovery problem. At the same time, however, we will see that cardinality recovery raises some additional technical challenges.

\subsection{Preliminary results}\label{sec:cardinality-preliminary}

We assume $(x_{1},\ldots,x_{n})$ to be a random sample from the model defined in \eqref{eq:exchangeable_model_hash}, and we consider estimating the number $K_{n}$ of distinct symbols in $\mathbf{X}_{n}$ through the conditional expectation of $K_{n}$ given the sketch $\mathbf{C}_{J}$ and the hash function $h$, i.e.,
\[
   \bar \varepsilon_k(P, h) := \E_{P}[K_n \mid \mathbf{C}_J = \mathbf c, h]. 
\]
Our first step is to replace  $\bar \varepsilon_k(P)$ with an estimator $\varepsilon_k(P)$ which can be expressed as a linear function of $\text{Pr}[f_{X_{n+1}}=r\,|\,\mathbf{C}_{J}=\mathbf{c}, h]$. The latter probability is available from Theorem ~\ref{teo_cond_prob} by marginalizing \eqref{cond_prob} over $h_{X_{n+1}}$.

To this end, let $M_{r,n}$ be the number of distinct symbols with frequency $r$ in $\mathbf{X}_{n}$, for $r=1,\ldots,n$, such that $K_{n}=\sum_{1\leq r\leq n}M_{r,n}$ and $n=\sum_{1\leq r\leq n}rM_{r,n}$. If $\mathcal{M}_{n,J}$ is the domain (support) of the conditional distribution of $\mathbf{M}_{n}=(M_{1,n},\ldots,M_{n,n})$ given $\mathbf{C}_{J}$, then by the law of total probability, for any $r \in \{0,1,\ldots,c_{j}\}$,
\begin{equation}\label{fund_ide}
\begin{aligned}
\text{Pr}[f_{X_{n+1}}=r\,|\,\mathbf{C}_{J}=\mathbf{c}, h] & =\sum_{\mathbf{m}\in\mathcal{M}_{n,J}} \text{Pr}[f_{X_{n+1}}=r\,|\,\mathbf{M}_{n}=\mathbf{m}, \mathbf{C}_{J}=\mathbf{c}, h] \text{Pr}[\mathbf{M}_{n}=\mathbf{m}\,|\,\mathbf{C}_{J}=\mathbf{c}, h] \\
& =\sum_{\mathbf{m}\in\mathcal{M}_{n,J}} \text{Pr}[f_{X_{n+1}}=r\,|\,\mathbf{M}_{n}=\mathbf{m}] \text{Pr}[\mathbf{M}_{n}=\mathbf{m}\,|\,\mathbf{C}_{J}=\mathbf{c}, h],
\end{aligned}
\end{equation}
where the second equality follows from the fact that $\mathbf{M}_{n}$ is a sufficient statistics for $f_{X_{n+1}}$, see, e.g., Equations (11)-(13) in \cite{Ber(23)}.
Write $\mathfrak{p}_{r,n}:=\text{Pr}[f_{X_{n+1}}=r\,|\,\mathbf{M}_{n}=\mathbf{m}]$ for the conditional probability that $X_{n+1}$ is a symbol with frequency $r$ in $\mathbf{X}_{n}$. This is known as the $r$-order coverage of $\mathbf{X}_{n}$ \citep{Goo(53)}.
To obtain $\varepsilon_k(P, h)$, we replace $\mathfrak{p}_{r,n}$ in \eqref{fund_ide} with an estimator $\tilde{\mathfrak{p}}_{r,n} = \beta_r m_{r+1}$ for some $\beta > 0$. Then we have
\[
\text{Pr}[f_{X_{n+1}}=r\,|\,\mathbf{C}_{J}=\mathbf{c}, h] \approx \sum_{\mathbf{m}\in\mathcal{M}_{n,J}} \beta_r m_{r+1} \text{Pr}[\mathbf{M}_{n}=\mathbf{m}\,|\,\mathbf{C}_{J}=\mathbf{c}, h].
\]
Now observe that the conditional expectation of $M_{r+1, n}$ given $\mathbf C_j$ equals
\[
    \E_P[M_{r+1, n} \mid \mathbf C_J = \mathbf c, h] = \sum_{\mathbf{m}\in\mathcal{M}_{n,J}} m_{r+1} \text{Pr}[\mathbf{M}_{n}=\mathbf{m}\,|\,\mathbf{C}_{J}=\mathbf{c}, h],
\]
which suggests the following approximation 
\[
    \E_P[M_{r+1, n} \mid \mathbf C_J = \mathbf c, h] \approx \frac{1}{\beta_r} \text{Pr}[f_{X_{n+1}}=r\,|\,\mathbf{C}_{J}=\mathbf{c}, h].
\]
Since $K_n = \sum_{r = 1}^n M_{r, n}$, we can exploit the relation above to obtain an estimator for $\bar \varepsilon_k(P)$ as follows
\begin{equation}\label{fund_ide1}
    \begin{aligned}
    \bar \varepsilon_k(P, h) = \E_{P}[K_{n}\,|\,\mathbf{C}_{J} = \mathbf{c}, h] &= \sum_{r=0}^{n-1}\E_{P}[M_{r+1,n}\,|\,\mathbf{C}_{J}=\mathbf{c}, h] \\
    & \approx \sum_{r=0}^{n-1} \frac{1}{\beta_r} \text{Pr}[f_{X_{n+1}}=r\,|\,\mathbf{C}_{J}=\mathbf{c}, h] =: \varepsilon_k(P, h).
\end{aligned}
\end{equation}

% Furthermore, writing the conditional expectation of $K_n$ given the sketch as
% \begin{equation}\label{fund_ide1}
%     \varepsilon_{k}(P):=\E_{P}[K_{n}\,|\,\mathbf{C}_{J}=\mathbf{c}]=\sum_{r=0}^{n-1}\E_{P}[M_{r+1,n}\,|\,\mathbf{C}_{J}=\mathbf{c}]
% \end{equation}
% and replacing $\mathfrak{p}_{r,n}$ in \eqref{fund_ide} with any estimate that is linear in $m_{r+1}$, we have that $\varepsilon_k(P)$ can be expressed as a linear function of $\text{Pr}[f_{X_{n+1}}=r\,|\,\mathbf{C}_{J}=\mathbf{c}]$. The latter probability is available from Theorem ~\ref{teo_cond_prob} by marginalizing \eqref{cond_prob} over $h_{X_{n+1}}$. See Appendix~\ref{proof_fund_ide1} for further details.

% If we replace $\mathfrak{p}_{r,n}$ in \eqref{fund_ide} with any of its estimate that is linear in $m_{r+1}$, then \eqref{fund_ide} provides
% \begin{equation}\label{fund_ide1}
% \varepsilon_{k}(P):=\E_{P}[K_{n}\,|\,\mathbf{C}_{J}=\mathbf{c}]=\sum_{r=0}^{n-1}\E_{P}[M_{r+1,n}\,|\,\mathbf{C}_{J}=\mathbf{c}]
% \end{equation}
% as a (linear) function of $\text{Pr}[f_{X_{n+1}}=r\,|\,\mathbf{C}_{J}=\mathbf{c}]$, which is available from Theorem ~\ref{teo_cond_prob} by marginalizing \eqref{cond_prob} over $h_{X_{n+1}}$. See Appendix~\ref{proof_fund_ide1} for the proof of \eqref{fund_ide1}. 

The next proposition gives $\varepsilon_{k}(P, h)$ in \eqref{fund_ide1} by making use of the Good-Turing estimate of $\mathfrak{p}_{r,n}$ \citep{Goo(53)}, which arguably is the sole nonparametric estimate of $\mathfrak{p}_{r,n}$ to be a linear function of $m_{r+1}$.

\begin{proposition}\label{teo_cond_dist}
For $n\geq1$, let $(x_{1},\ldots,x_{n})$ be a sample modeled according to \eqref{eq:exchangeable_model_hash}, and let $\mathbf{C}_{J}=\mathbf{c}$ be the corresponding sketch. If $\hat{\mathfrak{p}}_{r,n}$ is the Good-Turing estimate of  $\mathfrak{p}_{r,n}$, i.e.,
\begin{displaymath}
\hat{\mathfrak{p}}_{r,n}=(r+1)\frac{m_{r+1}}{n},
\end{displaymath}
then
\begin{equation}\label{estim_k}
\varepsilon_{k}(P, h)
= n\sum_{j=1}^{J}q_{j}\sum_{r=0}^{c_{j}}\frac{1}{r+1}\pi_{j}(r;P)
= n\sum_{j=1}^{J}\frac{q_{j}}{c_{j}+1}\sum_{s\in\mathbb{S}_{j}}\left(1-\left(1-\frac{p_{s}}{q_{j}}\right)^{c_{j}+1}\right).
\end{equation}
\end{proposition}

See Appendix~\ref{proof_teo_cond_dist} for the proof of Proposition~\ref{teo_cond_dist}. 
Similarly to $\varepsilon_{f}(P, h)$ in the context of frequency recovery, $\varepsilon_{k}(P, h)$ is a functional of the unknown distribution $P_{j}$ on $\mathbb{S}_{j}:=\{s\in\mathbb{S}\text{ : }h(s)=j\}$, induced from the distribution $P$ through random hashing. However, a notable difference between $\varepsilon_{f}(P, h)$ and $\varepsilon_{k}(P, h)$ is that the latter involves all $C_{j}$'s, and it also depends on the unknown probability $q_{j}:=\text{Pr}[h(X_{i})=j \mid h]$, which must be estimated from $\mathbf{C}_{J}$. 

An intuitive plug-in estimate of $q_{j}$ may be obtained by exploiting the fact that $\mathbf{C}_{J}$ is distributed as a Multinomial distribution with parameter $(n,q_{1},\ldots,q_{J})$. Then, for each $j \in [J]$, the maximum likelihood estimator of $q_{j}$ is  
\begin{displaymath}
\hat{q}_{j}=\frac{c_{j}}{n}.
\end{displaymath} 
Therefore, we propose to approximate \eqref{estim_k} by replacing $q_{j}$ with the estimate $\hat{q}_{j}$. This will form the cornerstone of our approach of smoothed cardinality estimation, presented in Appendix~\ref{app:cardinality-smoothed}. 

\subsection{Worst-case analysis for cardinality recovery} \label{app:cardinality-recovery-worst}

A worst-case analysis is more challenging in the context of estimating $K_n$, compared the frequency recovery problem presented in Section~\ref{sec:np-estim}, because the expression of $\varepsilon_{k}(P, h)$ in \eqref{estim_k} does not depend linearly on the counts of a single hash bucket, unlike that for $\varepsilon_{f}(P, h)$ in \eqref{cond_prob_mean}.
Therefore, we find it useful to take a slightly different approach in this section. Noting that $K_n = \sum_{r=1}^n M_{r, n}$, it is intuitive to consider the problem of nonparametric estimation of $M_{r,n}$---the number of distinct symbols with frequency $r$.

To construct a class of estimators for $M_{r,n}$, we will rely on the following approximation which can be obtained along similar lines of Proposition~\ref{teo_cond_dist},
\begin{equation}\label{gen_prop}
    \mathbb E_{P}[M_{r+1, n}\,|\,\mathbf{C}_{J}=\mathbf{c}, h] \approx \frac{n}{r+1} \sum_{j=1}^J q_j \frac{r+1}{c_j + 1} \mathbb E_{P_j}[M_{r+1, c_j +1} \mid h].
\end{equation}

Intuitively, the sole information available from the sketch $\mathbf{C}_{J}$ about $\E_{P_j}[M_{r+1, c_j +1} \mid h]$ is that $M_{r+1, c_j +1} \leq (c_j + 1) / (r+1)$.
Therefore, it is natural to consider the class of linear estimators $\hat{M}_{r+1, c_j +1} := \beta_j (c_j + 1) / (r+1)$, with parameter $\beta_j \in [0,1]$.
This gives rise to the following class of estimators for $M_{r+1,n}$:
\begin{equation}\label{eq:m_est}
    \hat M_{\mathbf \beta} = \frac{n}{r+1} \sum_{j=1}^J q_j \beta_j = \frac{n}{r+1}\langle \bm{q}, \bm{\beta} \rangle,
\end{equation}
where $\langle \bm{q}, \bm{\beta} \rangle$ is the scalar product between $\mathbf{q}=(q_{1},\ldots,q_{J})$ and the $\boldsymbol{\beta}=(\beta_{1},\ldots,\beta_{J})$. 

The next theorem provides an upper bound for the quadratic risk $R(\hat M_{\mathbf \beta}; P, h) = \mathbb E_P[(\hat M_{\mathbf \beta} - M_{r+1, n})^2 \mid h]$ associated with the class of estimators defined above, under some relatively weak assumptions on the scalar product $\langle \bm{q}, \bm{\beta} \rangle$.
\begin{theorem}\label{teo:worst_case_m}
For any fixed $r\geq0$, assume that $\langle \bm{q}, \bm{\beta} \rangle < (r+1) / 2n $. Then, $R(\hat M_{\mathbf \beta}; P, h) \leq U(\bm{q}, \bm{\beta})$, where
\begin{align} \label{eq:cardinality-worst-U}
    U(\bm{q}, \bm{\beta})&:= \left(\frac{n}{r+1}\right)^2 \langle \bm{q}, \bm{\beta} \rangle^2+  \left(1 -  2 \frac{n}{r+1} \langle \bm{q}, \bm{\beta} \rangle \right) A_{r,n}+B_{r,n},
\end{align}
for
\begin{displaymath}
A_{r,n}:=\binom{n}{r+1}  \left(\frac{r}{n-1}\right)^r \left(1 - \frac{r}{n-1}\right)^{n - r - 1}
\end{displaymath}
and
\begin{displaymath}
B_{r,n}:=\binom{n}{r+1 \ r+1}\left(\frac{r}{n - 2}\right)^{2r} \left( 1 -  \frac{2r}{n - 2}\right)^{n - 2r - 2}.
\end{displaymath}
\end{theorem}
\begin{proof}
    As a first step, we write the quadratic risk associated to the estimator \eqref{eq:m_est}. In particular, we write
\begin{align}\label{worst_case_m}
\begin{split}
	& \E_{P} \left[  \left(\hat M_{r+1, n} - M_{r+1, n}\right)^2 \mid h\right]\\
	& \quad = \left(\frac{n}{r+1}\right)^2 \langle \bm{q}, \bm{\beta} \rangle^2 - 2 \frac{n}{r+1} \langle \bm{q}, \bm{\beta} \rangle \binom{n}{r+1} \sum_{s \in \mathbb S} p_s^{r+1} (1 - p_s)^{n - r - 1} \\
	& \quad\quad +\binom{n}{r+1} \sum_{s \in \mathbb S} p_s^{r+1} (1 - p_s)^{n - r - 1}\\
	&\quad\quad +  \binom{n}{r+1 \ r+1} \sum_{s \neq v \in \mathbb S} p_s^{r+1} p_v^{r+1}(1 - p_s - p_v)^{n - 2r - 2},
      \end{split}
\end{align}
where $ \binom{n}{x \ y}  = \frac{n!}{x! y! (n - x - y)!}$. Now, the second term and the third term on the right-hand side of \eqref{worst_case_m} can be handled by using Proposition 3 in \cite{Pai(22c)}. In particular, we have
\begin{equation}\label{eq:ub1}
	\sum_{s \in \mathbb S} p_s^{r+1} (1 - p_s)^{n - r - 1}  \leq \max_{q \in [0, 1]} q^{r} (1 - q)^{n - r } = \left(\frac{r}{n-1}\right)^r \left(1 - \frac{r}{n-1}\right)^{n - r - 1}.
\end{equation}
As far as the last term of \eqref{worst_case_m} concerned, from Proposition 2 in \cite{Pai(22c)} we can write that
\begin{equation}\label{eq:ub2}
	 \sum_{s \neq v \in \mathbb S} p_s^{r+1} p_v^{r+1}(1 - p_s - p_v)^{n - 2r - 2} \leq \max_{q_1, q_2 \in \Delta_2} q_1^{r} q_2^{r}(1 - q_1 - q_2)^{n - 2r - 2} =: \psi(q_1, q_2),
\end{equation}
where $\Delta_2$ is the two-dimensional simplex $\{x, y \in [0, 1]^2 \text{ s.t. } x + y \leq 1\}$. Also, for any $q_1, q_2 \in \Delta_2$,
\begin{displaymath}
r \log q_1 + r \log q_2 = 2r \left( \frac{\log q_1 + \log q_2}{2} \right) \leq 2r \log \left( \frac{q_1 + q_2}{2} \right)
\end{displaymath}
such that
\begin{displaymath}
	\psi(q_1, q_2) \leq \psi\left(\frac{q_1 + q_2}{2}, \frac{q_1 + q_2}{2} \right),
\end{displaymath}
i.e.,
\begin{displaymath}
	\max_{q_1, q_2 \in \Delta_2} q_1^{r} q_2^{r}(1 - q_1 - q_2)^{n - 2r}  = \max_{q \in [0, 0.5]} q^{2r} (1 - 2q)^{n - 2r - 2} = \left(\frac{r}{n - 2}\right)^{2r} \left( 1 -  \frac{2r}{n - 2}\right)^{n - 2r - 2 }.
\end{displaymath}
Now, if put together the above findings, under the assumption that $\langle \bm{q}, \bm{\beta} \rangle < (r+1) / 2n $, we have
\begin{align*}
	& \E_{P} \left[  \left(\hat M_{r+1, n} - M_{r+1, n}\right)^2 \mid h \right]\\
	& \quad\leq  \left(\frac{n}{r+1}\right)^2 \langle \bm{q}, \bm{\beta} \rangle^2\\
	& \quad\quad+ \left(1 -  2 \frac{n}{r+1} \langle \bm{\pi}, \bm{\beta} \rangle \right) \binom{n}{r+1}  \left(\frac{r}{n-1}\right)^r \left(1 - \frac{r}{n-1}\right)^{n - r - 1} \\
	& \quad\quad+ \left(\frac{r}{n - 2}\right)^{2r} \left( 1 -  \frac{2r}{n - 2}\right)^{n - 2r - 2},
\end{align*}
which completes the proof.
\end{proof}
Note that the assumption $\langle \bm{q}, \bm{\beta} \rangle < (r+1) / 2n $ relates  $\bm{\beta}$ to the distribution $P$ through the probabilities $q_j$, and it is easy to check.

Since the only unknown quantity in $U(\bm{q}, \bm{\beta})$ is $\bm{q}$, it makes sense to substitute the maximum likelihood estimator $\hat q_j = c_j / n$ and then optimize the risk with respect to $\bm{\beta}$.
This leads to the conclusion that the worst-case optimal $\bm{\beta}$ is any point in $[0, 1]^J$ satisfying
\begin{displaymath}
	\sum_{j=1}^J c_j \beta_j = n \binom{n-1}{r} \left( \frac{r}{n-1}\right)^r \left( 1 - \frac{r}{n-1} \right)^{n- r }.
\end{displaymath}
Therefore, this optimization problem unfortunately does not admit a unique solution.  
Plugging any of the worst-case optimal $\bm{\beta}$ in \eqref{eq:m_est}, we can check from Eq. (2) in \cite{Gne(07)} that the estimator $\hat M_\beta$ coincides with the expected value of the number of symbols with frequency $r$ in a sample of size $n$ from a uniform distribution over $n / r$ symbols.
Therefore, as in the case of the frequency estimation problem, the worst-case estimator is rather unsatisfactory and motivates the use of smoothed estimators also for the cardinality recovery.

Alternatively, instead of substituting the maximum likelihood estimator for $q_j$, we can proceed in a ``min-max'' fashion and study the following optimization problem
\begin{align*}
    \min_{\bm{\beta} \in [0, 1]^J} \max_{\bm{q} \in \Delta_J} U(\bm{q}, \bm{\beta}),
\end{align*}
where $U(\bm{q}, \bm{\beta})$ is the upper bound for the quadratic risk defined in \eqref{eq:cardinality-worst-U} and the $\Delta_J$ is the $J$-th dimensional simplex. 
By proceeding similarly to Appendix~\ref{app:minimax}, we assume that $P$ has at most $KJ$ support points and that the induced distribution $P_{j}$ has at most $K$ points, for any $j \in [J]$. 
From the discussion above, it should not be surprising that the estimator obtained is not unique and not interesting as well.
Indeed, we performed an extensive numerical study (not reported here in the interest of brevity) which confirmed that the solution to minimax problem is $q_j = 1/J$, for each $j \in [J]$, while the optimal value of $\bm{\beta}$ is not unique.

We conclude this section by emphasizing that a limitation of this worst-case analysis is that the upper bound $U(\bm{q}, \bm{\beta})$ for the quadratic risk that is not tight, in contrast with the analysis of the frequency estimation problem from Section~\ref{sec:np-estim}.
However, the analysis presented in this section can still be informative. In fact, if $P$ is the uniform distribution over $n / r$ support points such that $q_j = 1/J$, then $|R(\hat M_{\beta}, P) - U(\bm{q}, \beta)| \leq K$ where $K$ is a constant that does not depend on $\bm{\beta}$. 
Therefore, also in the case of the cardinality recovery problem, we can see that fully nonparametric estimation is challenging and the estimators we obtain using a worst-case framework are essentially trivial. This conclusion suggests that additional modelling assumptions are needed to obtain satisfactory estimators and motivates our framework for smoothed estimation.

\subsection{Smoothed cardinality estimation} \label{app:cardinality-smoothed}

Our approach of smoothed estimation is motivated by the general difficulty of obtaining nontrivial cardinality estimates in a (fully) nonparametric setting, namely without any assumption on $P$. This is demonstrated in Appendix~\ref{app:cardinality-recovery-worst} by presenting worst-case analysis of this problem analogous to those previously discussed in Section~\ref{sec:np-estim} in the context of frequency estimation. Following in the footsteps of Section~\ref{sec:smoothed}, we take the results of Theorem~\ref{teo_smooth} as a starting point to derive estimators of $K_{n}$ as the expected value of $\varepsilon_{k}(P, h)$ with respect to the smoothing distribution in the class of NRMs. 
We focus on the DP and the NGGP. By applying \eqref{mass_dp}, the next proposition gives an explicit and simple estimator of $K_{n}$ under the DP smoothing assumption.

\begin{proposition}\label{prop_dp_cardinality}
For $n\geq1$, let $(x_{1},\ldots,x_{n})$ be a sample modeled according to \eqref{eq:exchangeable_model_hash}, and let $\mathbf{C}_{J}=\mathbf{c}$ be the corresponding sketch. If $P\sim\text{DP}(\theta,G_{0})$, for some $\theta$, it follows that
\begin{align} \label{eq:K-hat-dp}
  \hat{K}_{n}^{\text{DP}}
  := \E_{P\sim\text{DP}(\theta,G_{0}), h \sim \mathscr{H}_J}[\varepsilon_{k}(P, h)]
  = n\frac{\theta}{J}\sum_{j=1}^{J}\frac{q_{j}}{1+c_{j}}\left(\psi\left(1+c_{j}+\frac{\theta}{J}\right)-\psi\left(\frac{\theta}{J}\right)\right),
\end{align}
where $\psi$ is the digamma function, defined as the derivative of the logarithm of the Gamma function.
\end{proposition}

See Appendix~\ref{proof_prop_dp_cardinality} for a proof of Proposition~\ref{prop_dp_cardinality}. The smoothed cardinality estimator $\hat{K}_{n}^{\text{DP}}$ takes an explicit form that can be evaluated easily, besides from the issue that it still depends on the unknown probabilities $q_{j}$'s.  As anticipated above, we rely on a plug-in approach and substitute $q_j$ with $\hat{q}_{j}=c_{j}/n$. 

Finally, by applying \eqref{mass_ngg}, the next proposition gives an explicit and simple estimator of $K_{n}$ under the NGGP smoothing assumption.
\begin{proposition}\label{prop_ngg_cardinality}
For $n\geq1$, let $(x_{1},\ldots,x_{n})$ be a sample modeled according to \eqref{eq:exchangeable_model_hash}, and let $\mathbf{C}_{J}=\mathbf{c}$ be the corresponding sketch. If $P\sim\text{NGGP}(\theta,G_{0},\alpha,\tau)$, for some $(\theta,\alpha,\tau)$ it follows that
\begin{align}  \label{eq:prop-nggp-card}
\hat{K}_{n}^{\text{\tiny{(NGGP)}}}
  & := \E_{P\sim\text{NGGP}(\theta,G_{0},\alpha,\tau), h \sim \mathscr{H}_J}[\varepsilon_{k}(P, h)] \\
  &= n\sum_{j=1}^{J}\frac{q_{j}}{1+c_{j}}\int_{0}^{1}\frac{1-(1-v)^{c_{j}+1}}{v}f_{V_{\theta,\alpha,\tau}}(v)(\ddr v),
\end{align}
where $f_{V_{\theta,\alpha,\tau}}$ is the probability density function of the random variable $V_{\theta,\alpha,\tau}$ in \eqref{struct_ngg}.
\end{proposition}

See Appendix~\ref{proof_prop_ngg_cardinality} for a proof of Proposition~\ref{prop_ngg_cardinality}. Similar to the estimator obtained in Proposition~\ref{prop_dp_cardinality}, the smoothed cardinality estimator $\hat{K}_{n}^{\text{\tiny{(NGGP)}}}$ in \eqref{eq:prop-nggp-card} also involves the unknown probabilities $q_{j}$'s, which we propose to replace in practice with their the maximum likelihood estimates $\hat{q}_{j}=c_{j}/n$. Note that the expression of the estimator in  Proposition~\ref{prop_ngg_cardinality} is slightly more involved compared to that in Proposition~\ref{prop_dp_cardinality} because it entails an integral with respect to the probability density function of the random variable $V_{\theta,\alpha,\tau}$ in \eqref{struct_ngg}. However, this integral is not difficult to evaluate in practice because it can be approximated via standard Monte Carlo methods. Indeed from \eqref{struct_ngg} it straightforward to generate independent samples for $V_{\theta,\alpha,\tau}$. Note that we will simulate these values only once and reuse them for all the $J$ integrals in \eqref{eq:prop-nggp-card}, in order to reduce the variance of the estimator.

\subsection{Proofs for Appendix~\ref{sec:cardinality-preliminary}}

\subsubsection{Proof of Proposition~\ref{teo_cond_dist}}\label{proof_teo_cond_dist}

Along lines similar to the proof of Equation \eqref{fund_ide1}, we replace the coverage probability $\mathfrak{p}_{n,r}$ in \eqref{fund_ide} with the Good-Turing estimate $\hat{\mathfrak{p}}_{n,r}$. In particular, for any $r=0,1,\ldots,c_{j}$ we can write
\begin{displaymath}
\text{Pr}[f_{X_{n+1}}=r \mid \mathbf{C}_{J}=\mathbf{c}, h] = \sum_{\mathbf{m}\in\mathcal{M}_{n,J}}\frac{r+1}{n} m_{r+1}\text{Pr}[\mathbf{M}_{n}=\mathbf{m} \mid \mathbf{C}_{J}=\mathbf{c}, h],
\end{displaymath}
i.e.,
\begin{align*}
  \E_{P}[M_{r+1,n} \mid \mathbf{C}_{J}=\mathbf{c}, h]
  &=\sum_{\mathbf{m}\in\mathcal{M}_{n,J}} m_{r+1}\text{Pr}[\mathbf{M}_{n}=\mathbf{m} \mid \mathbf{C}_{J}=\mathbf{c}, h]\\
  &=\frac{n}{r+1}\text{Pr}[f_{X_{n+1}}=r \mid \mathbf{C}_{J}=\mathbf{c}, h].
\end{align*}
According to this expression, we can directly apply Theorem~\ref{teo_cond_prob}, from which if follows that 
\begin{align*}
\text{Pr}[f_{X_{n+1}}=r \mid \mathbf{C}_{J}=\mathbf{c}, h] &= \sum_{j=1}^{J}\text{Pr}[f_{X_{n+1}}=r \mid \mathbf{C}_{J}=\mathbf{c},h(X_{n+1})=j, h]\text{Pr}[h(X_{n+1})=j \mid h]\\
&=\sum_{j=1}^{J}q_{j}\pi_{j}(r;P, h)
\end{align*}
i.e., $\E_{P}[M_{r+1,n} \mid \mathbf{C}_{J}=\mathbf{c}, h]=\frac{n}{r+1}\sum_{1\leq j\leq J}q_{j}\pi_{j}(r;P, h)$. Now, since $K_{n}=\sum_{1\leq r\leq n}M_{r,n}$,
\begin{align*}
\E_{P}[K_{n} \mid \mathbf{C}_{J}=\mathbf{c}, h]&=\sum_{r=0}^{n-1}\E_{P}[M_{r+1,n} \mid \mathbf{C}_{J}=\mathbf{c}, h]\\
&=\sum_{r=0}^{n-1}\frac{n}{r+1}\sum_{j=1}^{J}q_{j}\pi_{j}(r;P, h)\\
&=\sum_{r=0}^{n-1}\frac{n}{r+1}\sum_{j=1}^{J}q_{j}{c_{j}\choose r}\sum_{s\in\mathbb{S}_{j}}\left(\frac{p_{s}}{q_{j}}\right)^{r+1}\left(1-\frac{p_{s}}{q_{j}}\right)^{c_{j}-r}\\
&=n\sum_{j=1}^{J}q_{j}\sum_{s\in\mathbb{S}_{j}}\sum_{r=0}^{n-1}\frac{1}{r+1}{c_{j}\choose r}\left(\frac{p_{s}}{q_{j}}\right)^{r+1}\left(1-\frac{p_{s}}{q_{j}}\right)^{c_{j}-r}\\
&=n\sum_{j=1}^{J}q_{j}\sum_{s\in\mathbb{S}_{j}}\sum_{r=0}^{c_{j}}\frac{1}{r+1}{c_{j}\choose r}\left(\frac{p_{s}}{q_{j}}\right)^{r+1}\left(1-\frac{p_{s}}{q_{j}}\right)^{c_{j}-r}\\
&=n\sum_{j=1}^{J}\sum_{s\in\mathbb{S}_{j}}\frac{q_{j}}{c_{j}+1}\left(1-\left(1-\frac{p_{s}}{q_{j}}\right)^{c_{j}+1}\right)
\end{align*}
with the last identity following by a summation of Binomial probabilities. The proof is completed.

%%%%%%%%%%%%%%%%%%%%%%%%%%%%%%%%
%%%%%%%%%%%%%%%%%%%%%%%%%%%%%%%%
%%%%%%%%%%%%%%%%%%%%%%%%%%%%%%%%
%%%%%%%%%%%%%%%%%%%%%%%%%%%%%%%%

\subsection{Proofs for Appendix~\ref{app:cardinality-smoothed}}

For the following proofs, it is useful to recall the expression of $\varepsilon_{k}(P)$ in terms of the $\pi_{j}(r,P)$'s. From \eqref{cond_prob_mean}:
\begin{equation}\label{pk6}
\varepsilon_{k}(P, h)=\sum_{j=1}^{J}q_{j}\sum_{r=0}^{c_{j}}\frac{n}{r+1}\pi_{j}(r,P, h)
\end{equation}

\subsubsection{Proof of Proposition~\ref{prop_dp_cardinality}}\label{proof_prop_dp_cardinality}

%With an argument similar to that of the proof of Proposition~\ref{prop_dp}, we can see that 
The estimator of $K_{n}$ in \eqref{eq:K-hat-dp} follows from \eqref{pk6} by replacing the probabilities $\pi_{j}(r,P, h)$'s with the estimated probabilities $\pi_{j}(r)$ in \eqref{mass_dp}. That is,
\begin{align*}
\hat{K}_{n}^{\text{DP}}&=\sum_{j=1}^{J}q_{j}\sum_{r=0}^{c_{j}}\frac{n}{r+1}\int_{0}^{1}{c_{j}\choose r}v^{r}(1-v)^{c_{j}-r}f_{V_{\theta}}(v)\ddr v\\
&=n\sum_{j=1}^{J}\frac{q_{j}}{c_{j}+1}\int_{0}^{1}\frac{1-(1-v)^{c_{j}+1}}{v}f_{V_{\theta}}(v)\ddr v\\
&=n\sum_{j=1}^{J}\frac{q_{j}}{c_{j}+1}\left(\frac{\theta}{J}\psi\left(1+c_{j}+\frac{\theta}{J}\right)-\frac{\theta}{J}\psi\left(\frac{\theta}{J}\right)\right),
\end{align*}
where $\psi$ denotes the digamma function, defined as $\psi(z)=\frac{\ddr}{\ddr z}\log\Gamma(z)$.

\subsubsection{Proof of Proposition~\ref{prop_ngg_cardinality}}\label{proof_prop_ngg_cardinality}

%With an argument similar to that of the proof of Proposition~\ref{prop_nggp}, 
The estimator of $K_{n}$ follows from \eqref{pk6} by replacing the probabilities $\pi_{j}(r,P, h)$'s with the estimated probabilities $\pi_{j}(r)$ in \eqref{mass_ngg}. That is,
\begin{align*}
\hat{K}_{n}^{\text{\tiny{(NGGP)}}}&=\sum_{j=1}^{J}q_{j}\sum_{r=0}^{c_{j}}\frac{n}{r+1}\int_{0}^{1}{c_{j}\choose r}v^{r}(1-v)^{c_{j}-r}f_{V_{\theta,\alpha,\tau}}(v)\ddr v\\
&=n\sum_{j=1}^{J}\frac{q_{j}}{c_{j}+1}\int_{0}^{1}\frac{1-(1-v)^{c_{j}+1}}{v}f_{V_{\theta,\alpha,\tau}}(v)\ddr v,
\end{align*}
and this can not be simplified further by solving the integral in $v$.

\section{Details on sketches with multiple hash functions}\label{app:multihash}

\subsection{Preliminary results}

Under model \eqref{eq:exchangeable_model_multhash}, we consider the problem of estimating the empirical frequency $f_{X_{n+1}}$ from the sketch $\mathbf{C}_{M,J}$. In particular, following the approach developed in Section~\ref{sec:setup} for $M=1$, this problem requires to compute the conditional distribution of $f_{X_{n+1}}$ given $\mathbf{C}_{M,D}$ and the buckets $\mathbf{h}_{M}(X_{n+1}):=(h_{1}(X_{n+1}),\ldots,h_{M}(X_{n+1}))$ into which $X_{n+1}$ is hashed.

It is possible to derive explicit, though cumbersome, expressions for the distribution of $\mathbf{C}_{M,J}$ and for the conditional distribution of $f_{X_{n+1}}$ given $\mathbf{C}_{M,D}$ and $\mathbf{h}_{M}(X_{n+1})$,
as shown in Appendix~\ref{distr_mult} below.  
Here, we present a simplified version of these distributions under the following two additional assumptions.
\begin{itemize}
\item[A1)] Independence from random hashing. For any $(j_{1},\ldots,j_{M})\in\{1,\ldots,J\}^{M\times n}$,
\begin{displaymath}
\text{Pr}[h_{1}(X_{i})=j_{1},\ldots,h_{M}(X_{i})=j_{M}]=\prod_{l=1}^{M}\text{Pr}[h_{l}(X_{i})=j_{l}].
\end{displaymath}
\item[A2)] Identity in distribution with respect to hashing. For any pair of indices $i_{1}\neq i_{2}$,
\begin{displaymath}
\text{Pr}[h_{l}(X_{i_{1}})=j_{l}]=\text{Pr}[h_{l}(X_{i_{2}})=j_{l}].
\end{displaymath}
\end{itemize}
Under the model in \eqref{eq:exchangeable_model_multhash}, assuming that Assumptions A1) and A2) also hold, we can write
\begin{equation}\label{multinom_multiple}
\text{Pr}[\mathbf{C}_{M,J}=\mathbf{c} \mid \mathbf{h}_{M}]=\mathscr{C}(n,\mathbf{c})\prod_{l=1}^{M}\prod_{j=1}^{J}q_{l,j}^{c_{l,j}},
\end{equation}
where the $q_{l,j}$'s are probabilities defined as $q_{l,j}: =\text{Pr}[h_{X_{i}}=j]$, for any $l\in[M]$ and $j\in[J]$, and $\mathscr{C}(n,\mathbf{c})$ is a (combinatorial) coefficient, defined in Appendix~\ref{appendix_multinom_multiple}, that reduces to the Multinomial coefficient if $M=1$. 
See Appendix~\ref{appendix_multinom_multiple} for the proof of Equation \eqref{multinom_multiple}. 

\begin{theorem}\label{teo_cond_prob_multiple}
For $n\geq1$ let $(x_{1},\ldots,x_{n})$ be a random sample modeled according to \eqref{eq:exchangeable_model_multhash}, and $\mathbf{C}_{M,J}=\mathbf{c}$ the corresponding sketch. Assume that Assumptions A1) and A2) also hold. If $\mathbb{S}_{j_{l}}:=\{s\in\mathbb{S}\text{ : }h_{l}(s)=j_{l}\}$, for any $j_{l}\in [J]$ and $l \in [M]$, then for any $r \in \{0,1,\ldots,n\}$,
\begin{align}\label{cond_prob_mult}
\pi_{\mathbf{j}}(r;P, \mathbf{h}_{M})&:=\prob[f_{X_{n+1}}=r \mid \mathbf{C}_{M,J}=\mathbf{c},\mathbf{h}_{M}(X_{n+1})=\mathbf{j}, \mathbf{h}_{M}]\\
&\notag={n\choose r}\frac{\mathscr{C}(n-r,\mathbf{c})}{\mathscr{C}(n,\mathbf{c})}\prod_{l=1}^{M}q_{l,j_{l}}^{\left(\frac{1}{M}-1\right)(r+1)}\sum_{s\in\mathbb{S}_{j_{l}}}\left(\frac{p_{s}}{q_{l,j_{l}}}\right)^{\frac{r+1}{M}}\left(1-\frac{p_{s}}{q_{l,j_{l}}}\right)^{c_{l,j_{l}}-r}I(r\leq c_{l,j_{l}}).
\end{align}
\end{theorem}

See Appendix~\ref{app_teo_cond_prob_multiple} for the proof of Theorem~\ref{teo_cond_prob_multiple}. If $M=1$, then Theorem~\ref{teo_cond_prob_multiple} reduces to Theorem~\ref{teo_cond_prob}, as $\mathscr{C}(n,\mathbf{c})$ reduces to the Multinomial coefficient. From here, an intuitive approach to estimate the empirical frequency $f_{X_{n+1}}$ from the sketch $\mathbf{C}_{M,J}$ is to follow in the footsteps of Sections~\ref{sec:setup} and~\ref{sec:smoothed}, considering the means of the distribution \eqref{cond_prob_mult}, i.e.,
\begin{displaymath}
\varepsilon_{f}(P, \mathbf{h}_{M}):=\sum_{r=0}^{n} r \pi_{\mathbf{j}}(r; \mathbf{h}_{M}P, \mathbf{h}_{M}).
\end{displaymath}
However, the expression in \eqref{cond_prob_mult} anticipates that the resulting estimator of $f_{X_{n+1}}$ is not practical because it depends on the combinatorial coefficients $\mathscr{C}(n-r,\mathbf{c})$, for $r \in \{0,1,\ldots,n\}$, which are computationally expensive to evaluate if $n-r$ is large. 

\subsection{Distributions under the model \eqref{eq:exchangeable_model_multhash}}\label{distr_mult}

We determine the distribution of $\mathbf{C}_{M,J}$ and the conditional distribution of $f_{X_{n+1}}$, given $\mathbf{C}_{M,J}$ and $\mathbf{h}(X_{n+1})$. Note that $\mathbf{h}(X_{n+1})\in\{1,\ldots,J\}^{M}$, such that for any  $\mathbf{j}\in\{1,\ldots,J\}^{M}$, we set
\begin{equation}\label{probab_sketched}
q_{\mathbf{j}}=\text{Pr}[\mathbf{h}(X_{n+1})=\mathbf{j} \mid \mathbf{h}_{M}].
\end{equation}
The distribution of $\mathbf{C}_{M,J}$ is a generalization of the Multinomial distribution. In particular, because $\mathbf{X}_{n}$ is a random sample from the discrete distribution $P$ on $\mathbb{S}$, then we can write,
\begin{align}\label{new_for}
&\text{Pr}[\mathbf{C}_{M,J}=\mathbf{c} \mid \mathbf{h}_{M}] =\sum_{(j_{1}(1),\ldots,j_{M}(1),\ldots,(j_{1}(n),\ldots,j_{M}(n))\in S_{n,\mathbf{c}}}\prod_{i=1}^{n}q_{\mathbf{j}(i)}
\end{align}
where
\begin{align*}
&S_{n,\mathbf{c}}=\left\{\{1,\ldots,J\}^{M\times n}\text{ : }\sum_{i=1}^{n}I(j_{1}(i)=1)=c_{1,1},\ldots,\sum_{i=1}^{n}I(j_{M}(i)=J)=c_{M,J}\right\}.
\end{align*}
Now, we consider the conditional distribution of $f_{X_{n+1}}$, given $\mathbf{C}_{M,J}$ and $\mathbf{h}(X_{n+1})$, that is
\begin{multline}\label{estimate_freq_mul}
\text{Pr}[f_{X_{n+1}}=r \mid  \mathbf{C}_{M,J}=\mathbf{c},\mathbf{h}(X_{n+1})=\mathbf{j}, \mathbf{h}_{M}] = \\
\qquad \frac{\text{Pr}[f_{X_{n+1}}=r,\mathbf{C}_{M,J}=\mathbf{c},\mathbf{h}(X_{n+1})=\mathbf{j} \mid \mathbf{h}_{M}] }{ \text{Pr}[\mathbf{C}_{M,J}=\mathbf{c},\mathbf{h}(X_{n+1})=\mathbf{j} \mid \mathbf{h}_{M}]}
\end{multline}
We compute \eqref{estimate_freq_mul} by evaluating (separately) the numerator and the denominator. We start by evaluating the denominator of the conditional probability in \eqref{estimate_freq_mul}. From \eqref{new_for}, we write
\begin{align}\label{denom_multiple}
&\text{Pr}[\mathbf{C}_{M,J}=\mathbf{c},\mathbf{h}(X_{n+1})=\mathbf{j} \mid \mathbf{h}_{M}]\\
&\notag\quad=q_{\mathbf{j}}\text{Pr}[\mathbf{C}_{M,J}=\mathbf{c} \mid \mathbf{h}_{M}]\\
&\notag\quad=q_{\mathbf{j}}\sum_{(j_{1}(1),\ldots,j_{m}(1),\ldots,(j_{1}(n),\ldots,j_{m}(n))\in S_{n,\mathbf{c}}}\prod_{i=1}^{n}q_{\mathbf{j}(i)}
\end{align}
Now, we evaluate the numerator of the conditional probability in \eqref{estimate_freq_mul}. Let us define the event $B(n,r)=\{X_{1}=\cdots=X_{r}=X_{n+1},\{X_{r+1},\ldots,X_{n}\}\cap\{X_{n+1}\}=\emptyset\}$. Then, we write
\begin{align*}
&\text{Pr}[f_{X_{n+1}}=r,\mathbf{C}_{M,J}=\mathbf{c},\mathbf{h}(X_{n+1})=\mathbf{j} \mid \mathbf{h}_{M}]\\
&\quad=\text{Pr}\left[f_{X_{n+1}}=r,\sum_{i=1}^{n}I(h_{1}(X_{i})=1)=c_{1,1},\ldots,\sum_{i=1}^{n}I(h_{1}(X_{i})=J)=c_{1,J},\ldots\right.\\
&\quad\quad\ldots,\left.\sum_{i=1}^{n}I(h_{l}(X_{i})=1)=c_{l,1},\ldots,\sum_{i=1}^{n}I(h_{l}(X_{i})=J)=c_{l,J},\ldots\right.\\
&\quad\quad\quad\ldots,\left.\sum_{i=1}^{n}I(h_{M}(X_{i})=1)=c_{M,1},\ldots,\sum_{i=1}^{n}I(h_{M}(X_{i})=J)=c_{M,J},\mathbf{h}(X_{n+1})=\mathbf{j} \mid \mathbf{h}_{M}\right]\\
&\quad={n\choose r}\sum_{s\in\mathbb{S}}\text{Pr}\left[B(n,r),X_{n+1}=s,\sum_{i=1}^{n}I(h_{1}(X_{i})=1)=c_{1,1},\ldots,\sum_{i=1}^{n}I(h_{1}(X_{i})=J)=c_{1,J},\ldots\right.\\
&\quad\quad\ldots,\left.\sum_{i=1}^{n}I(h_{l}(X_{i})=1)=c_{l,1},\ldots,\sum_{i=1}^{n}I(h_{l}(X_{i})=J)=c_{l,J},\ldots\right.\\
&\quad\quad\ldots,\left.\sum_{i=1}^{n}I(h_{M}(X_{i})=1)=c_{M,1},\ldots,\sum_{i=1}^{n}I(h_{M}(X_{i})=J)=c_{M,J},\mathbf{h}(X_{n+1})=\mathbf{j} \mid \mathbf{h}_{M}\right]\\
&\quad={n\choose r}\sum_{s\in\mathbb{S}}\text{Pr}\left[B(n,r),X_{n+1}=s,\sum_{i=r+1}^{n}I(h_{1}(X_{i})=1)=c_{1,1},\ldots\right.\\
&\quad\quad\ldots,\left.\sum_{i=r+1}^{n}I(h_{1}(X_{i})=j_{1})=c_{1,j_{1}}-r,\ldots,\sum_{i=r+1}^{n}I(h_{1}(X_{i})=J)=c_{1,J},\ldots\right.\\
&\quad\quad\ldots,\left.\sum_{i=r+1}^{n}I(h_{l}(X_{i})=1)=c_{l,1},\ldots\right.\\
&\quad\quad\ldots,\left.\sum_{i=r+1}^{n}I(h_{l}(X_{i})=j_{l})=c_{1,j_{l}}-r,\ldots,\sum_{i=r+1}^{n}I(h_{l}(X_{i})=J)=c_{l,J},\ldots\right.\\
&\quad\quad\ldots,\left.\sum_{i=r+1}^{n}I(h_{M}(X_{i})=1)=c_{M,1},\ldots\right.\\
&\quad\quad\ldots,\left.,\sum_{i=r+1}^{n}I(h_{M}(X_{i})=j_{M})=c_{1,j_{M}}-r,\ldots,\sum_{i=r+1}^{n}I(h_{M}(X_{i})=J)=c_{M,J},\mathbf{h}(X_{n+1})=\mathbf{j} \mid \mathbf{h}_{M}\right].
\end{align*}
Accordingly, the distribution of $(f_{X_{n+1}},\mathbf{C}_{M,J},\mathbf{h}(X_{n+1}))$ is completely determined by the distribution of $(X_{1},\ldots,X_{n},X_{n+1})$. Therefore, from the previous equation, we can write that
\begin{align*}
&\text{Pr}[f_{X_{n+1}}=r,\mathbf{C}_{M,J}=\mathbf{c},\mathbf{h}(X_{n+1})=\mathbf{j} \mid \mathbf{h}_{M}]\\
&\quad={n\choose r}\sum_{s\in\mathbb{S}}\text{Pr}\left[X_{1}=\cdots=X_{r}=X_{n+1}=s,X_{r+1}\neq s,\ldots,X_{n}\neq s\right]\\
&\quad\quad\times\text{Pr}\left[\sum_{i=r+1}^{n}I(h_{1}(X_{i})=1)=c_{1,1},\ldots,\sum_{i=r+1}^{n}I(h_{1}(X_{i})=j_{1})=c_{1,j_{1}}-r,\ldots\right.\\
&\quad\quad\quad\left.\ldots,\sum_{i=r+1}^{n}I(h_{1}(X_{i})=J)=c_{1,J},\ldots\right.\\
&\quad\quad\quad\ldots,\left.\sum_{i=r+1}^{n}I(h_{l}(X_{i})=1)=c_{l,1},\ldots,\sum_{i=r+1}^{n}I(h_{l}(X_{i})=j_{l})=c_{1,j_{l}}-r,\ldots\right.\\
&\quad\quad\quad\left.\ldots,\sum_{i=r+1}^{n}I(h_{l}(X_{i})=J)=c_{l,J},\ldots\right.\\
&\quad\quad\quad\ldots,\left.\sum_{i=r+1}^{n}I(h_{M}(X_{i})=1)=c_{M,1},\ldots,\sum_{i=r+1}^{n}I(h_{M}(X_{i})=j_{M})=c_{1,j_{M}}-r,\ldots\right.\\
&\quad\quad\quad\left.\ldots,\sum_{i=r+1}^{n}I(h_{M}(X_{i})=J)=c_{M,J},\mathbf{h}(X_{n+1})=\mathbf{j}\mid X_{r+1}\neq s,\ldots,X_{n}\neq s,X_{n+1}=s \mid \mathbf{h}_{M}\right]\\
&\quad={n\choose r}\sum_{s\in\mathbb{S}}p_{s}^{r+1}(1-p_{s})^{n-r}I(s\in\mathbb{S}\text{ : }(\mathbf{h}(s)=\mathbf{j}))\\
&\quad\quad\times\text{Pr}\left[\sum_{i=r+1}^{n}I(h_{1}(X_{i})=1)=c_{1,1},\ldots,\sum_{i=r+1}^{n}I(h_{1}(X_{i})=j_{1})=c_{1,j_{1}}-r,\ldots\right.\\
&\quad\quad\quad\left.\ldots,\sum_{i=r+1}^{n}I(h_{1}(X_{i})=J)=c_{1,J},\ldots\right.\\
&\quad\quad\quad\ldots,\left.\sum_{i=r+1}^{n}I(h_{l}(X_{i})=1)=c_{l,1},\ldots,\sum_{i=r+1}^{n}I(h_{l}(X_{i})=j_{l})=c_{1,j_{l}}-r,\ldots\right.\\
&\quad\quad\quad\left.\ldots,\sum_{i=r+1}^{n}I(h_{l}(X_{i})=J)=c_{l,J},\ldots\right.\\
&\quad\quad\quad\ldots,\left.\sum_{i=r+1}^{n}I(h_{M}(X_{i})=1)=c_{M,1},\ldots,\sum_{i=r+1}^{n}I(h_{M}(X_{i})=j_{M})=c_{1,j_{M}}-r,\ldots\right.\\
&\quad\quad\quad\left.\ldots,\sum_{i=r+1}^{n}I(h_{M}(X_{i})=J)=c_{M,J}\mid X_{r+1}\neq s,\ldots,X_{n}\neq s,X_{n+1}=s, \mathbf{h}_{M}\right].
\end{align*}
Based on the last identity, we can define the set $\mathbb{S}_{\mathbf{j}}=\{s\in\mathbb{S}\text{ : } (\mathbf{h}(s)=\mathbf{j})\}$, and then we write
\begin{align}\label{num_multiple}
&\text{Pr}[f_{X_{n+1}}=r,\mathbf{C}_{M,J}=\mathbf{c},\mathbf{h}(X_{n+1})=\mathbf{j} \mid \mathbf{h}_{M}]\\
&\notag\quad={n\choose r}\sum_{s\in\mathbb{S}_{\mathbf{j}}}p_{s}^{r+1}(1-p_{s})^{n-r}\\
&\notag\quad\quad\times\sum_{(j_{1}(1),\ldots,j_{M}(1),\ldots,(j_{1}(n-r),\ldots,j_{M}(n-r))\in S_{n-r,\mathbf{c}}}\prod_{i=1}^{n-r}\frac{(q_{\mathbf{j}(i)}-p_{s}I(\mathbf{j}=\mathbf{j}(i)))}{(1-p_{s})}\\
&\notag\quad={n\choose r}\sum_{s\in\mathbb{S}_{\mathbf{j}}}p_{s}^{r+1}\sum_{(j_{1}(1),\ldots,j_{M}(1),\ldots,(j_{1}(n-r),\ldots,j_{M}(n-r))\in S_{n-r,\mathbf{c}}}\prod_{i=1}^{n-r}(q_{\mathbf{j}(i)}-p_{s}I(\mathbf{j}=\mathbf{j}(i))),
\end{align}
where
\begin{align*}
S_{n-r,\mathbf{c}}=\left\{\{1,\ldots,J\}^{M\times n}\text{ : }\sum_{i=r+1}^{n}I(j_{1}(i)=1)=c_{1,1},\ldots,\sum_{i=r+1}^{n}I(j_{M}(i)=J)=c_{M,J}\right\}.
\end{align*}
Therefore, according to \eqref{estimate_freq_mul}, which is combined with \eqref{denom_multiple} with and \eqref{num_multiple}, we obtain the following:
\begin{align}\label{final_mult_gen}
&\text{Pr}[f_{X_{n+1}}=r \mid \mathbf{C}_{M,J}=\mathbf{c},\mathbf{h}(X_{n+1})=\mathbf{j}, \mathbf{h}_{M}]\\
&\notag\quad={n\choose r}\sum_{s\in\mathbb{S}_{\mathbf{j}}}p_{s}^{r+1}\frac{\sum_{(j_{1}(1),\ldots,j_{M}(1),\ldots,(j_{1}(n-r),\ldots,j_{M}(n-r))\in S_{n-r,\mathbf{c}}}\prod_{i=1}^{n-r}(q_{\mathbf{j}(i)}-p_{s}I(\mathbf{j}=\mathbf{j}(i)))}{q_{\mathbf{j}}\sum_{(j_{1}(1),\ldots,j_{M}(1),\ldots,(j_{1}(n),\ldots,j_{M}(n))\in S_{n,\mathbf{c}}}\prod_{i=1}^{n}q_{\mathbf{j}(i)}}.
\end{align}

\subsection{Proof of Equation \eqref{multinom_multiple}}\label{appendix_multinom_multiple}

The proof relies on a direct application of the assumptions A1) and A2) to \eqref{new_for}. In particular, used these assumptions, we can factorizes the probabilities $q_{\mathbf{j}(i)}$'s, and then we write
\begin{align*}
\prod_{i=1}^{n}q_{\mathbf{j}(i)}&=q_{1,j_{1}}(1)\cdots q_{M,j_{M}}(1)\cdots q_{1,j_{1}}(n)\cdots q_{M,j_{M}}(n)
%&=q_{1,j_{1}}^{\sum_{i=1}^{n}I(j_{1}(i)=j_{1})}\cdots q_{m,j_{m}}^{\sum_{i=1}^{n}I(j_{m}(i)=j_{m})}\\
=\prod_{l=1}^{M}q_{l,1}^{c_{l,1}}\cdots q_{l,J}^{c_{l,J}},
%&=q_{1,j_{1}}^{c_{1,j_{1}}}\cdots q_{m,j_{m}}^{c_{m,j_{m}}}.
\end{align*}
i.e.,
\begin{align}\label{new_for1}
\text{Pr}[\mathbf{C}_{M,J}=\mathbf{c}]&=\sum_{(j_{1}(1),\ldots,j_{M}(1),\ldots,(j_{1}(n),\ldots,j_{M}(n))\in S_{n,\mathbf{c}}}\prod_{i=1}^{n}q_{\mathbf{j}(i)}\\
%&\notag\quad=\sum_{(j_{1}(1),\ldots,j_{m}(1),\ldots,(j_{1}(n),\ldots,j_{m}(n))\in S_{n,(\mathbf{c}_{1},\ldots,\mathbf{c}_{m})}}q_{1,j_{1}}^{c_{1,j_{1}}}\cdots q_{m,j_{m}}^{c_{m,j_{m}}}\\
&\notag\quad=\left(\sum_{(j_{1}(1),\ldots,j_{M}(1),\ldots,(j_{1}(n),\ldots,j_{M}(n))\in S_{n,\mathbf{c}}}1\right)\prod_{l=1}^{M}q_{l,1}^{c_{l,1}}\cdots q_{l,J}^{c_{l,J}},
\end{align}
where 
\begin{displaymath}
\mathscr{C}(n,\mathbf{c})=\left(\sum_{(j_{1}(1),\ldots,j_{M}(1),\ldots,(j_{1}(n),\ldots,j_{;}(n))\in S_{n,\mathbf{c}}}1\right)
\end{displaymath}
is a (combinatorial) coefficient, depending on $n$, $M$, $J$ and the $c_{l,j}$'s. This completes the proof.

\subsection{Proof of Theorem~\ref{teo_cond_prob_multiple}}\label{app_teo_cond_prob_multiple}

The proof relies on a direct application of the assumptions A1) and A2) to \eqref{final_mult_gen}. In particular, used these assumptions, we start by rewriting the denominator of \eqref{final_mult_gen} as follows:
\begin{align}\label{den_mult_simpl}
&\text{Pr}[\mathbf{C}_{M,J}=\mathbf{c},\mathbf{h}(X_{n+1})=\mathbf{j} \mid \mathbf{h}_{M}]\\
&\notag\quad=q_{1,j_{1}}\cdots q_{M,j_{M}}\\
&\notag\quad\quad\times\mathscr{C}(n,\mathbf{c})\prod_{l=1}^{M}q_{l,1}^{c_{l,1}}\cdots q_{l,J}^{c_{l,J}}\\
&\notag\quad=\mathscr{C}(n,\mathbf{c})\prod_{l=1}^{M}q_{l,1}^{c_{l,1}+I(j_{l}=1)}\cdots q_{l,J}^{c_{l,J}+I(j_{l}=J)}\\
&\notag\quad=\mathscr{C}(n,\mathbf{c}) q_{1,j_{1}}^{c_{1,j_{1}}+1}\cdots q_{M,j_{M}}^{c_{M,j_{M}}+1}\prod_{l=1}^{M}\prod_{1\leq t\neq j_{l}\leq J}q_{l,t}^{c_{l,t}}.
\end{align}
Along similar lines, under the assumptions A1) and A2) we rewrite the denominator of \eqref{final_mult_gen} as follows:
\begin{align}\label{num_mult_simpl}
&\text{Pr}[f_{X_{n+1}}=r,\mathbf{C}_{M,J}=\mathbf{c},\mathbf{h}(X_{n+1})=\mathbf{j} \mid \mathbf{h}_{M}]\\
&\notag\quad={n\choose r}\sum_{s\in\mathbb{S}_{\mathbf{j}}}p_{s}^{r+1}\\
&\notag\quad\quad\times\mathscr{C}(n-r,\mathbf{c})\\
&\notag\quad\quad\quad\times\prod_{l=1}^{M}(q_{l,1}-p_{s}I(j_{l}=1))^{c_{l,1}-rI(j_{l}=1)}\cdots (q_{l,J}-p_{s}I(j_{l}=J))^{c_{l,J}-rI(j_{l}=J)}\\
&\notag\quad={n\choose r}\sum_{s\in\mathbb{S}_{\mathbf{j}}}p_{s}^{r+1}\\
&\notag\quad\quad\times\mathscr{C}(n-r,\mathbf{c})\\
&\notag\quad\quad\quad\times\prod_{l=1}^{M}q_{l,1}^{c_{l,1}-rI(j_{l}=1)}\cdots q_{l,J}^{c_{l,J}-rI(j_{l}=J)}\\
&\notag\quad\quad\quad\times\left(1-\frac{p_{s}}{q_{l,1}}I(j_{l}=1)\right)^{c_{l,1}-rI(j_{j}=1)}\cdots\left(1-\frac{p_{s}}{q_{l,J}}I(j_{l}=J)\right)^{c_{l,J}-rI(j_{l}=J)}\\
&\quad={n\choose r}\sum_{s\in\mathbb{S}_{\mathbf{j}}}p_{s}^{r+1}\\
&\quad\quad\times\mathscr{C}(n-r,\mathbf{c})\\
&\quad\quad\quad\times q_{1,j_{1}}^{c_{1,j_{1}}-r}\left(1-\frac{p_{s}}{q_{1,j_{1}}}\right)^{c_{1,j_{1}}-r}\cdots q_{M,j_{M}}^{c_{M,j_{M}}-r}\left(1-\frac{p_{s}}{q_{M,j_{M}}}\right)^{c_{M,j_{M}}-r}\prod_{l=1}^{M}\prod_{1\leq t\neq j_{l}\leq J}q_{l,t}^{c_{l,t}}.
\end{align}
Finally, by combining \eqref{den_mult_simpl} and \eqref{den_mult_simpl}, the conditional probability \eqref{final_mult_gen} can be written as
\begin{align*}
&\text{Pr}[f_{X_{n+1}}=r \mid \mathbf{C}_{M,J}=\mathbf{c},\mathbf{h}(X_{n+1})=\mathbf{j}, \mathbf{h}_{M}]\\
&\quad={n\choose r}\frac{\mathscr{C}(n-r,\mathbf{c})}{\mathscr{C}(n,\mathbf{c})}\\
&\quad\quad\times\sum_{s\in\mathbb{S}_{\mathbf{j}}}p_{s}^{r+1}\frac{q_{1,j_{1}}^{c_{1,j_{1}}-r}\left(1-\frac{p_{s}}{q_{1,j_{1}}}\right)^{c_{1,j_{1}}-r}\cdots q_{M,j_{M}}^{c_{M,j_{M}}-r}\left(1-\frac{p_{s}}{q_{M,j_{M}}}\right)^{c_{M,j_{M}}-r}\prod_{l=1}^{M}\prod_{1\leq t\neq j_{l}\leq J}q_{l,t}^{c_{l,t}}}{q_{1,j_{1}}^{c_{1,j_{1}}+1}\cdots q_{M,j_{M}}^{c_{M,j_{M}}+1}\prod_{l=1}^{M}\prod_{1\leq t\neq j_{l}\leq J}q_{l,t}^{c_{l,t}}}\\
&\quad={n\choose r}\frac{\mathscr{C}(n-r,\mathbf{c})}{\mathscr{C}(n,\mathbf{c})}\\
&\quad\quad\times\sum_{s\in\mathbb{S}_{\mathbf{j}}}p_{s}^{(r+1)(1-M)}\left(\frac{p_{s}}{q_{1,j_{1}}}\right)^{r+1}\left(1-\frac{p_{s}}{q_{1,j_{1}}}\right)^{c_{1,j_{1}}-r}\cdots \left(\frac{p_{s}}{q_{M,j_{M}}}\right)^{r+1}\left(1-\frac{p_{s}}{q_{M,j_{M}}}\right)^{c_{M,j_{M}}-r}\\
&\quad={n\choose r}\frac{\mathscr{C}(n-r,\mathbf{c})}{\mathscr{C}(n,\mathbf{c})}\\
&\quad\quad\times\sum_{s\in\mathbb{S}_{j_{1}}}\cdots\sum_{s\in\mathbb{S}_{j_{M}}}p_{s}^{(r+1)(1-M)}\left(\frac{p_{s}}{q_{1,j_{1}}}\right)^{r+1}\left(1-\frac{p_{s}}{q_{1,j_{1}}}\right)^{c_{1,j_{1}}-r}\cdots \left(\frac{p_{s}}{q_{M,j_{M}}}\right)^{r+1}\left(1-\frac{p_{s}}{q_{M,j_{M}}}\right)^{c_{M,j_{M}}-r}\\
&\quad={n\choose r}\frac{\mathscr{C}(n-r,\mathbf{c})}{\mathscr{C}(n,\mathbf{c})}\\
&\quad\quad\times\sum_{s\in\mathbb{S}_{j_{1}}}p_{s}^{\left(\frac{1}{M}-1\right)(r+1)}\left(\frac{p_{s}}{q_{1,j_{1}}}\right)^{r+1}\left(1-\frac{p_{s}}{q_{1,j_{1}}}\right)^{c_{1,j_{1}}-r}\cdots\\
&\quad\quad\cdots\times \sum_{s\in\mathbb{S}_{j_{M}}}p_{s}^{\left(\frac{1}{M}-1\right)(r+1)}\left(\frac{p_{s}}{q_{M,j_{M}}}\right)^{r+1}\left(1-\frac{p_{s}}{q_{M,j_{M}}}\right)^{c_{M,j_{M}}-r}\\
&\quad={n\choose r}\frac{\mathscr{C}(n-r,\mathbf{c})}{\mathscr{C}(n,\mathbf{c})}q_{1,j_{1}}^{\left(\frac{1}{M}-1\right)(r+1)}\cdots q_{M,j_{M}}^{\left(\frac{1}{M}-1\right)(r+1)}\\
&\quad\quad\times\sum_{s\in\mathbb{S}_{j_{1}}}\left(\frac{p_{s}}{q_{1,j_{1}}}\right)^{\frac{r+1}{m}}\left(1-\frac{p_{s}}{q_{1,j_{1}}}\right)^{c_{1,j_{1}}-r}\cdots \sum_{s\in\mathbb{S}_{j_{M}}}\left(\frac{p_{s}}{q_{M,j_{M}}}\right)^{\frac{r+1}{M}}\left(1-\frac{p_{s}}{q_{M,j_{M}}}\right)^{c_{M,j_{M}}-r},
\end{align*}
which completes the proof.

\section{Further Simulations}

\subsection{Cardinality recovery}

We consider the same setup as Section~\ref{sec:ex_singlehash}. As the evaluation metric, for the cardinality estimator we report the absolute error between the true number of distinct symbols and the estimated one.

We now consider the cardinality estimation problem, assuming the same data generating processes as above.
We set $n \in \{100, 1000, 10000, 100000\}$. The data are hashed by a random hash function of with $J=128$.
Figure~\ref{fig:card_singlehash} reports the true and estimated cardinalities for a simulated setting, while Tables~\ref{tab:card_singlehash_py} and~\ref{tab:card_singlehash_zipf} report the errors achieved by the different smoothing distributions, averaging across 50 independent repetitions. 
The NGGP smoothing generally outperforms the DP when the discount parameter of the Pitman-Yor process is greater than zero. Conversely, the DP yields slightly better performances when the discount is equal to zero. 
Similarly, the NGGP performs better when the parameter of the Zipf distribution is closer to one (i.e., heavier tails of the distribution), while the DP gives better estimates if $c=2.2$ .

\begin{figure}[!htb]
    \centering
    \includegraphics[width=\linewidth]{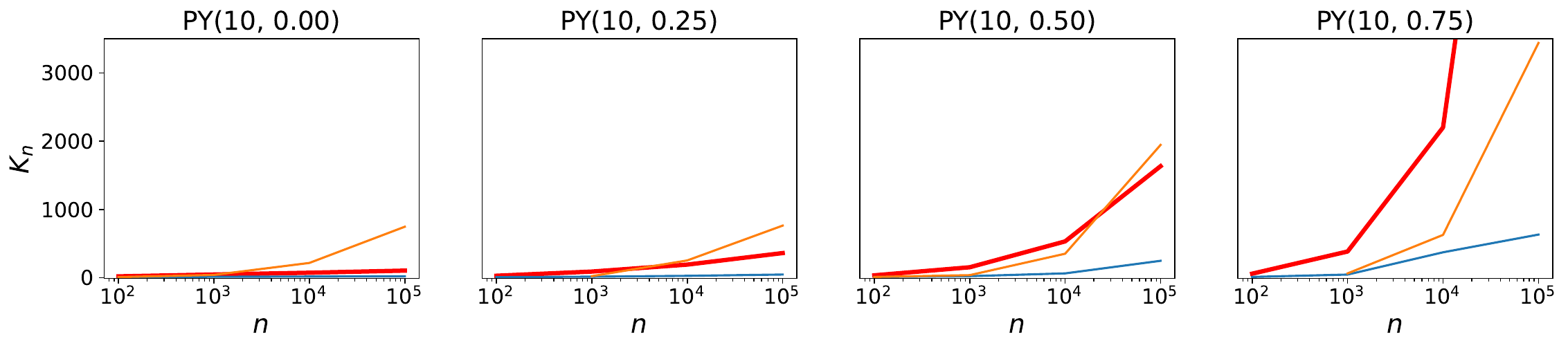}
    \includegraphics[width=\linewidth]{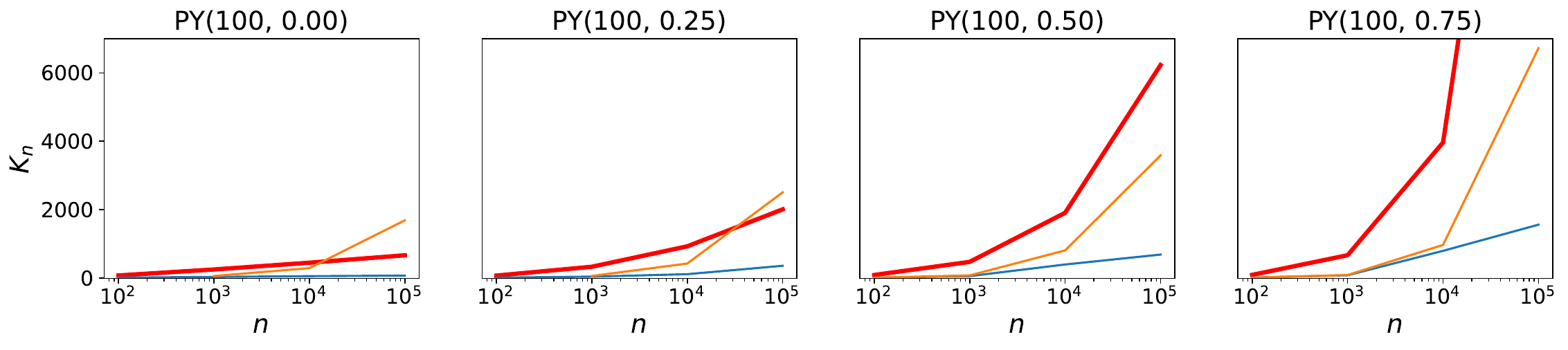}
    \includegraphics[width=\linewidth]{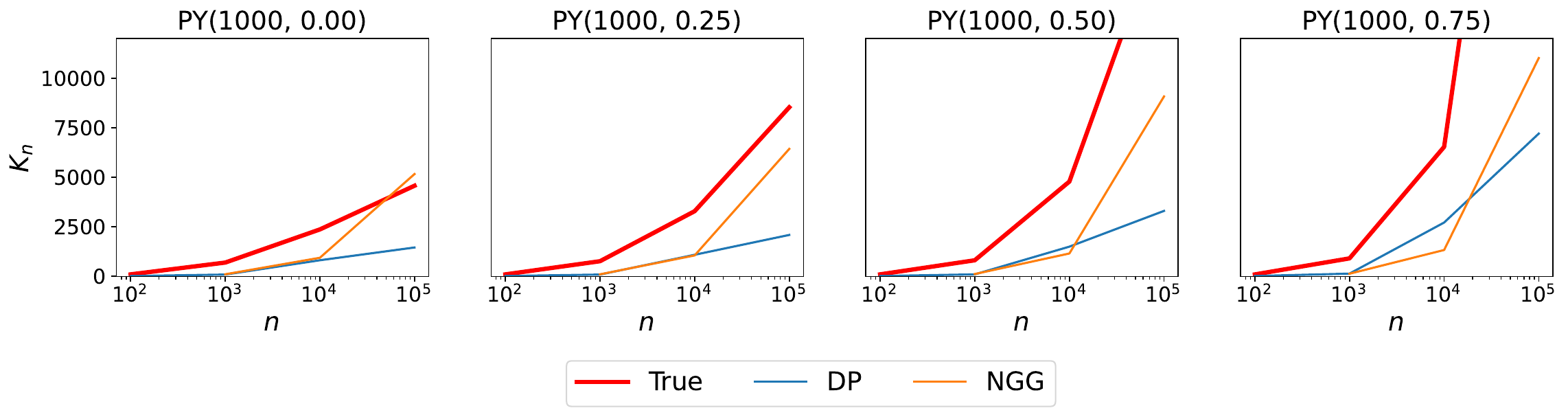}
    \caption{True and estimated cardinalities for the simulation setup in Section~\ref{sec:ex_singlehash}.
    }
    \label{fig:card_singlehash}
\end{figure}

\subsection{The role on $J$ with a single hash function}\label{sec:app_role_J}
We now assess the effect of the parameter $J$, that is the width of the hash function, letting $J = 10, 100, 1000, 10000$ and simulating $n=100,000$ observations either from $\mathrm{PYP}(100, 0.25)$ or from $\mathrm{PYP}(100, 0.75)$. Observe that in the first case, the expected numbers of distinct symbols in the sample is approximately 2200, so that when $J=10000$ we are likely occupying more memory than saving the original dataset in memory.
Figure~\ref{fig:jeffect_freq} shows the MAEs for the frequency recovery problem as a function of $J$, averaged over 50 independent replications. For low frequency tokens, the MAE converges to zero as $J$ increases, while for high frequency tokens (empirical frequency larger than 16) this does not happen even in the case of $J=10,000$. 
Figure~\ref{fig:jeffect_card} shows the absolute error between the true number of distinct symbols and the estimated one, as a function of $J$, averaged over 50 independent replications.
For both the frequency and cardinality estimation problems, the NGGP smoothing  usually provides better estimates, with the notably exception of $J=10,000$ and data generated from a Pitman-Yor process with parameters $(100, 0.75)$, where the DP smoothing has a lower MAE for the frequency estimation problem and absolute error for the cardinality estimation problem.

\begin{figure}[!htb]
    \centering
    \includegraphics[width=\linewidth]{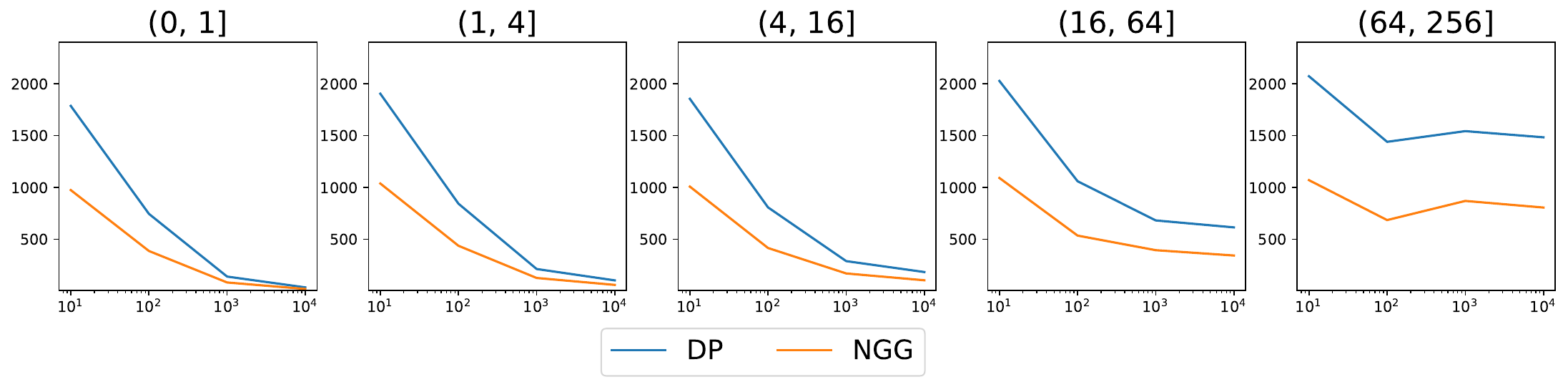}
    \includegraphics[width=\linewidth]{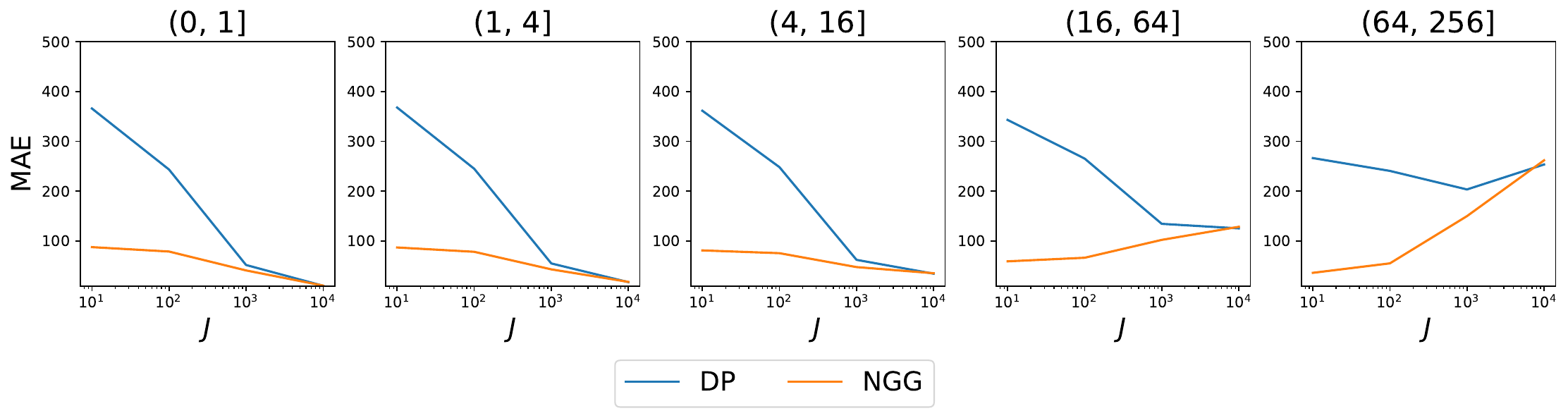}
    \caption{MAEs for the frequency estimators with DP and NGG smoothing. Data generated from a Pitman-Yor process with parameters $(100, 0.25)$ (top row) and $(100, 0.75)$ (bottom row). The width of the hash $J$ varies across the $x$-axis.}
    \label{fig:jeffect_freq}
\end{figure}

\begin{figure}[!htb]
    \centering
    \includegraphics[width=0.45\linewidth]{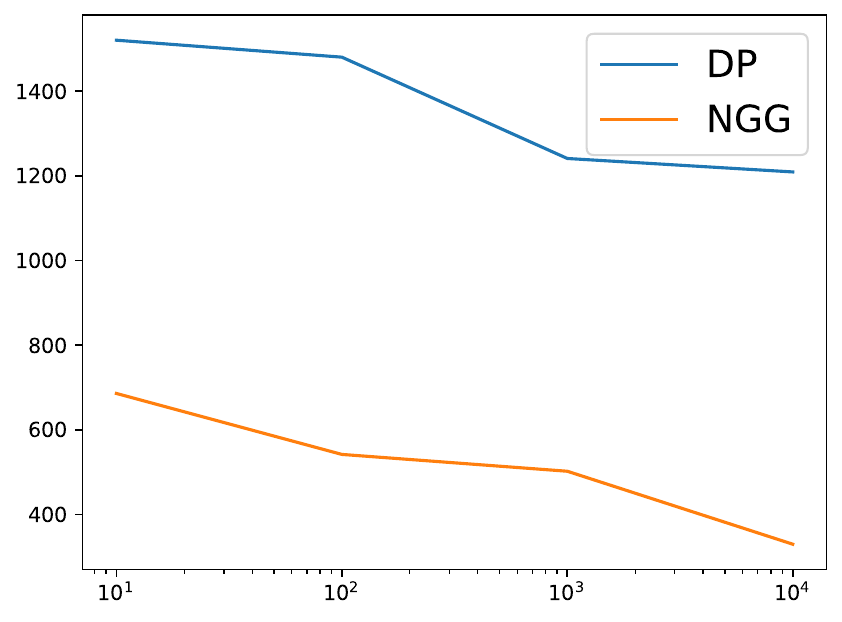}%
    \includegraphics[width=0.45\linewidth]{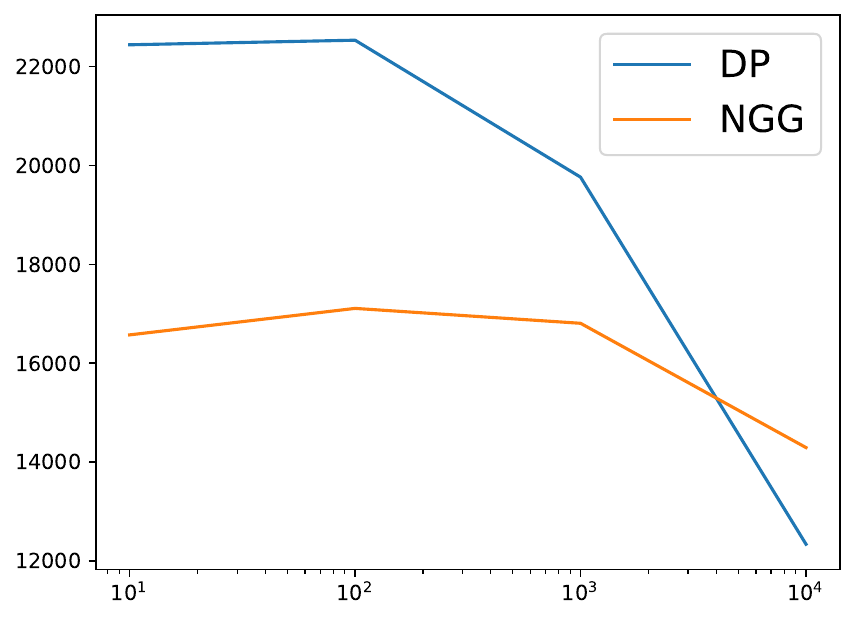}
    \caption{
    Absolute difference between the true number of distinct symbols and estimated cardinality. Data generated from a Pitman-Yor process with parameters $(100, 0.25)$ (left plot) and $(100, 0.75)$ (right plot). The width of the hash $J$ varies across the $x$-axis.}
    \label{fig:jeffect_card}
\end{figure}

\FloatBarrier
\subsection{Further Tables and Plots}

\begin{table}[]
\centering
\begin{tabular}{rrlrrrrrr}
\toprule
  $\theta$ &  $\alpha$ & Model &  (0, 1] &  (1, 4] &  (4, 16] &  (16, 64] &  (64, 256] &  (256, Inf] \\
\midrule
   1 &   0.00 &    DP &   24.85 &    4.46 &     6.95 &      0.32 &      10.60 &      111.91 \\
   1 &   0.00 &   NGG &   17.97 &    3.50 &     4.40 &      8.46 &      35.84 &     4644.27 \\ \hline
   1 &   0.25 &    DP &  103.12 &  141.27 &   400.13 &    853.32 &    1749.00 &     1746.67 \\
   1 &   0.25 &   NGG &   58.87 &   74.15 &   216.87 &    478.36 &     907.83 &     4225.67 \\ \hline
   1 &   0.50 &    DP &  550.62 &  536.02 &   619.38 &    990.39 &    1402.58 &     2724.20 \\
   1 &   0.50 &   NGG &  294.82 &  287.98 &   325.94 &    510.39 &     690.28 &     3418.29 \\ \hline
   1 &   0.75 &    DP &  507.53 &  512.51 &   541.02 &    669.31 &     849.22 &     2111.93 \\
   1 &   0.75 &   NGG &  240.19 &  242.07 &   270.94 &    306.70 &     375.30 &     1631.42 \\ \hline
  10 &   0.00 &    DP &   45.61 &  198.41 &   915.70 &   2828.01 &    4471.81 &     9842.68 \\
  10 &   0.00 &   NGG &   23.00 &  100.83 &   467.43 &   1439.54 &    2231.31 &     5096.32 \\ \hline
  10 &   0.25 &    DP &  497.14 &  882.69 &   902.39 &   1282.83 &    3177.91 &     7977.46 \\
  10 &   0.25 &   NGG &  259.17 &  459.68 &   466.99 &    653.48 &    1607.32 &     3921.58 \\ \hline
  10 &   0.50 &    DP &  599.22 &  626.34 &   724.14 &    994.69 &    1758.74 &     4236.60 \\
  10 &   0.50 &   NGG &  359.08 &  372.31 &   427.47 &    567.20 &    1003.07 &     2444.15 \\ \hline
  10 &   0.75 &    DP &  458.63 &  468.30 &   486.77 &    643.85 &     860.62 &     1431.65 \\
  10 &   0.75 &   NGG &  207.78 &  216.11 &   228.90 &    280.41 &     353.82 &      825.32 \\ \hline
 100 &   0.00 &    DP &  746.32 &  821.35 &  1042.69 &   1479.36 &    2587.94 &     4133.92 \\
 100 &   0.00 &   NGG &  378.60 &  409.30 &   518.76 &    729.82 &    1246.50 &     1807.36 \\ \hline
 100 &   0.25 &    DP &  607.72 &  634.07 &   718.03 &   1043.62 &    1612.92 &     2528.17 \\
 100 &   0.25 &   NGG &  317.57 &  332.77 &   374.05 &    535.96 &     802.53 &     1112.01 \\ \hline
 100 &   0.50 &    DP &  453.52 &  459.22 &   514.10 &    637.53 &     900.43 &     1299.67 \\
 100 &   0.50 &   NGG &  219.62 &  225.15 &   250.25 &    297.07 &     389.18 &      442.91 \\ \hline
 100 &   0.75 &    DP &  243.98 &  248.98 &   253.64 &    285.24 &     305.41 &      285.03 \\
 100 &   0.75 &   NGG &   79.33 &   79.10 &    77.36 &     70.62 &      69.61 &      326.54 \\ \hline
1000 &   0.00 &    DP &  291.86 &  296.77 &   315.70 &    352.96 &     411.05 &      312.08 \\
1000 &   0.00 &   NGG &  100.93 &  101.93 &   104.21 &    106.81 &      95.60 &      138.32 \\ \hline
1000 &   0.25 &    DP &  199.55 &  200.98 &   204.88 &    218.42 &     202.47 &      109.08 \\
1000 &   0.25 &   NGG &   83.90 &   83.46 &    81.73 &     75.22 &      52.85 &      169.91 \\ \hline
1000 &   0.50 &    DP &  110.68 &  110.57 &   108.81 &     96.99 &      59.97 &      131.87 \\
1000 &   0.50 &   NGG &   66.85 &   65.94 &    61.48 &     47.25 &      36.83 &      204.19 \\ \hline
1000 &   0.75 &    DP &   43.18 &   41.91 &    37.18 &     20.56 &      48.94 &      247.70 \\
1000 &   0.75 &   NGG &   53.69 &   52.59 &    48.07 &     30.42 &      40.62 &      201.60 \\ \hline
\bottomrule
\end{tabular}
\caption{Mean Absolute Errors for the frequency recovery simulation setup in Section~\ref{sec:ex_singlehash} for data generated from a PYP.}
    \label{tab:freq_singlehash_py}
\end{table}

\begin{table}[]
\centering
\begin{tabular}{rrlrrrr}
\toprule
  $\theta$ &  $\alpha$ & Model &   $n=100$ &   $n=1,000$ &   $n=10,000$ &   $n=10,0000$ \\
\midrule
   1.00 &   0.00 &    DP &  1.07 &   0.97 &    0.86 &     1.84 \\
   1.00 &   0.00 &   NGG &  0.96 &  28.61 &   80.25 &    43.36 \\ \hline
   1.00 &   0.25 &    DP &  2.47 &  17.01 &   35.29 &    67.68 \\
   1.00 &   0.25 &   NGG &  2.19 &  20.72 &  189.59 &   610.90 \\ \hline
   1.00 &   0.50 &    DP & 12.12 &  43.50 &  160.44 &   506.53 \\
   1.00 &   0.50 &   NGG &  2.79 &  22.26 &   60.13 &   348.61 \\ \hline
   1.00 &   0.75 &    DP & 38.78 & 263.71 & 1433.17 &  9158.79 \\
   1.00 &   0.75 &   NGG & 28.46 & 249.30 & 1204.08 &  6608.67 \\ \hline
  10 &   0.00 &    DP & 13.06 &  33.21 &   56.90 &    75.58 \\
  10 &   0.00 &   NGG & 11.56 &   7.87 &  184.96 &   798.94 \\ \hline
  10 &   0.25 &    DP & 27.71 &  83.23 &  185.06 &   349.92 \\
  10 &   0.25 &   NGG & 28.05 &  71.80 &   89.69 &   908.58 \\ \hline
  10 &   0.50 &    DP & 27.64 & 130.49 &  445.35 &  1392.30 \\
  10 &   0.50 &   NGG & 24.96 & 115.52 &  165.33 &   264.86 \\ \hline
  10 &   0.75 &    DP & 52.26 & 346.10 & 1847.56 & 12044.61 \\
  10 &   0.75 &   NGG & 51.86 & 330.95 & 1645.83 &  8337.21 \\ \hline
 100 &   0.00 &    DP & 58.00 & 215.85 &  408.23 &   623.73 \\
 100 &   0.00 &   NGG & 59.20 & 192.26 &  161.97 &  1280.29 \\ \hline
 100 &   0.25 &    DP & 65.52 & 290.13 &  783.18 &  1564.87 \\
 100 &   0.25 &   NGG & 60.32 & 273.48 &  455.92 &   565.01 \\ \hline
 100 &   0.50 &    DP & 69.44 & 415.64 & 1456.00 &  5638.79 \\
 100 &   0.50 &   NGG & 70.75 & 393.27 & 1105.10 &  1970.70 \\ \hline
 100 &   0.75 &    DP & 79.37 & 589.92 & 3210.12 & 21654.22 \\
 100 &   0.75 &   NGG & 80.19 & 575.14 & 3049.16 & 16436.74 \\ \hline
1000 &   0.00 &    DP & 84.46 & 609.12 & 1642.72 &  3144.14 \\
1000 &   0.00 &   NGG & 87.86 & 588.69 & 1273.71 &  1285.24 \\ \hline
1000 &   0.25 &    DP & 82.48 & 665.54 & 2215.45 &  6546.01 \\
1000 &   0.25 &   NGG & 85.82 & 652.16 & 2137.90 &  1894.59 \\ \hline
1000 &   0.50 &    DP & 84.28 & 726.75 & 2997.70 & 14396.68 \\
1000 &   0.50 &   NGG & 87.34 & 711.04 & 3428.72 &  9739.97 \\ \hline
1000 &   0.75 &    DP & 84.91 & 803.67 & 3969.91 & 32892.45 \\
1000 &   0.75 &   NGG & 88.58 & 790.90 & 5456.84 & 30052.63 \\ \hline
\bottomrule
\end{tabular}
\caption{Mean Absolute Errors for the cardinality recovery simulation setup in Section~\ref{sec:ex_singlehash} for data generated from a PYP.}
    \label{tab:card_singlehash_py}
\end{table}

\begin{table}[]
    \centering
\begin{tabular}{rlrrrr}
\toprule
 Params & Model &   100 &   1000 &   10000 &  100000 \\
\midrule
   1.30 &    DP & 36.14 & 244.64 & 1377.48 & 9203.69 \\
   1.30 &   NGG &   NaN & 236.01 & 1193.71 & 6623.06 \\ \hline
   1.60 &    DP & 16.81 &  90.53 &  404.33 & 1645.03 \\
   1.60 &   NGG &   NaN &  55.57 &  141.33 &  295.76 \\ \hline
   1.90 &    DP &  7.49 &  40.03 &  152.04 &  540.63 \\
   1.90 &   NGG &  4.34 &   5.68 &   87.70 &  455.90 \\ \hline
   2.20 &    DP &  3.34 &  20.99 &   71.63 &  224.75 \\
   2.20 &   NGG &  3.96 &  15.86 &  133.26 &  553.39 \\
\bottomrule
\end{tabular}
    \caption{Mean Absolute Errors for the frequency recovery simulation setup in Section~\ref{sec:ex_singlehash} for data generated from a Zipf distribution.}
    \label{tab:card_singlehash_zipf}
\end{table}

\begin{table}
\centering
\resizebox{\textwidth}{!}{
\begin{tabular}{rrrllrrrrrr}
\toprule
  $\theta$ &  $\alpha$ &  J & Model &  Rule &  (0, 1] &  (1, 4] &  (4, 16] &  (16, 64] &  (64, 256] &  (256, $\infty$] \\
\midrule
  10.00 &   0.25 & 50 &   CMS &   CMS &   45.04 &  101.22 &   375.80 &    980.86 &    3277.48 &    14945.69 \\
  10.00 &   0.25 & 50 & D-CMS & D-CMS &   45.04 &  101.22 &   375.80 &    980.86 &    3277.48 &    14945.69 \\
  10.00 &   0.25 & 50 &    DP &   MIN &   18.67 &   43.76 &   129.42 &    387.30 &    1001.78 &     4944.18 \\
  10.00 &   0.25 & 50 &    DP &   PoE &   41.25 &   97.07 &   343.93 &   1083.98 &    3320.18 &    14995.98 \\
  10.00 &   0.25 & 50 &   NGG &   MIN &    0.88 &    2.37 &    14.11 &     56.02 &     202.01 &     3096.35 \\
  10.00 &   0.25 & 50 &   NGG &   PoE &    1.00 &    2.69 &     8.96 &     34.12 &     131.95 &     3664.00 \\ \midrule
  10.00 &   0.75 & 50 &   CMS &   CMS & 1111.81 & 1128.69 &  1199.58 &   1437.59 &    2185.28 &     5240.75 \\
  10.00 &   0.75 & 50 & D-CMS & D-CMS &  112.42 &  130.65 &   184.77 &    376.77 &    1046.46 &     4196.57 \\
  10.00 &   0.75 & 50 &    DP &   MIN &   92.51 &   91.44 &    88.56 &     81.05 &      62.34 &     1074.53 \\
  10.00 &   0.75 & 50 &    DP &   PoE &  288.11 &  288.47 &   293.61 &    326.79 &     351.62 &      863.49 \\
  10.00 &   0.75 & 50 &   NGG &   MIN &    2.37 &    1.08 &     4.71 &     27.41 &     117.34 &     1364.42 \\
  10.00 &   0.75 & 50 &   NGG &   PoE &    1.00 &    2.53 &     8.23 &     31.20 &     121.97 &     1373.37 \\  \midrule
 100.00 &   0.25 & 50 &   CMS &   CMS &  137.17 &  190.65 &   375.47 &    949.30 &    2218.66 &     5681.97 \\
 100.00 &   0.25 & 50 & D-CMS & D-CMS &   70.92 &  123.32 &   294.83 &    855.65 &    2127.67 &     5597.44 \\
 100.00 &   0.25 & 50 &    DP &   MIN &   42.96 &   47.17 &    64.57 &    113.54 &     170.26 &      572.81 \\
 100.00 &   0.25 & 50 &    DP &   PoE &  117.72 &  151.72 &   293.01 &    847.90 &    2064.75 &     5675.91 \\
 100.00 &   0.25 & 50 &   NGG &   MIN &    4.66 &    3.59 &     4.37 &     21.03 &     114.48 &     1024.91 \\
 100.00 &   0.25 & 50 &   NGG &   PoE &    1.00 &    2.70 &     8.91 &     33.19 &     133.37 &     1047.04 \\  \midrule
 100.00 &   0.75 & 50 &   CMS &   CMS & 2222.70 & 2226.40 &  2274.01 &   2384.39 &    2838.38 &     3886.42 \\
 100.00 &   0.75 & 50 & D-CMS & D-CMS &  184.81 &  191.58 &   218.39 &    324.20 &     668.01 &     1628.11 \\
 100.00 &   0.75 & 50 &    DP &   MIN &   46.13 &   44.77 &    39.63 &     20.12 &      66.97 &      666.53 \\
 100.00 &   0.75 & 50 &    DP &   PoE &   59.71 &   58.40 &    53.39 &     32.62 &      53.26 &      642.00 \\
 100.00 &   0.75 & 50 &   NGG &   MIN &    1.72 &    0.71 &     5.42 &     28.42 &     119.23 &      750.57 \\
 100.00 &   0.75 & 50 &   NGG &   PoE &    0.98 &    2.51 &     8.15 &     31.11 &     120.75 &      751.21 \\  \midrule
1000.00 &   0.25 & 50 &   CMS &   CMS & 2299.76 & 2306.34 &  2345.44 &   2469.56 &    2812.57 &     3336.17 \\
1000.00 &   0.25 & 50 & D-CMS & D-CMS &  307.78 &  314.26 &   336.40 &    420.87 &     657.49 &     1017.79 \\
1000.00 &   0.25 & 50 &    DP &   MIN &   29.56 &   27.96 &    22.06 &     10.86 &      87.93 &      370.95 \\
1000.00 &   0.25 & 50 &    DP &   PoE &   34.40 &   32.82 &    26.96 &     12.19 &      81.95 &      364.12 \\
1000.00 &   0.25 & 50 &   NGG &   MIN &    2.78 &    1.21 &     5.08 &     29.05 &     119.90 &      407.21 \\
1000.00 &   0.25 & 50 &   NGG &   PoE &    1.00 &    2.69 &     8.84 &     33.00 &     124.25 &      412.00 \\  \midrule
1000.00 &   0.75 & 50 &   CMS &   CMS & 3817.12 & 3819.58 &  3829.78 &   3863.39 &    3967.53 &     4134.95 \\
1000.00 &   0.75 & 50 & D-CMS & D-CMS &  181.38 &  183.33 &   187.24 &    202.90 &     253.27 &      330.54 \\
1000.00 &   0.75 & 50 &    DP &   MIN &    4.88 &    3.35 &     2.74 &     24.88 &     108.61 &      378.49 \\
1000.00 &   0.75 & 50 &    DP &   PoE &    5.04 &    3.51 &     2.68 &     24.70 &     108.43 &      378.32 \\
1000.00 &   0.75 & 50 &   NGG &   MIN &    0.72 &    0.82 &     6.47 &     29.03 &     112.81 &      383.24 \\
1000.00 &   0.75 & 50 &   NGG &   PoE &    0.56 &    1.11 &     6.79 &     29.30 &     113.44 &      384.03 \\
\bottomrule
\end{tabular}}
\caption{Mean Absolute Errors for the frequency recovery simulation setup in Section~\ref{sec:ex_multihash} for data generated from a PYP, $J=50$, $M=20$}
    \label{tab:freq_multihash_J50}
\end{table}

\begin{table}
\centering
\resizebox{\textwidth}{!}{
\begin{tabular}{rrrllrrrrrr}
\toprule
  $\theta$ &  $\alpha$ &   J & Model &  Rule &  (0, 1] &  (1, 4] &  (4, 16] &  (16, 64] &  (64, 256] &  (256, Inf] \\
\midrule
  10.00 &   0.25 & 100 &   CMS &   CMS &   45.04 &  101.22 &   375.80 &    980.86 &    3277.48 &    14945.69 \\
  10.00 &   0.25 & 100 & D-CMS & D-CMS &   45.04 &  101.22 &   375.80 &    980.86 &    3277.48 &    14945.69 \\
  10.00 &   0.25 & 100 &    DP &   MIN &   29.61 &   67.93 &   230.57 &    723.42 &    2106.58 &     9701.60 \\
  10.00 &   0.25 & 100 &    DP &   PoE &   41.25 &   97.07 &   343.94 &   1084.00 &    3320.19 &    14995.98 \\
  10.00 &   0.25 & 100 &   NGG &   MIN &    1.65 &    4.26 &    28.21 &    132.02 &     365.16 &     2957.30 \\
  10.00 &   0.25 & 100 &   NGG &   PoE &    0.98 &    2.64 &     8.87 &     34.07 &     131.93 &     3663.73 \\ \midrule
  10.00 &   0.75 & 100 &   CMS &   CMS &  519.01 &  549.33 &   610.85 &    833.56 &    1555.86 &     4662.20 \\
  10.00 &   0.75 & 100 & D-CMS & D-CMS &   69.26 &   85.94 &   143.43 &    346.10 &    1036.09 &     4193.80 \\
  10.00 &   0.75 & 100 &    DP &   MIN &   94.84 &   95.10 &    97.49 &    110.57 &     128.34 &      890.61 \\
  10.00 &   0.75 & 100 &    DP &   PoE &  300.24 &  304.49 &   341.90 &    535.87 &     925.71 &     2605.57 \\
  10.00 &   0.75 & 100 &   NGG &   MIN &    6.31 &    4.91 &     3.06 &     22.41 &     108.97 &     1351.34 \\
  10.00 &   0.75 & 100 &   NGG &   PoE &    0.95 &    2.47 &     8.14 &     31.15 &     121.71 &     1372.77 \\ \midrule
 100.00 &   0.25 & 100 &   CMS &   CMS &   69.51 &  122.83 &   308.23 &    875.87 &    2164.50 &     5625.14 \\
 100.00 &   0.25 & 100 & D-CMS & D-CMS &   45.13 &   96.05 &   277.16 &    838.09 &    2130.47 &     5593.26 \\
 100.00 &   0.25 & 100 &    DP &   MIN &   30.92 &   42.51 &    82.06 &    213.45 &     444.65 &      884.52 \\
 100.00 &   0.25 & 100 &    DP &   PoE &   69.51 &  114.89 &   260.00 &    832.20 &    2084.65 &     5708.11 \\
 100.00 &   0.25 & 100 &   NGG &   MIN &    4.95 &    5.60 &     9.07 &     20.30 &      80.13 &      989.03 \\
 100.00 &   0.25 & 100 &   NGG &   PoE &    1.00 &    2.69 &     8.90 &     33.10 &     132.93 &     1047.06 \\ \midrule
 100.00 &   0.75 & 100 &   CMS &   CMS & 1039.88 & 1050.45 &  1089.99 &   1216.56 &    1650.26 &     2720.45 \\
 100.00 &   0.75 & 100 & D-CMS & D-CMS &  113.12 &  121.60 &   152.18 &    264.04 &     653.78 &     1656.27 \\
 100.00 &   0.75 & 100 &    DP &   MIN &   65.74 &   64.68 &    60.93 &     44.66 &      49.33 &      582.24 \\
 100.00 &   0.75 & 100 &    DP &   PoE &   97.97 &   97.15 &    94.41 &     81.48 &      51.73 &      509.93 \\
 100.00 &   0.75 & 100 &   NGG &   MIN &    4.98 &    3.49 &     2.93 &     24.93 &     114.06 &      740.63 \\
 100.00 &   0.75 & 100 &   NGG &   PoE &    0.65 &    1.80 &     6.90 &     30.13 &     120.91 &      753.23 \\ \midrule
1000.00 &   0.25 & 100 &   CMS &   CMS & 1021.71 & 1033.41 &  1074.24 &   1202.27 &    1555.70 &     2069.55 \\
1000.00 &   0.25 & 100 & D-CMS & D-CMS &  185.76 &  194.24 &   221.06 &    318.56 &     599.81 &     1023.25 \\
1000.00 &   0.25 & 100 &    DP &   MIN &   51.33 &   50.03 &    45.27 &     27.36 &      55.54 &      315.90 \\
1000.00 &   0.25 & 100 &    DP &   PoE &   67.91 &   66.75 &    62.44 &     44.85 &      44.24 &      285.51 \\
1000.00 &   0.25 & 100 &   NGG &   MIN &    7.58 &    5.95 &     3.10 &     23.10 &     111.75 &      396.16 \\
1000.00 &   0.25 & 100 &   NGG &   PoE &    0.51 &    2.12 &     8.28 &     32.41 &     123.41 &      412.46 \\ \midrule
1000.00 &   0.75 & 100 &   CMS &   CMS & 1835.93 & 1839.43 &  1847.09 &   1882.85 &    1987.37 &     2142.51 \\
1000.00 &   0.75 & 100 & D-CMS & D-CMS &  127.82 &  129.18 &   136.15 &    156.19 &     218.52 &      315.23 \\
1000.00 &   0.75 & 100 &    DP &   MIN &   10.34 &    8.83 &     4.05 &     19.23 &     102.13 &      370.61 \\
1000.00 &   0.75 & 100 &    DP &   PoE &   10.89 &    9.37 &     4.42 &     18.66 &     101.46 &      369.95 \\
1000.00 &   0.75 & 100 &   NGG &   MIN &    2.77 &    1.32 &     4.32 &     26.75 &     110.44 &      379.83 \\
1000.00 &   0.75 & 100 &   NGG &   PoE &    1.76 &    0.73 &     5.41 &     27.90 &     111.53 &      380.32 \\
\bottomrule
\end{tabular}}
\caption{Mean Absolute Errors for the frequency recovery simulation setup in Section~\ref{sec:ex_multihash} for data generated from a PYP, $J=100$, $M=10$}
    \label{tab:freq_multihash_J100}
\end{table}

\begin{table}
\centering
\resizebox{\textwidth}{!}{
\begin{tabular}{rrcllrrrrrr}
\toprule
  $\theta$ &  $\alpha$ &   J & Model & Rule &  (0, 1] &  (1, 4] &  (4, 16] &  (16, 64] &  (64, 256] &  (256, Inf] \\ \midrule
10.00 &   0.25 & 500 &   CMS &   CMS &   50.17 &  106.09 &   388.14 &    980.86 &    3277.48 &    14946.15 \\
  10.00 &   0.25 & 500 & D-CMS & D-CMS &   50.17 &  106.09 &   388.14 &    980.86 &    3277.48 &    14946.15 \\
  10.00 &   0.25 & 500 &    DP &   MIN &   42.01 &  115.40 &   336.56 &   1062.52 &    3252.49 &    14706.44 \\
  10.00 &   0.25 & 500 &    DP &   PoE &   42.97 &  117.29 &   343.85 &   1083.06 &    3318.38 &    14995.29 \\
  10.00 &   0.25 & 500 &   NGG &   MIN &   12.56 &   34.93 &    72.10 &    197.38 &     382.21 &     2904.90 \\
  10.00 &   0.25 & 500 &   NGG &   PoE &   18.99 &   41.73 &   121.37 &    397.09 &    1177.96 &     5152.19 \\ \midrule
  10.00 &   0.75 & 500 &   CMS &   CMS &  137.37 &  152.45 &   214.06 &    441.80 &    1163.42 &     4300.87 \\
  10.00 &   0.75 & 500 & D-CMS & D-CMS &   62.34 &   79.28 &   136.09 &    355.28 &    1043.25 &     4198.52 \\
  10.00 &   0.75 & 500 &    DP &   MIN &   73.30 &   80.06 &   104.81 &    217.93 &     552.55 &     1601.72 \\
  10.00 &   0.75 & 500 &    DP &   PoE &   91.74 &  102.16 &   142.72 &    332.15 &     914.40 &     2548.65 \\
  10.00 &   0.75 & 500 &   NGG &   MIN &   27.88 &   28.37 &    28.18 &     31.67 &      60.22 &     1291.65 \\
  10.00 &   0.75 & 500 &   NGG &   PoE &   23.91 &   23.39 &    22.14 &     24.53 &      58.23 &     1297.18 \\ \midrule
 100.00 &   0.25 & 500 &   CMS &   CMS &   73.42 &  120.80 &   300.88 &    866.07 &    2190.34 &     5664.67 \\
 100.00 &   0.25 & 500 & D-CMS & D-CMS &   72.17 &  118.63 &   296.93 &    858.18 &    2180.84 &     5655.48 \\
 100.00 &   0.25 & 500 &    DP &   MIN &   57.82 &   97.25 &   209.62 &    623.53 &    1497.13 &     4094.27 \\
 100.00 &   0.25 & 500 &    DP &   PoE &   69.36 &  119.48 &   263.75 &    796.70 &    1991.72 &     5668.14 \\
 100.00 &   0.25 & 500 &   NGG &   MIN &   17.45 &   30.75 &    65.36 &    150.61 &     231.42 &      656.71 \\
 100.00 &   0.25 & 500 &   NGG &   PoE &   12.61 &   24.48 &    60.86 &    165.16 &     305.86 &      839.24 \\ \midrule
 100.00 &   0.75 & 500 &   CMS &   CMS &  238.09 &  247.23 &   284.89 &    413.02 &     833.72 &     1885.43 \\
 100.00 &   0.75 & 500 & D-CMS & D-CMS &   90.76 &   98.44 &   129.36 &    249.29 &     641.29 &     1645.40 \\
 100.00 &   0.75 & 500 &    DP &   MIN &   86.31 &   88.46 &    96.26 &    117.89 &     200.68 &      431.21 \\
 100.00 &   0.75 & 500 &    DP &   PoE &  105.62 &  109.40 &   121.73 &    158.32 &     292.72 &      564.81 \\
 100.00 &   0.75 & 500 &   NGG &   MIN &   27.46 &   26.62 &    22.96 &     16.65 &      71.01 &      698.03 \\
 100.00 &   0.75 & 500 &   NGG &   PoE &   23.95 &   22.96 &    19.27 &     15.21 &      69.13 &      686.43 \\ \midrule
1000.00 &   0.25 & 500 &   CMS &   CMS &  244.19 &  255.58 &   288.69 &    420.10 &     777.75 &     1304.57 \\
1000.00 &   0.25 & 500 & D-CMS & D-CMS &  139.81 &  151.42 &   179.73 &    290.44 &     598.47 &     1054.99 \\
1000.00 &   0.25 & 500 &    DP &   MIN &   88.68 &   91.06 &    98.22 &    123.28 &     200.41 &      225.23 \\
1000.00 &   0.25 & 500 &    DP &   PoE &  107.09 &  110.62 &   122.70 &    163.23 &     284.70 &      385.02 \\
1000.00 &   0.25 & 500 &   NGG &   MIN &   36.12 &   35.03 &    32.41 &     25.51 &      62.97 &      339.46 \\
1000.00 &   0.25 & 500 &   NGG &   PoE &   33.08 &   32.34 &    28.43 &     21.29 &      60.65 &      336.40 \\ \midrule
1000.00 &   0.75 & 500 &   CMS &   CMS &  383.04 &  385.61 &   395.40 &    429.81 &     534.39 &      697.19 \\
1000.00 &   0.75 & 500 & D-CMS & D-CMS &   91.46 &   93.15 &    99.46 &    123.74 &     194.53 &      308.34 \\
1000.00 &   0.75 & 500 &    DP &   MIN &   36.51 &   35.31 &    30.98 &     18.00 &      58.01 &      284.73 \\
1000.00 &   0.75 & 500 &    DP &   PoE &   39.57 &   38.41 &    34.21 &     20.52 &      54.63 &      276.94 \\
1000.00 &   0.75 & 500 &   NGG &   MIN &   19.84 &   18.61 &    13.55 &     11.70 &      81.73 &      333.80 \\
1000.00 &   0.75 & 500 &   NGG &   PoE &   19.83 &   18.47 &    13.56 &     11.60 &      82.21 &      334.98 \\
\bottomrule
\end{tabular}}
\caption{Mean Absolute Errors for the frequency recovery simulation setup in Section~\ref{sec:ex_multihash} for data generated from a PYP, $J=500$, $M=2$}
    \label{tab:freq_multihash_J500}
\end{table}

\begin{table}
\centering
\resizebox{\textwidth}{!}{
\begin{tabular}{rrcllrrrrrr}
\toprule
   $\theta$ & $\alpha$ &    J & Model & Rule &  (0, 1] &  (1, 4] &  (4, 16] &  (16, 64] &  (64, 256] &  (256, Inf] \\
\midrule
  10.00 &   0.25 & 1000 &   CMS &   CMS &   121.71 &   210.71 &   466.97 &   1079.21 &    3495.76 &    15017.40 \\
  10.00 &   0.25 & 1000 & D-CMS & D-CMS & 59759.36 & 59570.16 & 59066.38 &  58949.93 &   55848.52 &    45342.42 \\
  10.00 &   0.25 & 1000 &    DP &   MIN &    59.11 &   161.54 &   409.40 &   1124.58 &    4036.21 &    14960.60 \\
  10.00 &   0.25 & 1000 &    DP &   PoE &    59.11 &   161.54 &   409.40 &   1124.58 &    4036.21 &    14960.60 \\
  10.00 &   0.25 & 1000 &   NGG &   MIN &    25.89 &    51.16 &   102.85 &    248.94 &     367.17 &     3041.74 \\
  10.00 &   0.25 & 1000 &   NGG &   PoE &    29.20 &    79.56 &   201.12 &    598.89 &    2086.05 &     7815.07 \\ \midrule
  10.00 &   0.75 & 1000 &   CMS &   CMS &   258.63 &   280.17 &   330.02 &    550.53 &    1230.85 &     4429.62 \\
  10.00 &   0.75 & 1000 & D-CMS & D-CMS & 32310.89 & 32287.77 & 32165.91 &  31910.38 &   31223.17 &    27253.27 \\
  10.00 &   0.75 & 1000 &    DP &   MIN &   168.66 &   165.42 &   219.89 &    360.93 &     882.13 &     2259.36 \\
  10.00 &   0.75 & 1000 &    DP &   PoE &   168.66 &   165.42 &   219.89 &    360.93 &     882.13 &     2259.36 \\
  10.00 &   0.75 & 1000 &   NGG &   MIN &    43.88 &    45.46 &    52.77 &     67.10 &      89.61 &     1229.09 \\
  10.00 &   0.75 & 1000 &   NGG &   PoE &    57.51 &    59.91 &    75.62 &    109.00 &     190.46 &     1023.42 \\ \midrule
 100.00 &   0.25 & 1000 &   CMS &   CMS &   271.30 &   324.67 &   478.94 &   1040.59 &    2376.88 &     5792.76 \\
 100.00 &   0.25 & 1000 & D-CMS & D-CMS & 20357.61 & 20233.24 & 20181.52 &  19594.20 &   18156.52 &    13960.61 \\
 100.00 &   0.25 & 1000 &    DP &   MIN &   254.66 &   311.50 &   392.33 &    932.73 &    1995.06 &     5383.13 \\
 100.00 &   0.25 & 1000 &    DP &   PoE &   254.66 &   311.50 &   392.33 &    932.73 &    1995.06 &     5383.13 \\
 100.00 &   0.25 & 1000 &   NGG &   MIN &    49.61 &    62.68 &   105.96 &    218.71 &     291.27 &      636.25 \\
 100.00 &   0.25 & 1000 &   NGG &   PoE &   153.42 &   182.34 &   231.03 &    520.43 &    1139.86 &     2548.97 \\ \midrule
 100.00 &   0.75 & 1000 &   CMS &   CMS &   256.53 &   265.52 &   299.19 &    426.94 &     863.24 &     1934.90 \\
 100.00 &   0.75 & 1000 & D-CMS & D-CMS &  8640.83 &  8629.88 &  8595.19 &   8459.97 &    8001.45 &     6977.91 \\
 100.00 &   0.75 & 1000 &    DP &   MIN &   131.45 &   135.05 &   151.14 &    210.71 &     372.15 &      699.64 \\
 100.00 &   0.75 & 1000 &    DP &   PoE &   131.45 &   135.05 &   151.14 &    210.71 &     372.15 &      699.64 \\
 100.00 &   0.75 & 1000 &   NGG &   MIN &    43.88 &    44.33 &    42.68 &     41.78 &      67.05 &      665.09 \\
 100.00 &   0.75 & 1000 &   NGG &   PoE &    48.03 &    48.57 &    48.45 &     49.38 &      60.60 &      596.49 \\ \midrule
1000.00 &   0.25 & 1000 &   CMS &   CMS &   256.59 &   265.99 &   299.63 &    428.07 &     784.99 &     1316.24 \\
1000.00 &   0.25 & 1000 & D-CMS & D-CMS &  3460.71 &  3452.05 &  3414.99 &   3278.11 &    2915.43 &     2387.37 \\
1000.00 &   0.25 & 1000 &    DP &   MIN &   139.89 &   144.02 &   161.81 &    220.43 &     394.43 &      586.98 \\
1000.00 &   0.25 & 1000 &    DP &   PoE &   139.89 &   144.02 &   161.81 &    220.43 &     394.43 &      586.98 \\
1000.00 &   0.25 & 1000 &   NGG &   MIN &    57.35 &    56.84 &    56.09 &     58.93 &      73.57 &      292.81 \\
1000.00 &   0.25 & 1000 &   NGG &   PoE &    66.98 &    68.09 &    69.05 &     80.35 &      99.14 &      164.05 \\ \midrule
1000.00 &   0.75 & 1000 &   CMS &   CMS &   251.62 &   254.17 &   264.29 &    299.62 &     399.55 &      572.01 \\
1000.00 &   0.75 & 1000 & D-CMS & D-CMS &  1572.53 &  1570.56 &  1562.70 &   1526.27 &    1430.70 &     1209.89 \\
1000.00 &   0.75 & 1000 &    DP &   MIN &    55.26 &    54.67 &    51.81 &     44.84 &      52.95 &      178.27 \\
1000.00 &   0.75 & 1000 &    DP &   PoE &    55.26 &    54.67 &    51.81 &     44.84 &      52.95 &      178.27 \\
1000.00 &   0.75 & 1000 &   NGG &   MIN &    39.16 &    38.57 &    34.60 &     25.54 &      50.92 &      269.60 \\
1000.00 &   0.75 & 1000 &   NGG &   PoE &    39.36 &    38.79 &    34.99 &     26.26 &      50.79 &      256.25 \\
\bottomrule
\end{tabular}}
\caption{Mean Absolute Errors for the frequency recovery simulation setup in Section~\ref{sec:ex_multihash} for data generated from a PYP, $J=1000$, $M=1$}
    \label{tab:freq_multihash_J1000}
\end{table}

%%% Local Variables:
%%% mode: latex
%%% TeX-master: "supp_jasa"
%%% End:

\FloatBarrier

\end{document}